\documentclass[10pt,journal,compsoc]{IEEEtran}

\ifCLASSOPTIONcompsoc
\usepackage[nocompress]{cite}
\else
\usepackage{cite}
\fi

\usepackage{graphicx,subfigure}

\usepackage{subfigure}
\usepackage{amsmath}
\usepackage{graphicx} 
\usepackage{caption}
\usepackage{graphicx,subfigure}
\usepackage{multicol} 
\usepackage{graphicx} 
\usepackage{amssymb}
\usepackage{float} 
\usepackage{wrapfig} 
\usepackage{booktabs}
\usepackage{lipsum} 
\usepackage{indentfirst}
\usepackage{bm}
\usepackage{color}

\usepackage{amsmath}
\usepackage{enumerate}
\usepackage{extarrows}
\usepackage{url}
\usepackage{graphicx}
\usepackage{epstopdf}
\usepackage{multirow}
\usepackage{tablefootnote}
\usepackage[flushleft]{threeparttable}
\usepackage{slashbox}
\usepackage{array}
\usepackage{multirow}
\usepackage{amsmath}
\usepackage{balance}

\usepackage{enumitem}

\usepackage{amsmath}
\usepackage{amsthm}

\usepackage{hyperref}
\hypersetup{hidelinks}

\usepackage[ruled,linesnumbered,vlined]{algorithm2e}

\usepackage{url}

\usepackage{cite}

\def\dcap{Q}
\def\pcap{\Pi}
\def\adfee{\pi}

\def\enCPT{E}
\def\encpt{e}

\def\val{\theta}
\def\cost{\beta}

\def\data{d}
\def\cycle{c}
\def\indata{r}

\newtheorem{problem}{\textbf{Problem}}
\newtheorem{question}{\textbf{Question}}

\newtheorem{lemma}{\textbf{Lemma}}
\newtheorem{corollary}{\textbf{Corollary}}
\newtheorem{proposition}{\textbf{Proposition}}
\newtheorem{definition}{\textbf{Definition}}
\newtheorem{theorem}{\textbf{Theorem}}

\begin{document}

\title{Monetizing Edge   Service in Mobile Internet Ecosystem}

\author{Zhiyuan~Wang,
        Lin~Gao,
        Tong~Wang,
        and Jingjing~Luo

	\IEEEcompsocitemizethanks{
		\IEEEcompsocthanksitem Zhiyuan Wang is with the School of Electronics and Information Engineering, Harbin Institute of Technology, Shenzhen, China, and the Department of Computer Science and Engineering, The Chinese University of Hong Kong, Shatin, N.T., Hong Kong, China.
        Lin Gao is with the School of Electronics and Information Engineering, Harbin Institute of Technology, Shenzhen, China, and the Shenzhen Institute of Artificial Intelligence and Robotics for Society, Shenzhen, China.
      Tong Wang and Jingjing Luo are with the School of Electronics and Information Engineering, Harbin Institute of Technology, Shenzhen, China.
 		E-mail: tongwang@hit.edu.cn.
      (The first two authors,  Zhiyuan Wang and Lin Gao, contributed equally to this work. Corresponding Author: Tong Wang)~~~~~~
}

}


\IEEEtitleabstractindextext{
\begin{abstract}
In mobile Internet ecosystem, Mobile Users (MUs) purchase wireless data services from Internet Service Provider (ISP) to access to Internet and acquire the interested content services (e.g., online game) from Content Provider (CP).
The popularity of intelligent functions (e.g., AI and 3D modeling) increases the computation-intensity of the content services, leading to a growing computation pressure for the MUs' resource-limited devices.
To this end,  {\emph{edge computing service}} is emerging as a promising approach to alleviate the MUs' computation pressure while keeping their quality-of-service, via offloading some computation tasks of MUs to edge (computing) servers deployed at the local network edge.
Thus, Edge Service Provider (ESP), who deploys the edge servers and offers the edge computing service, becomes an upcoming new stakeholder in the ecosystem.
	In this work, we study the economic interactions of MUs, ISP, CP, and ESP in the new ecosystem with edge computing service, where MUs can acquire the computation-intensive content services (offered by CP) and offload some computation tasks, together with the necessary raw input data, to edge servers (deployed by ESP) through ISP.
We first study the MU's Joint Content Acquisition and Task Offloading (J-CATO) problem, which aims to maximize his long-term payoff. We derive the \emph{off-line} solution with crucial insights, based on which we design an \emph{online} strategy with provable performance.
Then, we study the ESP's edge service monetization problem.
We propose a pricing policy that can achieve a \emph{constant fraction} of the ex post optimal revenue with an extra \emph{constant loss} for the ESP.
Numerical results show that the edge computing service can stimulate the MUs' content acquisition and improve the payoffs of MUs, ISP, and CP.
\end{abstract}

}

\maketitle

\IEEEdisplaynontitleabstractindextext

\IEEEpeerreviewmaketitle

%
\IEEEpeerreviewmaketitle

\section{Introduction}
\subsection{Background and Motivation}
\IEEEPARstart{M}{obile} Internet has been increasingly indispensable for Mobile Users (MUs) in the past decades.
Each MU typically signs a long-term contract with the Internet Service Provider (ISP) to obtain the wireless data service.
The contract offered by ISP usually corresponds to a monthly data plan, consisting of a monthly data cap, a lump-sum subscription fee, and a per-unit fee for exceeding the data cap \cite{sen2013survey}.
Accordingly, MUs, with the wireless data service, can acquire and enjoy various Internet content services (e.g., online game and video streaming) through the mobile applications of Content Providers (CPs) on their mobile devices \cite{dhamdhere2011twelve}.
There are two major trends in the mobile Internet ecosystem during the past several years.
\begin{itemize}
	\item The content service (offered by CPs) has been increasingly \textit{data-hungry} due to the popularity of the high-resolution videos, cloud-based services, and various social media.

	\item The content service (offered by CPs) has been increasingly \textit{computation-intensive} due to the intelligent functions (e.g., virtual reality, mobile games, and 3D modeling) within the CP's mobile application.
\end{itemize}

Regarding the increasing data volume, previous studies (e.g., \cite{zheng2017customized,wang2019multi}) have shown that ISP can alleviate the growing mobile Internet data traffic through more innovative wireless data services.
The study on the increasing computation volume, however, is still at the early stage.
On   one hand, the intelligent content service can help CPs attract more MUs.
On the other hand, the growing computation volume may degrade the MU's Quality of Experience (QoE), as the mobile devices are usually resource-limited.
 {\textbf{Mobile edge computing}}, allowing MUs to offload some computing tasks to the edge servers, is becoming the potential solution to the growing computation volume \cite{hu2015mobile}.
There have been some initial trials carried out by different third-party Edge Service Providers (ESPs).
For example, Vapor IO has opened two edge server sites in Chicago.
EdgeMicro has built a fully functional edge server in Englewood.

The edge service furnishes MUs with both opportunities and challenges.
The crucial part is the trade-off between the local execution and the edge execution.
Specifically, the choice of edge-execution helps MUs reduce the local-execution cost, but edge-execution is not free of charge.
First, the edge-execution of computation tasks requires that the MU should offload the necessary raw input data (e.g., the images in AR functions) to the edge servers, which potentially increases the wireless data usage.
Second, the self-interest ESP also wants to monetize the edge service.
In general, the above trade-off will affect the QoE of MUs, leading to different content acquisition behaviors.
Therefore, this motivates us to study the following key question:
\begin{question}
	What is the MU's optimal content acquisition and task offloading strategy?
\end{question}

As an upcoming stakeholder in the ecosystem, ESP (e.g., Vapor IO and EdgeMicro in US market) is self-interest and seek for more economic benefit from operating the edge servers.
However, comparing to the wireless data service (offered by ISP), the MUs' demand on the edge service is more random and unpredictable.
On the one hand, the local-execution capabilities are heterogeneous across the MU population, depending on the computation resource (e.g., the CPU frequency) of their mobile devices.
On the other hand, the MU's demand on the edge server is usually time-variant, affected by their acquired content service and the operating state of the mobile device (e.g., the battery volume).
All these issues will significantly affect how heavily the MU relies on the edge servers in practice.
These observations motivate us to investigate the following key question:
\begin{question}
	How should ESP monetize the edge service?
\end{question}

Besides the aforementioned strategic behaviors of MUs and ESP, it is crucial to unveil the economic effect of the up-coming edge service on the classic mobile Internet value chain with MUs (as the consumers) as well as ISP and CP (as the providers).
First of all, the edge service may imperceptible increase the MUs' wireless data usage, as the edge-execution relies on offloading the necessary raw input data.
Hence the ISP offering the wireless data service is possible to benefit from the edge service of ESP.
Furthermore, the edge service provides a new solution for the MUs to acquire the computation-intensive content service, which potentially increases the content acquisitions of MUs.
All the above conjectures highly depend on the strategic interplay between the MUs and ESP.
This motivates us to study the third key question in this work:
\begin{question}
	How will ESP's edge service monetization affect the ISP, CP, MUs, and the social welfare of the ecosystem?
\end{question}


This paper studies the new mobile Internet ecosystem with edge computing service.
We aim to demonstrate the economic effect of the edge service and stimulate the edge service monetization.

\subsection{Main Results and Key Contributions}

We investigate the mobile Internet ecosystem consisting of MUs, CPs, ISP, and ESP.
The MUs acquire and enjoy the CPs'   content services (e.g., online game) through the wireless data service offered by ISP and the edge computing service offered by ESP.
We take into account a multi-period operation horizon.
Each MU will make the joint content acquisition and task offloading decisions in each time slot (e.g., every day) with the purpose of monthly payoff maximization.
Hence the MU's Joint Content Acquisition and Task Offloading (J-CATO) problem is an online payoff maximization.
Moreover, ESP monetizes edge computing service though an appropriate pricing policy.

The main results and key contributions of this paper are summarized as follows:
\begin{itemize}
	\item \textit{A Business Model Study on Mobile Internet Ecosystem:}
		We study the economic interactions in mobile Internet ecosystem consisting of MUs, CPs, ISP, and ESP.
		Specifically, we aim to unveil the economic impact of the edge computing service.
		Our study is an initial step towards understanding a more complex business model.
	\item \textit{A Joint Analysis of the MU's Content Acquisition and Task Offloading:}
		We study the MU's Joint Content Acquisition and Task Offloading (J-CATO) problem from the simplified off-line version to the practical online context.
		We first solve the off-line J-CATO (which is non-convex) in closed-form through appropriate reformulations.
		Our analysis unveils the significant role of the shadow price of the wireless data usage.
		Moreover, we propose an online strategy with provable performance based on the intuition behind the shadow price.
	\item \textit{ESP's Edge Service Monetization:}
		We design a pricing policy for ESP to monetize the edge computing service without relying on any statistical information of the MU population (which is costly to measure).
		The key idea of the pricing policy is to iteratively explore and exploit good pricing choices.
		Moreover, we show that our pricing policy can achieve at least a constant fraction of the \textit{ex post optimal revenue} with an extra constant loss.
		By appropriately tuning the parameters, it can achieve a constant competitive ratio under mild conditions.
	\item \textit{Performance Evaluation and Insights:}
		We carry out extensive evaluations on the mobile Internet ecosystem with ESP monetizing edge service.
		We find that the edge service not only helps the MUs overcome the local-computing bottleneck, but also stimulates the content acquisitions of MUs.
		Meanwhile, both ISP and CP also benefit from the increasing content acquisitions.
		Therefore, the edge service leads to higher social benefit for the mobile Internet ecosystem.
\end{itemize}

The rest of the paper is as follows:
Section \ref{Section: Literature Review} reviews related literatures.
Section \ref{Section: System Model} introduces the system model.
Section \ref{Section: Off-line Analysis and Insights} presents the MU's off-line solution together with key insights.
Section \ref{Section: MU On-Line JCOP} studies the MU's online strategy.
Section \ref{Section: ESP Pricing Problem} investigates ESP's pricing policy.
Section \ref{Section: Numerical} presents the numerical results.
We conclude this paper in Section \ref{Section: Conclusion}.

\section{Literature Review}\label{Section: Literature Review}
This paper is related to two streams of studies, i.e., Internet ecosystem and edge computing.
In the following, we review the two streams of studies, respectively.
\subsection{Internet ecosystem}
The Internet ecosystem has been widely studied before (see, e.g., \cite{dhamdhere2011twelve,sen2013survey} for two comprehensive surveys).
The early studies on the Internet ecosystem mainly focused on the Internet data service offered by ISP.
The major research problems include pricing under the peering and transit relations (e.g., \cite{shakkottai2006economics}), the network neutrality and regulations (e.g., \cite{ma2013public}), and the revenue sharing mechanisms (e.g., \cite{he2006pricing,ma2010internet}).
Some follow-up research works took into account the economic interactions between ISPs and users under different business models of the Internet data service.
For example, Hande \textit{et al.} in \cite{hande2010pricing} investigated how the ISP sells the broadband Internet access to users under the flat-rate and the usage-based schemes.
Ma in \cite{ma2016usage} studied the congestion-prone market and how users' congestion sensitivity affect the optimal price and ISPs' competition.
However, the above studies merely took into account the one-period static setting, neglecting the multi-period dynamics.
Some other studies (e.g., \cite{zheng2018optimizing,wang2019economic}) explored the dynamic game-theoretic interactions between ISPs and users.
There were also some studies taking into account both the Internet data service (offered by ISPs) and the Internet content service (offered by CPs).
For example, Wu \textit{et al.} \cite{wu2011revenue} studied the revenue sharing and rate allocation problems between the content ISP and the eyeball ISP (who offer Internet data service to CPs and users, respectively).
Wong \textit{et al.} in \cite{joe2018sponsoring} studied how the CPs subsidize the users' cost on the Internet data services and showed that multiple stakeholders can benefit.

\subsection{Edge Computing}
There are many excellent studies on edge computing from the perspective of communication (e.g., \cite{mao2017survey}) and edge intelligence (e.g., \cite{8736011}).
Next we review some typical literatures among the most recent ones.

Many studies on edge computing focused on the energy-efficient offloading (e.g., \cite{you2016energy,sun2017emm,zhou2019energy}), joint communication and computation resource allocation (e.g., \cite{mao2017stochastic,yang2019joint,wang2016mobile}), wireless-powered system (e.g., \cite{mao2016dynamic,wang2017joint,xu2017online}), and edge caching (e.g., \cite{poularakis2020service,xu2018joint,tang2018enabling}).
For example, You \textit{et al} in \cite{you2016energy} studied the resource allocation for a multi-user MEC system
under time-division multiple access (TDMA) and orthogonal frequency-division multiple access (OFDMA), aiming to minimize the weighted sum of mobile energy consumption under the constraint on computation latency.
Mao \textit{et al} in \cite{mao2017stochastic} developed an online joint radio and computational resource management algorithm.
They leveraged the Lyapunove optimization method to minimize the long-term energy consumption and keep the task buffer stability.
Wang \textit{et al} in \cite{wang2017joint} considered a wireless powered multiuser MEC system, where a multi-antenna access point (AP) broadcasts wireless power to charge users and each user relies on the harvested energy to execute computation tasks.
Poularakis \textit{et al} in \cite{poularakis2020service} studied
the joint optimization of service placement and request routing in
dense MEC networks with multidimensional constraints.
They proposed an algorithm that achieves close-to-optimal performance using a randomized rounding technique.

The economic aspect of the edge service was overlooked.
There are only few studies on the business aspect of the edge service.
Specifically, Chen \textit{et al.} in \cite{chen2015efficient} investigated the multiple users' task offloading game and derived the Nash equilibrium.
Liu \textit{et al.} in \cite{liu2017price} studied how the ESP sets the price for the finite edge computation resource to maximize its revenue.
Xiong \textit{et al.} in \cite{xiong2019joint} jointly considered the interplay between the CPs' sponsoring and the ESP's edge caching services as a hierarchical three-stage Stackelberg game.
Nevertheless, the above studies did not characterize the users' content consumption behavior and neglected the multi-period dynamics.

This paper differs from the above studies in terms of both problem setup and the theoretical solution.
First, we focus on the economic interaction between MUs and ESP, and unveil the win-win impact of edge computing service.
Second, our proposed online MU policy addresses the non-separable payoff, which is different from the Lyapunov framework (as in \cite{mao2016dynamic,mao2017stochastic}).
Third, we also propose a dynamic pricing policy for ESP, which continuously explores and exploits good pricing outcome with provable discretization error.

\section{System Model}\label{Section: System Model}
We consider the mobile Internet ecosystem with a set $\mathcal{N}=\{1,2,...,N\}$ of Mobile Users (MUs), Content Providers (CPs), Internet Service Provider (ISP), and Edge Service Provider (ESP).
More specifically, each MU $n\in\mathcal{N}$ acquires and enjoys the content service of CPs (e.g., \textit{Tencent}, \textit{Facebook}, \textit{Pokemon Go}, etc) on the corresponding mobile applications.
Successful content service acquisition for each MU corresponds to the wireless content delivery (e.g., video streaming) and the computation task execution (e.g., image processing), which highly rely on ISP's wireless data service and ESP's edge computing service, respectively.
\begin{itemize}
	\item \textit{Wireless Content Delivery:}
		The MU $n\in\mathcal{N}$ can obtain the wireless data service from the ISP based on the monthly data plans.
	\item \textit{Computation Task Execution:}
		The MU $n\in\mathcal{N}$ can fulfill the computation tasks either locally at the mobile device or remotely utilizing the edge servers of ESP.
\end{itemize}

We will consider a one-month operation period, consisting of a set $\mathcal{T}=\{1,2,...,T\}$ of time slots.
Each time slot $t\in\mathcal{T}$ may correspond to one day or one hour.
Our analysis in this paper still holds when we consider multiple months.
Next we start with the service model for the mobile Internet ecosystem in Section \ref{Subsection: Service Model}.
We then characterize the MUs and service providers in Section \ref{Subsection: MU Model} and Section \ref{Subsection: Revenues of Providers}, respectively.
Table \ref{Table: Key notations} summarizes the key notations in this paper.

\subsection{Service Models}\label{Subsection: Service Model}
There are three types of services in the mobile Internet ecosystem, i.e., the Internet content service, the wireless data service, and the edge computing service.
Next we introduce the service models.

\subsubsection{Wireless Data Service}
Internet Service Provider (ISP) offers wireless data service based on the monthly data plan, which is a three-part tariff denoted by $\{\dcap,\pcap,\adfee\}$.
Specifically, the MU pays a monthly subscription fee $\pcap$ for the data usage up to the data cap $\dcap$.
And the MU pays the overage fee $\adfee$ for unit data usage exceeding the data cap.
Note that the monthly data cap $\dcap$ and the monthly subscription $\pcap$ of different MUs may be different, but the overage fee $\adfee$ is usually the same for the same ISP \cite{sen2013survey}.

\subsubsection{Edge Computing Service}
Edge Service Provider (ESP) monetizes edge computing service by allowing MUs to offload their computation tasks to the nearby edge servers.
We suppose that ESP charges the MUs based on the offloaded computation volume in a dynamic usage-based manner
That is, ESP can dynamically determine the price of unit computation volume (measured in CPU cycles) depending on the cost and the capacity.
Hence we let $p_{t}$ denote the unit price in slot $t$.
Accordingly, $\bm{p}=(p_{t}:t\in\mathcal{T})$ is the price vector determined by ESP.

\subsubsection{Internet Content Service}\label{Subsubsection: Internet Content Service}
We characterize the content service (offered by CPs) based on a random vector $(\data,\indata,\cycle)$, which jointly captures the per-slot requirement on both communication and computation.
The detailed elaborations are as follows:
\begin{itemize}
	\item The random variable $\data$, defined on the support $[0,\bar{\data}]$ indicates the \textit{data-usage level} of the content service.
	Specifically, $\data$ represents the total wireless data usage (including down-link and up-link) of acquiring CPs' Internet content for an entire time slot.
	
	\item The random variables $\indata$ and $\cycle$ jointly characterize the computation requirement of acquiring content service for one time slot.
	Specifically, the random variable $\indata$ with the support $[0,\bar{\indata}]$ represents the one-slot \textit{raw data amount} (e.g., raw images).
	The random variable $\cycle$ with the support $[0,\bar{\cycle}]$ represents the one-slot \textit{computing amount} (e.g., motion detection) measured in CPU cycles.
\end{itemize}

Based on the above content service model, acquiring the content service for $x$ fraction of time slot will correspond to the content delivery $x\data$ (in bit) and the computation task $(x\indata,x\cycle)$.
Specifically, $x\indata$ (in bit) and $x\cycle$ (in CPU cycles) represent the input raw data amount and computing amount, respectively.
Moreover, the computation task $(x\indata,x\cycle)$ can be executed at the mobile devices or at the edge servers (of ESP), which will be introduced later.

\subsection{MU Model}\label{Subsection: MU Model}
Next we introduce the MU model.
Specifically, we start with the demand realization, MU characteristics, and the MU's decision in Sections \ref{Subsection: MU Demand Realization}, \ref{Subsection: MU Characteristics}, and \ref{Subsection: MU Decisions}, respectively.
We then formulate the MU's monthly payoff in Section \ref{Subsection: MU Payoff}.

\subsubsection{MU Demand Realization}\label{Subsection: MU Demand Realization}
Based on the content service model in Section \ref{Subsubsection: Internet Content Service}, we let $(\bm{\data}_{n},\bm{\indata}_{n},\bm{\cycle}_{n})$ denote the content service realization of MU $n\in\mathcal{N}$.
The vector $\bm{\data}_{n}=\{\data_{n,t}:t\in\mathcal{T}\}$ is the data-usage realization, the vector $\bm{\indata}_{n}=\{\indata_{n,t}:t\in\mathcal{T}\}$ is the input raw data realization, and the vector $\bm{\cycle}_{n}=\{\cycle_{n,t}:t\in\mathcal{T}\}$ is the computing amount realization.

\subsubsection{MU Characteristics}\label{Subsection: MU Characteristics}
We characterize each MU $n\in\mathcal{N}$ taking into account his \textit{satisfaction} and \textit{dissatisfaction} from the content services.

First, the MUs get satisfaction (or happiness) from enjoying the content services.
	We let $U_{n,t}(x)$ denote MU $n$'s experienced satisfaction of acquiring the content service for $x$ fraction of the $t$-th slot.
	The utility function $U_{n,t}(\cdot)$ is both user-dependent and time-dependent, capturing the heterogeneous MU population and the time-variant preference, respectively.
	We suppose that $U_{n,t}(\cdot)$ takes the form of
	\begin{equation}\label{Equ: satisfaction}
		U_{n,t}(x) \triangleq \val_{n,t}\cdot u_{n,t}(x),
	\end{equation}
	where $\val_{n,t}$ is a scalar and represents MU $n$'s valuation (on the content service) in slot $t$.
	Moreover, $u_{n,t}(\cdot)$ is increasing and concave.
	We refer to $u_{n,t}(\cdot)$ as the normalized utility function of MU $n$ in slot $t$.
	
Second, the MUs also get dissatisfaction (or unhappiness) from the content service due to the resource-limited mobile devices.
We let $\enCPT_{n,t}(s)$ denote MU $n$'s experienced dissatisfaction for locally executing the computation tasks of amount $s$ (in CPU cycles) in slot $t$.\footnote{The dissatisfaction in MU's payoff captures the cost of executing the tasks locally, thus it is positively related to the computation amount.	Overall, it can capture the energy consumption in an indirect way compared to the previous study on edge computing (e.g., \cite{mao2017stochastic,yang2019joint}).}
It models the computation-intensive functions (e.g., image processing of AR applications) against with the resource-limited mobile devices.
We suppose that $\enCPT_{n,t}(\cdot)$ takes the form of
	\begin{equation}\label{Equ: dissatisfaction}
		\enCPT_{n,t}(s)\triangleq \cost_{n,t}\cdot e_{n,t}(s),
	\end{equation}
	where $\cost_{n,t}$ is a scalar and measures the sensitivity of MU $n$ in slot $t$.
	Moreover, $e_{n,t}(\cdot)$ is assumed to be increasing and convex, capturing the limited computation capacity.
	We refer to $e_{n,t}(\cdot)$ as the normalized cost function of MU $n$ in slot $t$.

Note that the MU's satisfaction and dissatisfaction depend on his decisions, which will be introduced next.

\subsubsection{MU Decisions}\label{Subsection: MU Decisions}
Each MU $n\in\mathcal{N}$ has two sets of decisions, i.e., the \textit{content-acquiring decisions} $\bm{x}_{n}$ and \textit{task-offloading decisions} $\bm{y}_{n}$.
\begin{itemize}
	\item We let $x_{n,t}\in[0,1]$ denote the content acquiring decision of MU $n$ in slot $t$.
	Specifically, $x_{n,t}$ represents the period length (i.e., fraction of slot) that MU $n$ spends on the content service in slot $t$.
	That is, the content acquiring decision $x_{n,t}$ leads to content delivery amount $x_{n,t}\data_{n,t}$, input raw data amount $x_{n,t}\indata_{n,t}$, and computing amount $x_{n,t}\cycle_{n,t}$.
	Accordingly, we denote $\bm{x}_{n}=\left( x_{n,t}\in[0,1]:t\in\mathcal{T} \right)$ as the content acquiring decisions of MU $n$.

	\item We let $y_{n,t}\in[0,1]$ denote the task-offloading decision of MU $n$ in slot $t$.
	Specifically, $y_{n,t}$ represents the fraction of computation task to be executed remotely at the nearby edge servers.
	That is, the MU tends to execute the computation task $x_{n,t}\cycle_{n,t}y_{n,t}$ at the edge servers by offloading the raw data $x_{n,t}\indata_{n,t}y_{n,t}$.
	Accordingly, we denote $\bm{y}_{n}=\left( y_{n,t}\in[0,1]:t\in\mathcal{T} \right)$ as the task-offloading decisions of MU $n$.
\end{itemize}

MU's content-acquiring decisions $\bm{x}$ and task-offloading decisions $\bm{y}$ will affect his wireless data usage (regarding ISP) and the edge server usage (regarding ESP), which eventually determine his monthly payoff.

\subsubsection{MU Payoff}\label{Subsection: MU Payoff}
Now we derive the MU $n$'s payoff based on the decisions $(\bm{x}_{n},\bm{y}_{n})$ and the content service realization $(\bm{\data}_{n},\bm{\indata}_{n},\bm{\cycle}_{n})$.
Overall, MU's payoff is defined as the difference between the utility and the total cost.
Moreover, the total cost consists of the sunk cost and the opportunistic cost.

\textit{\textbf{Utility:}}
The MU's utility corresponds to his satisfaction from the content service.
Hence the content-acquiring decision $\bm{x}_{n}$ will generate the monthly utility $\sum_{t=1}^{T}U_{n,t}(x_{n,t})$ for MU $n$.
Moreover, the content-acquiring decision $x_{n,t}$ (in slot $t$) also leads to the wireless data usage $\data_{n,t} x_{n,t}$ and the computing task amount $\cycle_{n,t} x_{n,t}$, both of which will incur cost for MU $n$ in slot $t$.

\textit{\textbf{Sunk Cost:}}
MU's sunk cost comes from executing the computation task $\cycle_{n,t} x_{n,t}$.
It includes the \textit{local-execution cost} and the \textit{edge-execution cost}.
\begin{itemize}
	\item The \textit{local-execution cost} corresponds to the dissatisfaction, defined in (\ref{Equ: dissatisfaction}), due to the resource-limited mobile device.
	Recall that the computation task of volume $\cycle_{n,t} x_{n,t}(1-y_{n,t})$ will be executed locally, thus incurs dissatisfaction for MU $n$ in slot $t$.
	
	\item The \textit{edge-execution cost} is the monetary payment (to ESP) for utilizing the edge servers.
	Mathematically, the payment is proportional to the volume of offloaded computation task based on the unit price $p_{t}$.
	Recall that the computation task of volume $\cycle_{n,t} x_{n,t}y_{n,t}$ will be executed at the edge servers, thus incurs monetary payment for MU $n$ in slot $t$.
\end{itemize}

Therefore, the decision $(x_{n,t},y_{n,t})$ in slot $t$ leads to the local-execution cost $\enCPT_{n,t}\left(\cycle_{n,t} x_{n,t}(1-y_{n,t})\right)$ and the edge-execution cost $p_{t}\cycle_{n,t} x_{n,t}y_{n,t}$ for MU $n$.
For notation simplicity, we define the \textit{virtual payoff} of MU $n$ in slot $t$ as follows:
\begin{equation}\label{Equ: virtual payoff}
	\begin{aligned}
		& f_{n,t}(x_{n,t},y_{n,t})
		\triangleq\\
		&  U_{n,t}(x_{n,t})
		- \enCPT_{n,t}\big(\cycle_{n,t} x_{n,t}(1-y_{n,t})\big)
		-  p_{t} \cycle_{n,t} x_{n,t}y_{n,t}.
	\end{aligned}
\end{equation}
Note that the three terms in the virtual payoff all correspond to the monetary meansurement.

\textit{\textbf{Opportunistic Cost:}}
The opportunistic cost is the MU's monetary payment (to ISP) for the wireless data usage exceeding the monthly data cap.
The MU's total data usage consists of \textit{content delivery} and \textit{raw data migration}.
\begin{itemize}
	\item The data usage of \textit{content delivery} merely depends on the MU's content-acquiring decision.
	Specifically, the content-acquiring decision $x_{n,t}$ in slot $t$ leads to the wireless data usage $\data_{n,t} x_{n,t}$.
	
	\item The data usage of \textit{raw data migration} depends on the acquiring and offloading choices.
	Specifically, the decisions $(x_{n,t},y_{n,t})$ lead to the raw data of volume $\indata_{n,t} x_{n,t}y_{n,t}$ to be migrated in slot $t$.
\end{itemize}

Therefore, the MU's wireless data usage under the decision $(x_{n,t},y_{n,t})$ in slot $t$ is given by
\begin{equation}\label{Equ: cap-consumption}
	h_{n,t}(x_{n,t},y_{n,t})
	\triangleq
	\data_{n,t} x_{n,t} + \indata_{n,t} x_{n,t} y_{n,t}.
\end{equation}

Based on the above discussions on utility and costs,  we express the monthly payoff of MU $n\in\mathcal{N}$ as follows:
\begin{equation}
	\begin{aligned}\label{Equ: MU payoff}
		S(\bm{x}_{n},\bm{y}_{n})
		\triangleq
		&\textstyle
		\sum\limits_{t=1}^{T} f_{n,t}(x_{n,t},y_{n,t}) \\
		&\textstyle
		- \adfee\left[ \sum\limits_{t=1}^{T}h_{n,t}(x_{n,t},y_{n,t}) - \dcap_{n} \right]^+
		-\pcap_{n},
	\end{aligned}
\end{equation}
where $\dcap_{n}$ and $\pcap_{n}$ are the monthly data cap and the monthly subscription fee of MU $n$, respectively.
Moreover. $\adfee$ is the per-unit fee for data usage exceeding the monthly data cap.

%
Each MU $n\in\mathcal{N}$ will (selfishly) maximize his monthly payoff $S(\bm{x}_{n},\bm{y}_{n})$.
In practice, however, each MU $n\in\mathcal{N}$ has to determine $(x_{n,t},y_{n,t})$ sequentially in each slot $t$ without knowing the future information.
Therefore, the MU's payoff maximization is an online Joint Content Acquisition and Task Offloading (J-CATO) problem.
We will study the off-line problem in Section \ref{Section: Off-line Analysis and Insights} and investigate the online problem in Section \ref{Section: MU On-Line JCOP}.

\begin{table}
	\setlength{\abovecaptionskip}{3pt}
	\setlength{\belowcaptionskip}{0pt}
	\renewcommand{\arraystretch}{1.25}		
	\caption{Key Notations.}
	\label{Table: Key notations}
	\centering
	\begin{tabular}{|c|c|l|}
		\hline
		\multicolumn{2}{|c|}{\textbf{Symbols}} 		& $\qquad\qquad\qquad$\textbf{Physical Meaning}					\\
		\hline\hline
		\multirow{13}{*}{MU}
		& $\data_{n,t}$					& The data-usage volume for MU $n$ in slot $t$	\\
		& $\cycle_{n,t}$				& The computation volume for MU $n$ in slot $t$	\\
		& $\indata_{n,t}$				& The raw data volume for MU $n$ in slot $t$	\\
		\cline{2-3}
		& $\val_{n,t}$				& MU $n$'s content valuation in slot $t$	\\
		& $\cost_{n,t}$				& MU $n$'s cost sensitivity in slot $t$	\\
		& $x_{n,t}$					& Content-acquiring decision of MU $n$ in slot $t$	\\
		& $y_{n,t}$					& Task-offloading decision of MU $n$ in slot $t$	\\
		& $z_{n,t}$					& The execution decision of MU $n$ in slot $t$	\\
		\cline{2-3}
		& $U_{n,t}(\cdot)$			& Satisfaction of MU $n$ in slot $t$, defined in (\ref{Equ: satisfaction})\\
		& $E_{n,t}(\cdot)$			& Dissatisfaction of MU $n$ in slot $t$, defined in (\ref{Equ: dissatisfaction})\\
		& $f_{n,t}(\cdot)$			& Virtual payoff of MU $n$ in slot $t$, defined in (\ref{Equ: virtual payoff})		\\
		& $h_{n,t}(\cdot)$			& Data usage of MU $n$ in slot $t$, defined in (\ref{Equ: cap-consumption})\\
		& $S_{n}(\cdot)$			& Monthly payoff of MU $n$, defined in (\ref{Equ: MU payoff})\\
		\hline
		\multirow{4}{*}{ISP}
		& $\dcap$					& The data cap offered by ISP	\\
		& $\pcap$					& The subscription fee charged by ISP	\\
		& $\adfee$					& The per-unit fee charged by ISP	\\
		& $V_{\text{ISP}}(\cdot)$	& The total revenue of ISP, defined in (\ref{Equ: Revenue of ISP})\\
		\hline
		\multirow{2}{*}{ESP}	
		& $p_{t}$					& The price of edge service in slot $t$	\\
		& $V_{\text{ESP}}(\cdot)$	& The total revenue of ESP, defined in (\ref{Equ: Revenue of ESP})\\
		\hline
	\end{tabular}
\end{table}

\subsection{Revenues of Providers}\label{Subsection: Revenues of Providers}
Next we introduce the revenue of each service provider (i.e., ESP, ISP, and CPs) based on the MU formulation.

\subsubsection{ESP Revenue}
Edge Service Provider (ESP) profits from the edge computing service and determines the unit price $p_{t}$ of utilizing the edge servers in each slot $t\in\mathcal{T}$.
We denote $\bm{p}=(p_{t}:t\in\mathcal{T})$ as the ESP's pricing for the edge service.
Accordingly, given all the MUs' content-acquiring decisions $\textbf{\textit{X}}=(\bm{x}_{n}:n\in\mathcal{N})$ and task-offloading decisions $\textbf{\textit{Y}}=(\bm{y}_{n}:n\in\mathcal{N})$, the total revenue of ESP is given by
\begin{equation}\label{Equ: Revenue of ESP}
	V_{\text{ESP}}(\bm{p},\textbf{\textit{X}},\textbf{\textit{Y}})
	\triangleq\textstyle
	\sum\limits_{n=1}^{N}\sum\limits_{t=1}^{T} p_{t}\cycle_{n,t} x_{n,t}y_{n,t},
\end{equation}
where $x_{n,t}$ and $y_{n,t}$ are MU $n$'s decisions in slot $t$, and depend on the pricing decisions of ESP.

\subsubsection{ISP Revenue}
Internet Service Provider (ISP) profits from the wireless data service based on the three-part tariff data plans.
Specifically, ISP's revenue consists of the monthly subscription fee and the overage fee for exceeding the monthly data cap.
Given all the MUs' content-acquiring decisions $\textbf{\textit{X}}=(\bm{x}_{n}:n\in\mathcal{N})$ and task-offloading decisions $\textbf{\textit{Y}}=(\bm{y}_{n}:n\in\mathcal{N})$, the monthly revenue of ISP is given by
\begin{equation}\label{Equ: Revenue of ISP}
	V_{\text{ISP}}(\textbf{\textit{X}},\textbf{\textit{Y}})
	\triangleq\textstyle
	\sum\limits_{n=1}^{N}\left( \pcap_{n} + \adfee\left[ \sum\limits_{t=1}^{T}h_{n,t}(x_{n,t},y_{n,t}) - \dcap_{n} \right]^+  \right),
\end{equation}
where $\dcap_{n}$ and $\pcap_{n}$ represent the monthly data cap and the monthly subscription fee of MU $n$, respectively.
In particular, we note that MUs' decisions $(\textbf{\textit{X}},\textbf{\textit{Y}})$ depend on how ESP prices the edge service, i.e., $\bm{p}$.
That is, ESP's pricing decisions may also affect ISP's revenue, which will be demonstrated in Section \ref{Subsection: Economic Impact of Edge Service}.

\subsubsection{CP Revenue}
Content Provider (CP) profits from displaying advertisements when the MUs are using the mobile applications \cite{joe2018sponsoring}.
Intuitively, the longer time period the MUs spend on the mobile applications, the more advertisements can be displayed.
Therefore, the revenue of CP is positively related to the total time period that the MUs spend on the content service, i.e., the content-acquiring decisions $\textbf{\textit{X}}$ of MUs.
Accordingly, we model the total revenue of CPs as follows:
\begin{equation}\label{Equ: Revenue of CP}
	V_{\text{CP}}(\textbf{\textit{X}})
	\triangleq\textstyle
	v\left(\sum\limits_{n=1}^{N}\sum\limits_{t=1}^{T} x_{n,t}\right),
\end{equation}
where $v(\cdot)$ represents a general revenue function, and is assumed to be increasing and concave as in previous literatures (e.g., \cite{joe2018sponsoring}).
Recall that the MUs' acquiring decisions $\textbf{\textit{X}}$ depends on ESP's pricing $\bm{p}$.
This means that ESP's pricing decision will affect the revenue of CPs as well.
We will demonstrate this in Section \ref{Section: Numerical}.

So far we have introduced the system model.
Next we start with analyzing the MU's off-line problem in Section \ref{Section: Off-line Analysis and Insights}, and study the MU's online strategy in Section \ref{Section: MU On-Line JCOP} (based on the off-line insights).
We then investigate how the ESP monetizes the edge service Section \ref{Section: ESP Pricing Problem}.

\section{MU Decision Problem: Off-line Analysis and Insights}\label{Section: Off-line Analysis and Insights}
This section focuses on the MU's off-line payoff maximization problem to unveil the key insights.
In particular, our analysis focuses on a generic MU, thus will neglect the MU index $n$ unless there is confusion.

\subsection{Problem Reformulation}
Suppose that the MU knows all the future information in advance, then the off-line Joint Content Acquisition and Task Offloading (J-CATO) problem is given by
\begin{problem}[Off-Line J-CATO]\label{Problem: JCOP}
	\begin{subequations}
		\begin{align}
			\{\bm{x^*},\bm{y^*}\}=\arg
			\max\limits_{\bm{x},\bm{y}} 			&\quad S(\bm{x},\bm{y})  \\
			\textit{s.t.}	&\quad x_{t}\in[0,1],\ \forall t\in\mathcal{T},\\
			&\quad y_{t}\in[0,1],\ \forall t\in\mathcal{T}.
		\end{align}
	\end{subequations}
\end{problem}

Problem \ref{Problem: JCOP} exhibits two difficulties.
First, the product term between variables $x_{t}$ and $y_{t}$ is non-convex.
Second, the overage payment is piece-wise linear.
Next we reformulate Problem \ref{Problem: JCOP} in the following two steps.

First, we introduce a set of new variables $\bm{z}=(z_{t}\in[0,1]:t\in\mathcal{T})$ to eliminate the product terms.
That is, we use $z_{t}$ to replace $x_{t}y_{t}$ for any $t\in\mathcal{T}$ in the MU monthly payoff, defined in (\ref{Equ: MU payoff}).
To ensure equivalence, we introduce the following conditions:
\begin{equation}\label{Equ: Reformulation Constraint r}
	0\le z_{t} \le x_{t},\ \forall t\in\mathcal{T}.
\end{equation}
For presentation convenience, we will refer to $\bm{z}$ as the MU's \textbf{executing decision}.
Accordingly, we express the MU's virtual payoff $f_{t}(\cdot)$ defined in (\ref{Equ: virtual payoff}) and the wireless data usage $h_{t}(\cdot)$ defined in (\ref{Equ: cap-consumption}) as follows:
	\begin{subequations}\label{Equ: new f g}
		\begin{align}
			\tilde{f}_{t}(x_{t},z_{t}) & \triangleq U_{t}(x_{t})
			- \enCPT_{t}\big((x_{t}-z_{t})\cycle_{t}\big) - p_{t}z_{t} \cycle_{t} ,\\
			\tilde{h}_{t}(x_{t},z_{t}) & \triangleq \data_{t} x_{t} + \indata_{t} z_{t}.\label{Equ: tilde g}
		\end{align}
	\end{subequations}

Second, we introduce a new variable $s\in\mathbb{R}$ to linearize the piece-wise linear term in the MU's monthly payoff, defined in (\ref{Equ: MU payoff}).
Mathematically, the new variable $s$ represents the data usage exceeding the monthly data cap.
For equivalence, we should ensure the following two conditions
	\begin{subequations}\label{Equ: Reformulation Constraint s}
		\begin{align}
			s&\ge 0, \label{Equ: Reformulation Constraint s 1}	\\
			s&\ge\textstyle \sum\limits_{t=1}^{T}\tilde{h}_{t}(x_{t},z_{t}) -\dcap. \label{Equ: Reformulation Constraint s 2}
		\end{align}
	\end{subequations}

Based on the above reformulations, now we are able to express the MU's monthly payoff as follows:
\begin{equation}
	\tilde{S}(\bm{x},\bm{z},s)
	\triangleq\textstyle
	\sum\limits_{t=1}^{T}\tilde{f}(x_{t},z_{t}) - \adfee s,
\end{equation}
and the reformulated off-line J-CATO problem is given by

\begin{problem}[Reformulated Off-Line J-CATO]\label{Problem: Reformulated JCOP}
	\begin{subequations}
		\begin{align}
			\{\bm{x^*},\bm{z^*},s^*\}=\arg
			\max\limits_{\bm{x},\bm{z},s} &\quad \tilde{S}(\bm{x},\bm{z},s) \\
			\textit{s.t.} 	&\quad (\ref{Equ: Reformulation Constraint r}),(\ref{Equ: Reformulation Constraint s}) \\
			&\quad x_{t}\in[0,1],\ \forall t\in\mathcal{T},\\
			&\quad z_{t}\in[0,1],\ \forall t\in\mathcal{T}.\label{Equ: Reformulated JCOP integer}
		\end{align}
	\end{subequations}
\end{problem}

Note that after the above reformulation, Problem \ref{Problem: Reformulated JCOP} is a convex optimization with differentiable objective and constraints.
We can solve it by analyzing the Karush-Kuhn-Tucker (KKT) conditions of Problem \ref{Problem: Reformulated JCOP}.
Furthermore, we have
\begin{equation}
	y_{t}^*=\left\{
	\begin{aligned}
		& 0,					&\textit{if }& x_{t}^*=0,\\
		& {z_{t}^*}/{x_{t}^*},	&\textit{if }& x_{t}^*>0,
	\end{aligned}
	\right.
	\quad\forall t\in\mathcal{T}.
\end{equation}

\subsection{Solution of Problem \ref{Problem: Reformulated JCOP}}
Next we solve Problem \ref{Problem: Reformulated JCOP} based on the Karush-Kuhn-Tucker (KKT) conditions (as Problem \ref{Problem: Reformulated JCOP} is convex).
Instead of directly presenting the mathematical solution, we will step-by-step elaborate the key insights of the KKT analysis.
These insights are crucially valuable to our online strategy in Section \ref{Section: MU On-Line JCOP}.

\subsubsection{Shadow Price}
Problem \ref{Problem: Reformulated JCOP} is a constrained optimization problem.
As we will see later, the constraint (\ref{Equ: Reformulation Constraint s 2}) plays a significant role in the KKT analysis.
Hence we let $\lambda$ denote the Lagrangian multiplier associated with constraint (\ref{Equ: Reformulation Constraint s 2}).
According to the constrained optimization in economics \cite{heckman1974shadow}, $\lambda$ can be interpreted as the \textit{shadow price} of the wireless data service.
We first present a basic property related to the shadow price $\lambda$ in Proposition \ref{Proposition: lambda}.
\begin{proposition}\label{Proposition: lambda}
	The shadow price $\lambda^*$ satisfying the KKT conditions of Problem \ref{Problem: Reformulated JCOP} is no larger than $\adfee$, i.e., $0\le\lambda^*\le\adfee$.
\end{proposition}

Proposition \ref{Proposition: lambda} shows a feasible range of the optimal shadow price satisfying the KKT conditions of Problem \ref{Problem: Reformulated JCOP}.
The feasible range mainly results from the three-part tariff wireless data service $\{\dcap,\pcap,\adfee\}$.
Specifically, the shadow price of the wireless data usage is zero if the MU's monthly data cap $\dcap$ is sufficient comparing to the MU's total data usage.
However, the shadow price is the same as the per-unit overage fee $\adfee$ if the MU's total data usage exceeds the monthly data cap $\dcap$.
In addition, the shadow price locates between the interval $(0,\adfee)$ if the MU's total data usage exactly equals to the monthly data cap $\dcap$.
We will introduce how to compute the optimal shadow price $\lambda^*$ in Section \ref{Subsubsection: Optimal Shadow Price}.
Before that, we first demonstrate the primal-dual solution structure in Section \ref{Subsubsection: Solution Structure}.

\subsubsection{Solution Structure}\label{Subsubsection: Solution Structure}
We elaborate the mathematical structure of the optimal primal-dual solution $(\bm{x}^{*},\bm{z}^{*},\lambda^*)$ of four different cases in Lemmas \ref{Lemma: primal-dual I}$\sim$\ref{Lemma: primal-dual IV}, respectively.
The four cases are characterized based on the MU's time-dependent features, i.e., $(\cost_{t},\val_{t})$.
Recall that $\val_{t}$ is the scalar in (\ref{Equ: satisfaction}) and indicates MU's valuation on the content service.
The larger $\val_{t}$ value means that MU has a greater demand on the content service in slot $t$.
In addition, $\cost_{t}$ is the scalar in (\ref{Equ: dissatisfaction}) and measures MU's sensitivity to the local execution.
The larger $\cost_{t}$ value means that MU is less tolerant to Quality of Experience (QoE) reduction, thus prefers to the choice of edge-execution.
Next we present the results for the four cases  in Lemmas \ref{Lemma: primal-dual I}$\sim$\ref{Lemma: primal-dual IV}.
\begin{lemma}\label{Lemma: primal-dual I}
	The optimal primal-dual solution $(\bm{x}^*,\bm{z}^*,\lambda^*)$ satisfies $(x_{t}^{*},z_{t}^{*})=(1,0)$, for the case of $(\cost_{t},\val_{t})\in\Omega_{t}^{\text{I}}(\lambda^*)$, where the set $\Omega_{t}^{\text{I}}(\cdot)$ is defined as follows:
	\begin{equation}\label{Equ: set I}
		\Omega_{t}^{\text{I}}(\lambda)
		\triangleq
		\Big\{\textstyle (\cost,\val): \val > \frac{\cost\cycle_{t} \encpt'_{t}(\cycle_{t}) + \data_{t}\lambda}{ u'_{t}(1)},
		\ \cost  < \frac{p_{t}\cycle_{t}+\indata_{t}\lambda}{\cycle_{t} \encpt'_{t}(\cycle_{t})}  \Big\},
	\end{equation}
	where $u'_{t}(\cdot)$ and $e'_{t}(\cdot)$ represent the derivative of the normalized utility and cost functions, respectively.
	Moreover, $(\data_{t},\indata_{t},\cycle_{t})$ is the content service realization in slot $t$.
\end{lemma}

Lemma \ref{Lemma: primal-dual I} presents the optimal primal-dual solution structure for the case of $(\cost_{t},\val_{t})\in\Omega_{t}^{\text{I}}(\lambda^*)$.
We elaborate this lemma in the following two aspects.
\begin{itemize}
	\item First, $\Omega_{t}^{\text{I}}(\lambda)$, defined in (\ref{Equ: set I}), is a set in terms of the MU's characteristics $(\cost,\val)$.
	It is time-dependent as it is defined based on the MU's satisfaction $u_{t}(\cdot)$ and dissatisfaction $e_{t}(\cdot)$ as well as the MU's content service realization $(\data_{t},\indata_{t},\cycle_{t})$.
	Moreover, the set $\Omega_{t}^{\text{I}}(\lambda)$ depends on the shadow price $\lambda$ as well.
	Hence the shadow price $\lambda$ also affects whether the MU's characteristic $(\cost_{t},\val_{t})$ belongs to the set $\Omega_{t}^{\text{I}}(\lambda)$ in slot $t$.
	
	\item Second, the case of $(\cost_{t},\val_{t})\in\Omega_{t}^{\text{I}}(\lambda^*)$ in Lemma \ref{Lemma: primal-dual I} implies that the characteristic $(\cost_{t},\val_{t})$ corresponds to a \textit{weak-sensitivity} and \textit{high-valuation} state, denoted by the gray region in Fig. \ref{fig: OfflineRegion}.
	In this case, the MU is ``self-sufficient'' in terms of the computation capacity.
	That is, the MU is able to acquire the content service for the entire slot (i.e., $x_{t}^{*}=1$) under the pure local execution mode (i.e., $z_{t}^{*}=0$).
\end{itemize}

\begin{lemma}\label{Lemma: primal-dual II}
	The optimal primal-dual solution $(\bm{x}^*,\bm{z}^*,\lambda^*)$ satisfies $(x_{t}^{*},z_{t}^{*})=\left(1,z_{t}^{\text{II}}(\lambda^*)\right)$ for the case of $(\cost_{t},\val_{t})\in\Omega_{t}^{\text{II}}(\lambda^*)$.
	Speicfically, $z_{t}^{\text{II}}(\cdot)$ is given by
	\begin{equation}\label{Equ: z II}
		\textstyle
		z_{t}^{\text{II}}(\lambda)
		\triangleq
		1 - {\encpt'^{-1}_{t}\left( \frac{p_{t}\cycle_{t}+\indata_{t}\lambda}{\cost_{t}\cycle_{t}} \right)}\Big/{\cycle_{t}},
	\end{equation}
	where $\encpt'^{-1}_{t}(\cdot)$ represents the inverse function of $e_{t}(\cdot)$.
	Moerover, the set $\Omega_{t}^{\text{II}}(\cdot)$ is given by
	\begin{equation}\label{Equ: set II}
		\Omega_{t}^{\text{II}}(\lambda)
		\triangleq
		\Big\{\textstyle
		 (\cost,\val): \val > \frac{p_{t}\cycle_{t}+(\data_{t}+\indata_{t})\lambda}{u'_{t}(1)},
		\ \cost  \ge \frac{p_{t}\cycle_{t}+\indata_{t}\lambda}{\cycle_{t} \encpt'_{t}(\cycle_{t})} \Big\}.
	\end{equation}
\end{lemma}

Lemma \ref{Lemma: primal-dual II} presents the optimal primal-dual solution structure for the case of $(\cost_{t},\val_{t})\in\Omega_{t}^{\text{II}}(\lambda^*)$.
Similarly, the set $\Omega_{t}^{\text{II}}(\cdot)$, defined in (\ref{Equ: set II}), is time-dependent and affected by the shadow price $\lambda$.
Moreover, the case of $(\cost_{t},\val_{t})\in\Omega_{t}^{\text{II}}(\lambda^*)$ corresponds to \textit{strong-sensitivity} and \textit{high-valuation} state, denoted by the blue region in Fig. \ref{fig: OfflineRegion}.
In this case, the MU's computation capacity cannot meet his computation tasks due to the strong-sensitivity (i.e., a large $\cost_{t}$ value).
Therefore, in spite of the full acquiring decision (i.e., $x_{t}^{*}=1$), the MU tends to partially offload his computation task for remote execution (i.e., $0<z_{t}^{\textit{II}}(\lambda^*)<1$).

\begin{lemma}\label{Lemma: primal-dual III}
	The optimal primal-dual solution $(\bm{x}^*,\bm{z}^*,\lambda^*)$ satisfies $(x_{t}^{*},z_{t}^{*})=\left(x^{\text{III}}_{t}(\lambda^*),z^{\text{III}}_{t}(\lambda^*)\right)$ for the case of $(\cost_{t},\val_{t})\in\Omega_{t}^{\text{III}}(\lambda^*)$, where $x^{\text{III}}_{t}(\cdot)$ and $z^{\text{III}}_{t}(\cdot)$ are given by
	\begin{subequations}
		\begin{align}
			x_{t}^{\text{III}}(\lambda)
			&\textstyle=  u'^{-1}_{t}\left( \frac{p_{t}\cycle_{t}+(\data_{t}+\indata_{t})\lambda}{\val_{t}} \right),\\
			z_{t}^{\text{III}}(\lambda)
			&\textstyle= x_{t}^{\text{III}}(\lambda)-\encpt'^{-1}_{t}\left( \frac{p_{t}\cycle_{t}+\indata_{t}\lambda}{\cost_{t}\cycle_{t}} \right)\Big/\cycle_{t},
		\end{align}	
	\end{subequations}	
	where $u'^{-1}_{t}(\cdot)$ represents the inverse function of $u'_{t}(\cdot)$.
	Moreover, the set $\Omega_{t}^{\text{III}}(\cdot)$ is
	\begin{equation}\label{Equ: set III}
	\begin{aligned}
		\Omega_{t}^{\text{III}}(\lambda)
		\triangleq
		&\bigg\{\textstyle (\cost,\val): \frac{p_{t}\cycle_{t} +(\data_{t}+\indata_{t})\lambda }{ u'_{t}\left( \encpt'^{-1}_{t}\left(\frac{\cycle_{t} p_{t}+\indata_{t}\lambda}{\cost\cycle_{t}}\right)\big/\cycle_{t} \right) }\le \val \le \\
		&\qquad\qquad\qquad\qquad\qquad\qquad \textstyle\frac{p_{t}\cycle_{t}+(\data_{t}+\indata_{t})\lambda}{u'_{t}(1)}\bigg\}.
	\end{aligned}		
	\end{equation}
\end{lemma}

Lemma \ref{Lemma: primal-dual III} presents the optimal primal-dual solution structure for the case of $(\cost_{t},\val_{t})\in\Omega_{t}^{\text{III}}(\lambda^*)$.
The set $\Omega_{t}^{\text{III}}(\cdot)$, defined in (\ref{Equ: set III}), corresponds to the \textit{strong-sensitivity} and \textit{low-valuation} state, denoted by the green region in Fig. \ref{fig: OfflineRegion}.
Different from Lemma \ref{Lemma: primal-dual II}, the MU in this case will only consume part of the slot on the content service (i.e., $0<x_{t}^{\textit{III}}(\lambda^*)<1$) due to the low-valuation (i.e., a small $\val_{t}$ value).

\begin{figure}
	\setlength{\abovecaptionskip}{3pt}
	\setlength{\belowcaptionskip}{0pt}
	\centering
	\includegraphics[width=0.6\linewidth]{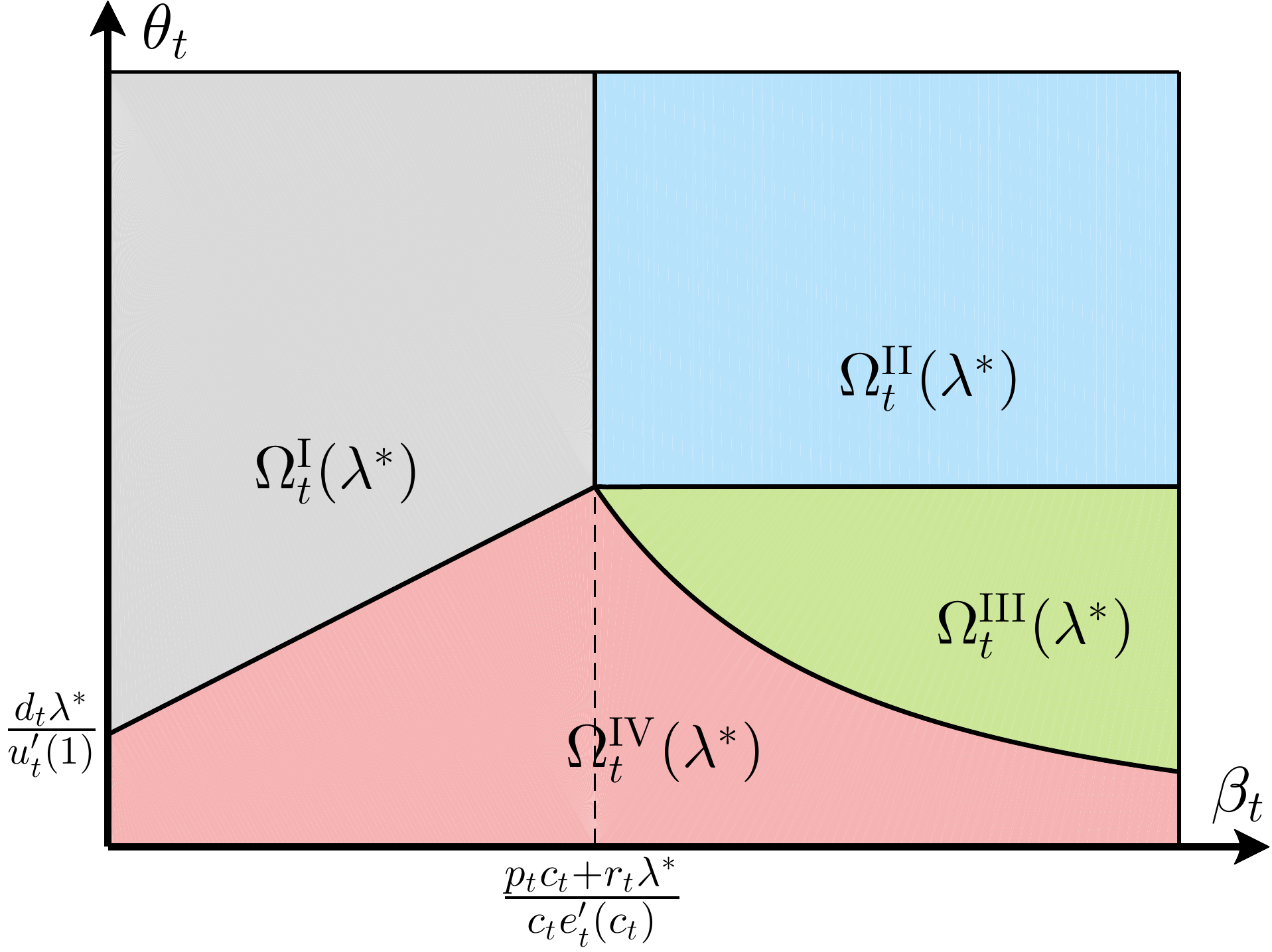}
	\caption{Regimes of MU's state.}
	\label{fig: OfflineRegion}
\end{figure}
\begin{lemma}\label{Lemma: primal-dual IV}
	The optimal primal-dual solution $(\bm{x}^*,\bm{z}^*,\lambda^*)$ satisfies $(x_{t}^{*},z_{t}^{*})=\left(x^{\text{IV}}_{t}(\lambda^*),0\right)$ for the case of $(\cost_{t},\val_{t})\in\Omega_{t}^{\text{IV}}(\lambda^*)$, where $x^{\text{IV}}_{t}(\cdot)$ satisfies
	\begin{equation}
		\textstyle
		\val_{t}u'_{t}\left(x_{t}^{\text{IV}}\right) -\cost\cycle_{t} \encpt'_{t}\left(x_{t}^{\text{IV}}\cycle_{t}\right) =\data_{t}\lambda ,
	\end{equation}
	and the set $\Omega_{t}^{\text{IV}}(\cdot)$ is given by
	\begin{equation}\label{Equ: set IV}
	\begin{aligned}
		\Omega_{t}^{\text{IV}}(\lambda)
		\triangleq
		&\bigg\{\textstyle (\cost,\val): \val<\frac{p_{t}\cycle_{t} +(\data_{t}+\indata_{t})\lambda }{ u'_{t}\left( \encpt'^{-1}_{t}\left(\frac{\cycle_{t} p_{t}+\indata_{t}\lambda}{\cost\cycle_{t}}\right)\big/\cycle_{t} \right) },\\
		&\qquad\qquad\qquad\qquad\qquad
		\textstyle \val \le \frac{\cost\cycle_{t} \encpt'_{t}(\cycle_{t}) + \data_{t}\lambda}{u'_{t}(1)} \bigg\}.
	\end{aligned}		
	\end{equation}
\end{lemma}

Lemma \ref{Lemma: primal-dual IV} presents the optimal primal-dual solution structure for the case of $(\cost_{t},\val_{t})\in\Omega_{t}^{\text{IV}}(\lambda^*)$.
The set $\Omega_{t}^{\text{IV}}(\cdot)$, defined in (\ref{Equ: set IV}), corresponds to the \textit{weak-sensitivity} and \textit{low-valuation} state, denoted by the red region in Fig. \ref{fig: OfflineRegion}.
Specifically, the low valuation and weak sensitivity enable the MU to be ``self-sufficiency'' in terms of the computation capacity.
Accordingly, the MU tends to consume part of the slot on the content service (i.e., $0<x_{t}^{\textit{IV}}(\lambda^*)<1$) under the pure local execution mode (i.e., $z_{t}^{*}=0$).

So far, we have introduced the optimal primal solutions $(\bm{x^*},\bm{z^*})$ given the optimal shadow price $\lambda^*$ with respect to four cases in Lemmas \ref{Lemma: primal-dual I}$\sim$\ref{Lemma: primal-dual IV}, respectively.
It is obvious that the payoff-maximizing MU will never choose edge-execution if the price $p_{t}\ge \bar{E}$ for any $t\in\mathcal{T}$, where $\bar{E}$ is given by
\begin{equation}\label{Equ: E_bar}
\bar{E}\triangleq\max_{t\in\mathcal{T}}\{\cost_{t}\cdot e'_{t}(\bar{\cycle})\}.
\end{equation}
This implies that the pricing strategy $\bm{\bar{p}}\triangleq[\bar{E},\bar{E},...,\bar{E}]$ corresponds to the case where ESP does not offer edge service, i.e., $V_{\text{ESP}}(\bm{\bar{p}})=0$.
This observation provides a critical price upper bound when we analyze the ESP's pricing policy in Section \ref{Section: ESP Pricing Problem}.

\subsubsection{Optimal Shadow Price}\label{Subsubsection: Optimal Shadow Price}
Lemmas \ref{Lemma: primal-dual I}$\sim$\ref{Lemma: primal-dual IV} imply that the shadow price $\lambda$ plays a significant role on the KKT analysis.
Before deriving the optimal shadow price $\lambda^*$, for notation simplicity, we define a mapping $W_{t}(\lambda)$ for each $t\in\mathcal{T}$ (based on Lemmas \ref{Lemma: primal-dual I}$\sim$\ref{Lemma: primal-dual IV}) as follows:
\begin{equation}\label{Equ: W_t}
	W_{t}(\lambda)
	\triangleq
	\left\{
	\begin{aligned}
		& \left(1,0\right),				
		&\textit{if }& (\cost_{t},\val_{t})\in\Omega_{t}^{\text{I}}(\lambda), \\
		& \left(1,z_{t}^{\text{II}}(\lambda)\right),			
		&\textit{if }& (\cost_{t},\val_{t})\in\Omega_{t}^{\text{II}}(\lambda), \\
		& \left( x^{\text{III}}_{t}(\lambda),z^{\text{III}}_{t}(\lambda) \right),	
		&\textit{if }& (\cost_{t},\val_{t})\in\Omega_{t}^{\text{III}}(\lambda), \\
		& \left( x^{\text{IV}}_{t}(\lambda),0\right),		
		&\textit{if }& (\cost_{t},\val_{t})\in\Omega_{t}^{\text{IV}}(\lambda),
	\end{aligned}
	\right.
\end{equation}
where $W_{t}(\cdot):\mathbb{R}\rightarrow\mathbb{R}^2$ maps from a shadow price value to the MU's one-slot acquiring decision and executing decision.
We denote the \textit{potential data usage under $\lambda$} as:
\begin{equation}\label{Equ: A_lambda}
	A(\lambda)
	\triangleq \textstyle
	\sum\limits_{t=1}^{T}\tilde{h}_{t}\Big( W_{t}(\lambda) \Big),
\end{equation}
where $\tilde{h}_{t}(\cdot)$, defined in (\ref{Equ: tilde g}), takes the MU's acquiring decision and executing decision as the input, and represents the MU's wireless data usage in slot $t$.
Note that the potential data usage $A(\lambda)$ is weakly-decreasing in $\lambda$.
That is, a higher shadow price leads to less wireless data usage.

Lemma \ref{Lemma: optimal primal} presents the optimal shadow price.
\begin{lemma}\label{Lemma: optimal primal}
	The optimal dual solution $\lambda^*$ satisfying the KKT conditions of  Problem \ref{Problem: Reformulated JCOP} is given by
	\begin{equation}\label{Equ: lambda opt}
		\lambda^*= \min(\adfee,\lambda^{\dag}),
	\end{equation}
	where $\lambda^{\dag}$ is defined as follows
	\begin{equation}\label{Equ: lambda dag}
		\begin{aligned}
			\lambda^{\dag}\triangleq \min_{\lambda\ge0} \  \lambda
			\qquad\textit{ s.t. } \  A(\lambda)\le\dcap.
		\end{aligned}
	\end{equation}
\end{lemma}

Lemma \ref{Lemma: optimal primal} implies that the optimal primal-dual solution of Problem \ref{Problem: Reformulated JCOP} has two possibilities.
\begin{itemize}
	\item If $A(\adfee)>\dcap$, then we have $\lambda^{\dag}>\adfee$ according to the definition (\ref{Equ: lambda dag}).
	Therefore, the optimal shadow price is the same as the per-unit fee according to (\ref{Equ: lambda opt}), i.e., $\lambda^*=\adfee$.
	In this case, the MU has a large wireless data demand, thus is charged overage fee $\adfee\left[A(\lambda^*)-\dcap\right]$ for the over usage $A(\lambda^*)-\dcap$.
	
	\item If $A(\adfee)\le\dcap$, then we have $\lambda^{\dag}\le\adfee$ according to the definition (\ref{Equ: lambda dag}).
	Therefore, the optimal shadow price is the same as $\lambda^{\dag}$ according to (\ref{Equ: lambda opt}), i.e., $\lambda^*=\lambda^{\dag}$.
	In this case, the MU has a small wireless data demand, thus the monthly data cap $\dcap$ is sufficient, i.e., $\dcap\ge A(\lambda^*)$.
	
\end{itemize}

Theorem \ref{Theorem: MU KKT} presents the optimal solution of Problem \ref{Problem: Reformulated JCOP} based on the previous KKT analysis in Lemmas \ref{Lemma: primal-dual I}$\sim$\ref{Lemma: optimal primal}.
Due to space limit, the detailed proof is given in an online technical report \cite{report}.
\begin{theorem}\label{Theorem: MU KKT}
	The shadow price $\lambda^*$ given in (\ref{Equ: lambda opt}) together with the primal solutions $(x_{t}^{*},z_{t}^{*})=W_{t}(\lambda^*)$ for any $t\in\mathcal{T}$ satisfy the KKT conditions of Problem \ref{Problem: Reformulated JCOP}.
\end{theorem}

By now, we have solved the MU's off-line J-CATO problem based on the KKT conditions of Problem \ref{Problem: Reformulated JCOP}.
The above KKT analysis implies the significance of the shadow price of the wireless data usage.
It also motivates the strategy of the online context in Section \ref{Section: MU On-Line JCOP}.

\section{MU Online Strategy}\label{Section: MU On-Line JCOP}
This section focuses on the MU's online J-CATO problem and proposes an online strategy.
Specifically, Section \ref{Subsection: Basic Idea and Strategy} elaborates the basic idea.
We then analyze the theoretic performance in Section \ref{Subsection: MU Performance Analysis}.

\subsection{Basic Idea and Strategy}\label{Subsection: Basic Idea and Strategy}
Recall that the MU's online payoff maximization problem is given by
\begin{subequations}\label{Equ: MU Online JCOP}
	\begin{align}
		\max\limits_{\bm{x},\bm{z}}&\quad \textstyle\sum\limits_{t=1}^{T} \tilde{f}_{t}(x_{t},z_{t})
		- \adfee\left[ \sum\limits_{t=1}^{T}\tilde{h}_{t}(x_{t},z_{t}) - \dcap \right]^+
		-\pcap	\label{Equ: MU Online JCOP - Objective} \\
		\textit{s.t.}	&\quad 0\le z_{t}\le x_{t}\le 1,\ \forall t\in\mathcal{T}.
	\end{align}
\end{subequations}

The above problem in (\ref{Equ: MU Online JCOP}) share some similarity with the Lyapunov optimization framework \cite{neely2010stochastic}.
But (\ref{Equ: MU Online JCOP}) is more challenging, since the second term in the objective (\ref{Equ: MU Online JCOP - Objective}), i.e., $\adfee[ \sum_{t=1}^{T}\tilde{h}_{t}(x_{t},z_{t}) - \dcap ]^+$, is non-additive over time.
There have been some studies (e.g., \cite{mahdavi2012trading,liakopoulos2019cautious}) focusing on the problems with additive objective and long-term constraints.
But our problem is more general as the piece-wise linear term in (\ref{Equ: MU Online JCOP - Objective}) degenerates into a long-term constraint when $\adfee$ substantially increases.
Next we introduce our method based on the insights discussed in Section \ref{Section: Off-line Analysis and Insights}.
We start with defining the following augmented Lagrangian function for each time slot $t\in\mathcal{T}$:
\begin{equation}\label{Equ: augmented Lagrangian}
	L_{t}(x_{t},z_{t},\lambda)
	\triangleq
	\textstyle
	\tilde{f}_{t}(x_{t},z_{t}) - \lambda\cdot\left( \tilde{h}_{t}(x_{t},z_{t}) - \frac{\dcap}{T} \right),
\end{equation}
where $\lambda\in[0,\adfee]$ is the shadow price.
Particularly, it is obvious that minimizing $\sum_{t=1}^{T}L_{t}(x_{t},z_{t},\lambda)$ over $\lambda\in[0,\adfee]$ leads to the MU's monthly payoff $\tilde{S}(\bm{x},\bm{z})$.
That is, the following equality holds
\begin{equation}\label{Equ: U=Lt}
	\tilde{S}(\bm{x},\bm{z})
	=\textstyle
	\min\limits_{\lambda\in[0,\adfee]} \sum\limits_{t=1}^{T}L_{t}(x_{t},z_{t},\lambda) - \pcap,\ \forall(\bm{x},\bm{z}),
\end{equation}
which implies the inherent relation between the augmented Lagrangian in (\ref{Equ: augmented Lagrangian}) and the MU's monthly payoff $\tilde{S}(\bm{x},\bm{z})$.

\begin{algorithm}[t]
	\caption{User's Online Strategy $\mathcal{A}$}\label{Algorithm: MU}
	\SetKwInOut{Input}{Input}
	\SetKwInOut{Output}{Output}
	\Input{Initial $U_{t}(\cdot)$ and $\enCPT_{t}(\cdot)$ for any $t\in\mathcal{T}$.}
	\Output{$(\bm{\hat{x}},\bm{\hat{z}})$ and $\hat{\bm{\lambda}}$.}
	\textbf{Initial} $\hat{\lambda}_1=0$ and step size $\bm{\eta}=\{\eta_{t}:\forall t\in\mathcal{T}\}$.	\\
	\For {$t=1$ \KwTo $T$ }
	{
		\textbf{\textit{Determine}} $(\hat{x}_{t},\hat{z}_{t})$ based on $\hat{\lambda}_{t}$ according to
		\begin{equation}
		(\hat{x}_{t},\hat{z}_{t})=W_{t}\left(\hat{\lambda}_{t}\right).	
		\end{equation} \\
		\textbf{\textit{Update}} $\hat{\lambda}_{t+1}$ based on $(\hat{x}_{t},\hat{z}_{t})$ according to
		\begin{equation}\label{Equ: MU update lambda}
		\hat{\lambda}_{t+1}=\mathcal{P}_{[0,\adfee]}\left( \hat{\lambda}_{t} + \eta_{t} \left[\textstyle \tilde{h}_{t}(\hat{x}_{t},\hat{z}_{t})-\frac{\dcap}{T}\right] \right).
		\end{equation}
	}
\end{algorithm}

We present the strategy $\mathcal{A}$ for MU's online J-CATO problem in Algorithm \ref{Algorithm: MU}.
In each slot $t\in\mathcal{T}$, the strategy $\mathcal{A}$ mainly includes two steps, which are elaborated as follows.
\begin{itemize}
	\item \textit{Line 3:} The MU determines $(\hat{x}_{t},\hat{z}_{t})$ according to the current shadow price $\hat{\lambda}_{t}$ and the function $W_{t}(\cdot)$ defined in (\ref{Equ: W_t}).
	Based on our off-line analysis in Section \ref{Section: Off-line Analysis and Insights}, this step essentially generates for the MU the optimal acquiring and executing decisions with respect to the current shadow price $\hat{\lambda}_{t}$.
	
	\item \textit{Line 4:} The MU updates the shadow price $\tilde{\lambda}_{t+1}$ according to the current wireless data usage $\tilde{h}_{t}(\hat{x}_{t},\hat{z}_{t})$ and the average quota $q\triangleq\dcap/T$.
	The notation $\mathcal{P}_{[0,\adfee]}(\cdot)$ in (\ref{Equ: MU update lambda}) is a projection to the interval $[0,\adfee]$.
	If the current wireless data usage is greater (or smaller) than the average quota, then the MU will increase (or decrease) the shadow price with an appropriate step-size $\eta_{t}$.
\end{itemize}

\subsection{Performance Analysis}\label{Subsection: MU Performance Analysis}
We elaborate the performance of the strategy $\mathcal{A}$ in terms of the MU's average monthly payoff gap between the decisions $(\bm{\hat{x}},\bm{\hat{z}})$ and the optimal decisions in hindsight, i.e., $(\bm{x^*},\bm{z^*})=\arg\max \tilde{S}(\bm{x},\bm{z})$.
That is, we are interested in
\begin{equation}
G_{T}(\mathcal{A})
\triangleq
\frac{1}{T}\left[ \tilde{S}(\bm{x^*},\bm{z^*}) - \tilde{S}(\bm{\hat{x}},\bm{\hat{z}}) \right].
\end{equation}
As we will see later, the performance mainly depends on two factors, i.e., \textit{demand divergence} and \textit{consumption fluctuation}.
We first formally define the two factors in Definitions \ref{Definition: Xi} and \ref{Definition: Phi}, respectively.

\begin{definition}[Demand Divergence]\label{Definition: Xi}
	Given the MU's content service realization $(\bm{\data},\bm{\indata},\bm{\cycle})$, the \textit{demand divergence} (with respect to the average quota $q\triangleq\dcap/T$) in slot $t$ is
	\begin{equation}
		\xi_{t} \triangleq \big|q-\data_{t}-\indata_{t}\big|.
	\end{equation}
	Accordingly, the \textit{maximal demand divergence} is given by
	\begin{equation}
		\begin{aligned}
			\Xi = \max\left( \big|q-\bar{\data}-\bar{\indata}\big|,q \right).
		\end{aligned}	
	\end{equation}
\end{definition}

\begin{definition}[Consumption Fluctuation]\label{Definition: Phi}
	Given the MU's optimal decisions $(\bm{x^*},\bm{z^*})$ in hindsight, the per-slot leftover quota (comparing to the average quota $q\triangleq\dcap/T$) is  $l_{t}=q-\tilde{h}_{t}(x^{*}_{t},z^{*}_{t})$.
	Hence $\bar{l}=\sum_{t=1}^{T}l_t/T$ is the average leftover quota.
	The {consumption fluctuation} up to slot $t$ is
	\begin{equation}
		\psi_{t}\triangleq \big| t\cdot\bar{l} - \textstyle\sum_{i=1}^{t}l_{i} \big|.
	\end{equation}
	Accordingly, the \textit{maximal consumption fluctuation} is given by
	\begin{equation}
		\Psi \triangleq \max_{1\le t\le T}\psi_{t}.
	\end{equation}
\end{definition}

Basically, $\psi_{t}$ measures the absolute different between the cumulative leftover quota and the average case up to the $t$-th slot.
In an extreme case where the MU equally consumes wireless data every day (i.e., $\tilde{h}_{t}(x^*_{t},z^*_{t})$ are identical for any $t\in\mathcal{T}$), then the fluctuation is zero (i.e., $\psi_{t}=0$ for any $t$).

Theorem \ref{Theorem: User Online} presents the MU's payoff gap under the strategy $\mathcal{A}$.
The proof is given in  \cite{report}.
\begin{theorem}\label{Theorem: User Online}
	The solution $(\bm{\hat{x}},\bm{\hat{z}})$ generated by strategy $\mathcal{A}$ in Algorithm \ref{Algorithm: MU} achieves the following MU payoff gap
	\begin{equation}
		G_{T}(\mathcal{A})
		\le\textstyle
		\frac{1}{T}\left(\frac{\adfee^2}{2} \frac{1}{\eta_{T}} +\left(\frac{\Xi^2}{2} +\Xi\Psi\right)\sum\limits_{t=1}^{T}\eta_{t}\right),
	\end{equation}
	where $\Xi$ and $\Psi$ are defined in Definition \ref{Definition: Xi} and Definition \ref{Definition: Phi}, respectively.
	Moreover, with the step-size $\eta_{t}=\frac{\adfee}{\Xi\sqrt{T}}$, we have
	\begin{equation}
		G_{T}(\mathcal{A})\le\textstyle
		\frac{\adfee(\Xi+\Psi)}{\sqrt{T}}.
	\end{equation}
\end{theorem}

Theorem \ref{Theorem: User Online} shows that the performance of strategy $\mathcal{A}$ depends on the overage fee $\adfee$ besides the aforementioned {demand divergence} and {consumption fluctuation}.
This is because that the crucial uncertainty of MU's online J-CATO problem is the optimal shadow price, which is related to $\adfee$.
We will also illustrate this in Section \ref{Subsection: Evaluation on MU's JOCP Strategy}.

So far, we have intorduced the MU's online problem and proposed an online strategy.
The MU's online problem is closely related to the ESP's online pricing problem, since they all need to make the decisions sequentially in each slot.
 will be introduced next.

\section{ESP Pricing Problem}\label{Section: ESP Pricing Problem}

In this section, we study how ESP monetizes the edge computing service and propose a pricing policy.

Recall that ESP's total revenue from all the MUs $\mathcal{N}$ over the period $\mathcal{T}$ is given in (\ref{Equ: Revenue of ESP}).
Due to the asymmetric information, however, ESP cannot explicitly predict how the MUs response to its pricing decisions.
This motivates us to design a pricing policy $\mathcal{P}$ that iteratively learns how to monetize the edge computing service without relying on the market information.
Later on, we will choose the \textit{ex post optimal revenue} as our benchmark and compare the revenue generated by the pricing policy $\mathcal{P}$ with the \textit{ex post optimal revenue}.
Basically, the \textit{ex post optimal revenue} is given by
\begin{equation}\label{Equ: ESP Revenue opt}
	V^{\star}_{\text{ESP}}
	\triangleq
	\max\limits_{p\ge p_{\textit{min}}}
	\textstyle\sum\limits_{t=1}^{T}\sum\limits_{n=1}^{N}p\cdot\cycle_{n,t}\cdot z_{n,t}(p),
\end{equation}
where $z_{n,t}(p)$ represents the MU $n$'s executing decision given the price $p$.
In addition, $p_{\textit{min}}>0$ corresponds to the potential minimum price in the practical market.

Next we introduce the proposed pricing policy in Section \ref{Subsection: Pricing Policy}.
We then present the performance in Section \ref{Subsection: Performance Analysis}.

\subsection{Pricing Policy}\label{Subsection: Pricing Policy}
We now introduce our proposed pricing policy $\mathcal{P}$ based on the following three aspects.
\subsubsection{Price Discretization}
We consider a set $\mathcal{K}=\{1,2,...,K\}$ of price candidates, consisting of all powers of $1+\epsilon$ between the minimal price $p_{\textit{min}}$ and the potential maximal price $\bar{E}$ (defined in (\ref{Equ: E_bar})).
Mathematically, we denote the $k$-th price candidate as $p(k)=p_{\textit{min}}(1+\epsilon)^{k}$ for any $k\in\mathcal{K}$.
Accordingly, we have $K=\lfloor\log_{1+\epsilon}{\bar{E}}/{p_{\textit{min}}}\rfloor$, where $\lfloor\cdot\rfloor$ is the floor function.
It is obvious that the parameter $\epsilon>0$ affects both the number of price candidates and the performance of pricing policy $\mathcal{P}$.
We will discuss it in Section \ref{Subsection: Performance Analysis}.

\subsubsection{Basic Idea}
The basic idea of the pricing policy $\mathcal{P}$ is to iteratively exploit and explore the price candidates across the MU population $\mathcal{N}$ within the operation period $\mathcal{T}$.
More specifically, the pricing policy $\mathcal{P}$ maintains (in each slot $t$) a weight vector $\bm{\omega}_t=[\omega_{t}(1),\omega_{t}(2),...,\omega_{t}(K)]$ for all of the price candidates.
The weight of the $k$-th candidate in slot $t$, denoted by $\omega_{t}(k)$, is positively related to its previous performance (in terms of the generated revenue).
Overall, policy $\mathcal{P}$ tends to choose the price candidate with good performance (i.e., exploitation) and keep an eye on the other candidates that may perform better in the future (i.e., exploration).

\subsubsection{Policy Description}
Algorithm \ref{Algorithm: ESP} describes the pricing policy $\mathcal{P}$.
We elaborate it in details as follows.

\textit{Line 4:}
The pricing policy $\mathcal{P}$ will offer MU $n$ edge service in price $p(\kappa_{n,t})$ in slot $t$ by randomly selecting $\kappa_{n,t}$ according to the probability distribution $h_{t}(k)$, defined in (\ref{Equ: candidate probability}).
Specifically, the probability distribution $h_{t}(k)$ is the combination of two parts tuned by a parameter $\gamma\in(0,1)$.
The first part (with the coefficient $1-\gamma$), i.e., ${\omega_{t}(k)}/{\sum_{k=1}^{K}\omega_{t}(k)}$, represents the exploitation on the good-perform candidates.
Note that a larger weight corresponds to a higher probability to be selected.
The second part (with the coefficient $\gamma$), i.e., ${(1+\epsilon)^{k}}/{\sum_{i=1}^{K}(1+\epsilon)^{i}}$, represents the exploration over all candidates.
Moreover, the exploration scheme here pays more attention to the higher price candidates instead of the uniform exploration.

\textit{Line 5:}
ESP will receive the payment $V_{t,n}(\kappa_{t,n})$ from MU $n$ at the end of slot $t$.
Accordingly, we denote the total revenue generated by the pricing policy $\mathcal{P}$ as
\begin{equation}
	V_{\text{ESP}}(\mathcal{P})
	\triangleq\textstyle
	\sum\limits_{n=1}^{T}\sum\limits_{t=1}^{T} V_{t,n}(\kappa_{t,n}),
\end{equation}
where $\kappa_{t,n}$ is a random variable due to the randomness in Line 4 of Algorithm \ref{Algorithm: ESP}.
In later performance analysis, we will focus on the expected revenue, denoted by $\mathbb{E}\left[V_{\text{ESP}}(\mathcal{P})\right]$.

\textit{Lines 6-9:}
The pricing policy $\mathcal{P}$ updates the weight vector for the next slot based on the revenue in the current slot.
Specifically, in Line 7, we compute the actual revenue generated by each price candidate under the random vector $\bm{\kappa}_{t}=(\kappa_{t,n}:n\in\mathcal{N})$ in slot $t$, denoted by $V_{t}(k,\bm{\kappa}_{t})$.
In Line 8, we compute the virtual candidate revenue $\hat{V}_{t}(k,\bm{\kappa}_{t})$ by appropriately normalizing $V_{t}(k,\bm{\kappa})$.
Eventually in Line 9, we update the weight $\omega_{t+1}(k)$ by multiplying the current weight $\omega_{t}(k)$ by an exponential expression with the base $1+\delta$ and the exponent $\hat{V}_{t}(k,\bm{\kappa})$.

So far, we have introduced the pricing policy $\mathcal{P}$ with parameters $(\epsilon,\gamma,\delta)$.
Next we focus on its performance.

\begin{algorithm}[t]
	\caption{ESP's Dynamic Pricing Policy $\mathcal{P}$}
	\label{Algorithm: ESP}
	\SetKwInOut{Input}{Input}
	\SetKwInOut{Output}{Output}
	\textbf{Initial} $(\epsilon,\gamma,\delta)$ and the weight $\omega_{1}(k)=1, \forall k\in\mathcal{K}$.\\
	\For {$t=1$ \KwTo $T$ }
	{
		\For {MU $n\in\mathcal{N}$ }
		{	
			Offer MU $n$  the edge service in price $p(\kappa_{t,n})$, where $\kappa_{t,n}\in\mathcal{K}$ is randomly drawn based on
			\begin{equation}\label{Equ: candidate probability}
				h_{t}(k)
				\triangleq
				\textstyle
				\frac{(1-\gamma)\omega_{t}(k)}{\sum\limits_{k=1}^{K}\omega_{t}(k)} + \frac{\gamma\cdot(1+\epsilon)^{k}}{\sum\limits_{i=1}^{K}(1+\epsilon)^{i}},\ \forall k\in\mathcal{K}.
			\end{equation}
			\\
			Receive MU $n$'s payment $V_{t,n}(\kappa_{t,n})$ as follows
			\begin{equation}
				V_{t,n}(\kappa_{t,n})=p(\kappa_{t,n})\cdot \cycle_{n,t}\cdot z_{n,t}\big( p(\kappa_{t,n}) \big).
			\end{equation}\\
		}
		\For {$k\in\mathcal{K}$ }
		{
			\textbf{\textit{Compute}} candidate revenue $V_{t}(k,\bm{\kappa}_{t})$
			\begin{equation}\label{Equ: candidate revenue}
				V_{t}(k,\bm{\kappa}_{t})
				\triangleq  \textstyle
				\sum_{n\in\mathcal{N}} V_{t,n}(\kappa_{t,n}) \cdot \mathbf{1}_{\{\kappa_{t,n}=k\}}.
			\end{equation} 		\\
			\textbf{\textit{Compute}} virtual candidate revenue $\hat{V}_{t}(k,\bm{\kappa}_{t})$
			\begin{equation}\label{Equ: virtual candidate revenue}
				\hat{V}_{t}(k,\bm{\kappa}_{t})
				\triangleq\textstyle
				\frac{V_{t}(k,\bm{\kappa}_{t})}{ N\bar{\cycle} p_{\textit{min}} } \cdot \frac{\gamma}{h_{t}(k)\sum\limits_{i=1}^{K}(1+\epsilon)^{i}}.
			\end{equation}\\
			
			\textbf{\textit{Update}} weight $\omega_{t+1}(k)$
			\begin{equation}\label{Equ: update weight}
				\omega_{t+1}(k)
				\triangleq
				\omega_{t}(k)\cdot (1+\delta)^{ \hat{V}_{t}(k,\bm{\kappa}_{t}) }.
			\end{equation}
		}
	}
\end{algorithm}

\subsection{Performance Analysis}\label{Subsection: Performance Analysis}
We measure the performance of the pricing policy $\mathcal{P}$ in terms of the expected revenue $\mathbb{E}\left[V_{\text{ESP}}(\mathcal{P})\right]$.
Specifically, we will compare it with $V_{\text{ESP}}^{\star}$ defined in (\ref{Equ: ESP Revenue opt}), i.e., the revenue achieved by the optimal fixed pricing in hindsight.
Overall, we will demonstrate the performance of the following form:
\begin{equation}\label{Equ: performance form}
	\mathbb{E}\left[V_{\text{ESP}}(\mathcal{P})\right] \ge \frac{V_{\text{ESP}}^{\star}}{\alpha} - \mathcal{O}\left(N\bar{E}\bar{\cycle}\ln\left( \ln\left( \textstyle\frac{\bar{E}}{p_{\textit{min}}} \right)\right) \right),
\end{equation}
which means that the pricing policy $\mathcal{P}$ achieves a constant fraction of the ex post optimal revenue $V_{\text{ESP}}^{\star}$ with an extra loss term (that does not depend on the operation period length $T$).
The above performance structure is given in Theorem \ref{Theorem: ESP} based on Lemma \ref{Lemma: discrete loss} and Lemma \ref{Lemma: EXP3}.
All proofs are given in  \cite{report}. 
Before presenting the results, for notation simplicity, we denote ESP's revenue under the $k$-th price candidate as follows:
\begin{equation}
	V^T_{\text{ESP}}(k)
	=\textstyle
	\sum\limits_{n=1}^{T}\sum\limits_{t=1}^{T} p(k)\cdot \cycle_{n,t} \cdot z_{n,t}^{*}\big( p(k) \big),\quad\forall k\in\mathcal{K}.
\end{equation}
It is obvious that we have $V_{\text{ESP}}^{\star}\ge \max_{k\in\mathcal{K}}V^T_{\text{ESP}}(k) $ due to the price discretization.
Nevertheless, Lemma \ref{Lemma: discrete loss} shows that rounding down to a power of $1+\epsilon$ will reduce ESP's revenue at most a factor of $1+\epsilon$.

\begin{lemma}\label{Lemma: discrete loss}
	There exists a price candidate ${\kappa}\in\mathcal{K}$ such that the ESP's revenue under the price $p(\kappa)=p_{\textit{min}}(1+\epsilon)^{\kappa}$ satisfies
	\begin{equation}
		V^T_{\text{ESP}}(\kappa) \ge \frac{ V_{\text{ESP}}^{\star} }{1+\epsilon}.
	\end{equation}
\end{lemma}

Lemma \ref{Lemma: discrete loss} builds up the relationship between the revenue under the discrete price candidates and the optimal revenue under the continuous pricing space, depending on the parameter $\epsilon$.
Lemma \ref{Lemma: EXP3} further builds up the relationship between the expected revenue under the pricing policy $\mathcal{P}$ and the revenue under the discrete price candidate.
\begin{lemma}\label{Lemma: EXP3}
	With the parameters $(\epsilon,\delta,\gamma)$, the pricing policy $\mathcal{P}$ described in Algorithm \ref{Algorithm: ESP} can achieve
	\begin{equation}\label{Equ: EXP3 bound}
		\begin{aligned}\textstyle
			\mathbb{E}\left[ V^{T}_{\text{ESP}}(\mathcal{P}) \right]
			\ge
			{(1-\gamma)\left(1-\frac{\delta}{2}\right)} V^T_{\text{ESP}}(k)
			-\Phi(\epsilon,\delta,\gamma),\forall k\in\mathcal{K},
		\end{aligned}	
	\end{equation}
	where $\Phi(\epsilon,\delta,\gamma)$ is a constant and given by
	\begin{equation}
		\Phi(\epsilon,\delta,\gamma)
		=
		\frac{1-\gamma}{\gamma}\cdot \frac{1+\epsilon}{\epsilon} \cdot \frac{N\bar{E}\bar{\cycle}}{\delta} \cdot \ln\left( \frac{ \ln \left( {\bar{E}}/{p_{\textit{min}}} \right) }{ \ln(1+\epsilon) } \right).
	\end{equation}
\end{lemma}

Lemma \ref{Lemma: EXP3} indicates that the pricing policy $\mathcal{P}$ can achieve at least a constant fraction of the revenue under arbitrary price candidate with an additional loss term.
Note that the loss term $\Phi(\epsilon,\delta,\gamma)$ in (\ref{Equ: EXP3 bound}) does not depend on the time slot number $T$.
This means that $\Phi(\epsilon,\delta,\gamma)$ is increasingly negligible as $T$ increases.
Furthermore, we  obtain Theorem \ref{Theorem: ESP} by combining Lemma \ref{Lemma: discrete loss} and Lemma \ref{Lemma: EXP3}.
\begin{theorem}
	\label{Theorem: ESP}
	With the parameters $(\epsilon,\delta,\gamma)$, the pricing policy $\mathcal{P}$ described in Algorithm \ref{Algorithm: ESP} can achieve
	\begin{equation}
		\begin{aligned}
			\mathbb{E}\left[ V_{\text{ESP}}^{T}(\mathcal{P}) \right]
			\ge
			\frac{(1-\gamma)\left(1-\frac{\delta}{2}\right)}{1+\epsilon} V^{\star}_{\text{ESP}}
			-\Phi(\epsilon,\delta,\gamma) .
		\end{aligned}	
	\end{equation}
\end{theorem}

Theorem \ref{Theorem: ESP} presents the performance form mentioned in (\ref{Equ: performance form}).
To have a better understanding on Theorem \ref{Theorem: ESP}, we further provide the following corollary.
\begin{corollary}\label{Corollary: ESP alpha}
	If the optimal revenue ESP satisfies $V^{\star}_{\text{ESP}}\ge \frac{8}{\alpha}\Phi\left(\frac{\alpha}{3},\frac{\alpha}{6},\frac{\alpha}{12}\right)$ for some constant $\alpha\in(0,1]$, then the pricing policy $\mathcal{P}$ with the parameters $(\epsilon,\delta,\gamma)=\left(\frac{\alpha}{3},\frac{\alpha}{6},\frac{\alpha}{12}\right)$ can achieve
	\begin{equation}
		\mathbb{E}\left[ V_{\text{ESP}}^{T}(\mathcal{P}) \right]
		\ge
		\frac{ V^{\star}_{\text{ESP}} }{1+\alpha}  .
	\end{equation}	
\end{corollary}

Corollary \ref{Corollary: ESP alpha} shows that under some mild condition the pricing policy $\mathcal{P}$ with appropriate parameters is $(1+\alpha)$-competitive.
Moreover, the condition corresponds to a lower bound of the \textit{ex post optimal revenue} $V^{\star}_{\text{ESP}}$.
Although the lower bound $\frac{8}{\alpha}\cdot\Phi\left(\frac{\alpha}{3},\frac{\alpha}{6},\frac{\alpha}{12}\right)$ increases in $\alpha$, it does not scale in $T$.
This means that the condition will hold as ESP's the edge service monetization goes on.

\begin{figure*}[t]
	\setlength{\abovecaptionskip}{0pt}
	\setlength{\belowcaptionskip}{0pt}
	\centering	
	\subfigure[Impact of edge service price $p$.]
	{\label{fig: User_p}\includegraphics[width=0.27 \linewidth]{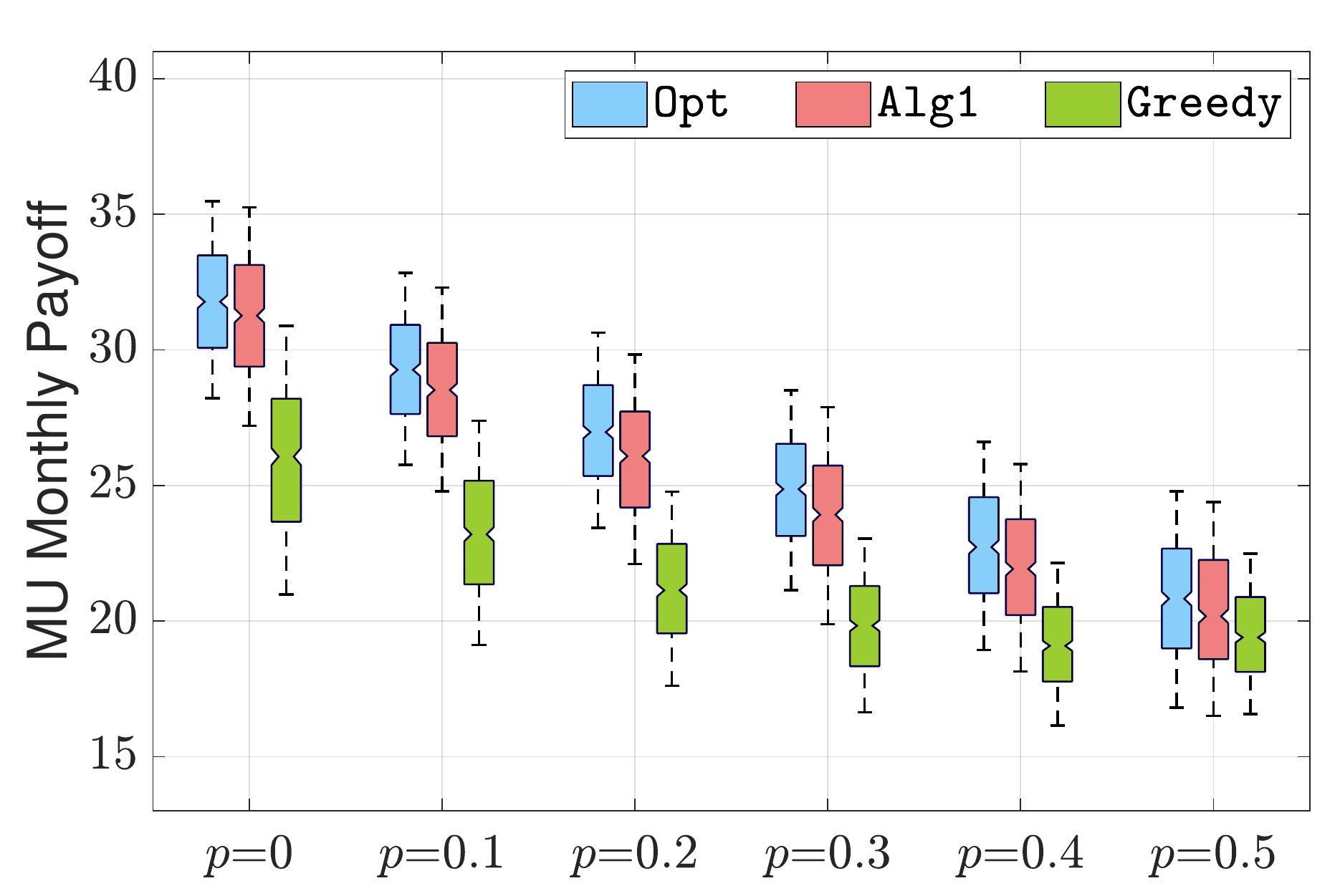}}\qquad
	\subfigure[Impact of monthly data cap $\dcap$.]
	{\label{fig: User_Q}\includegraphics[width=0.27 \linewidth]{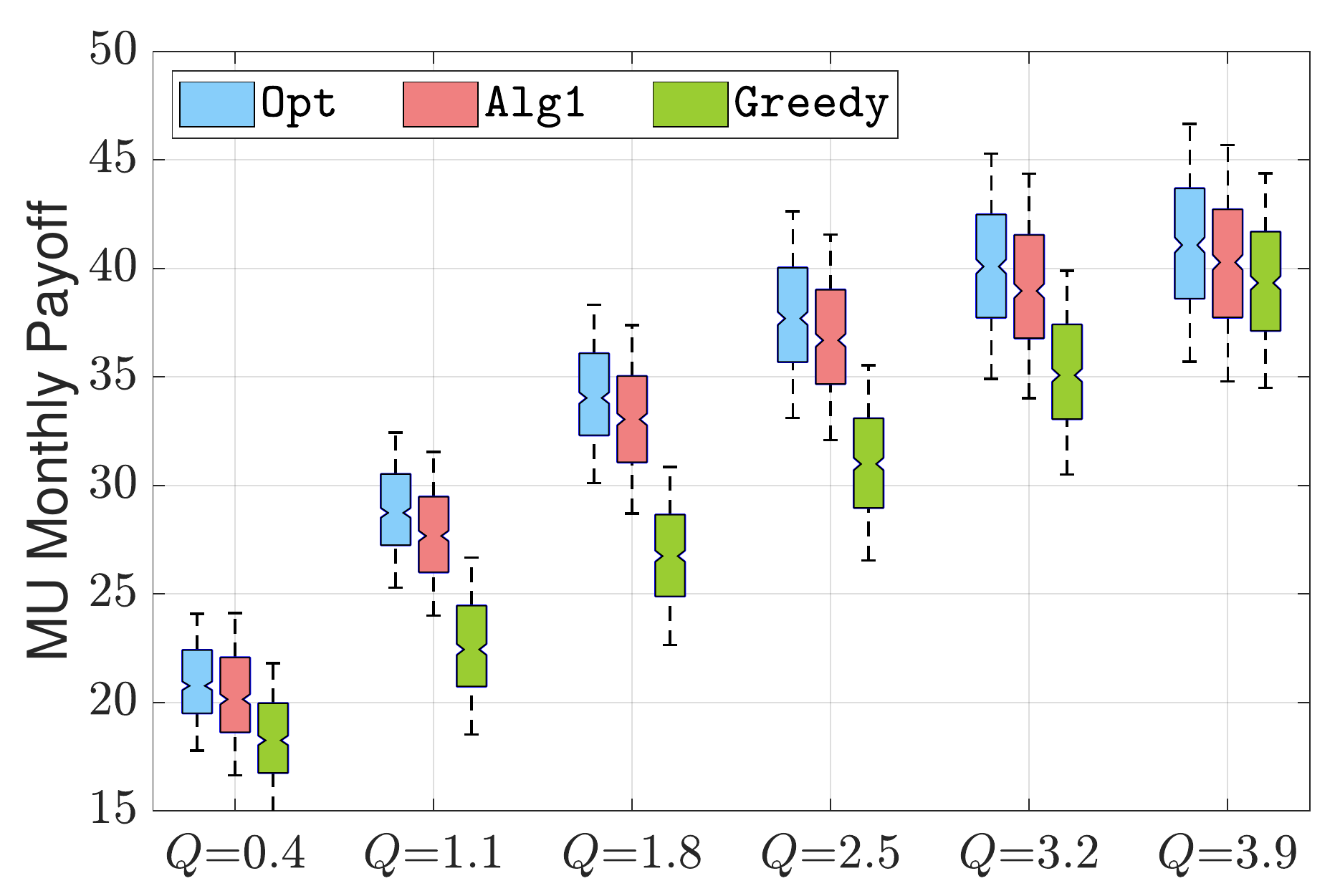}}\qquad
	\subfigure[Impact of overage fee $\adfee$.]
	{\label{fig: User_pi}\includegraphics[width=0.27 \linewidth]{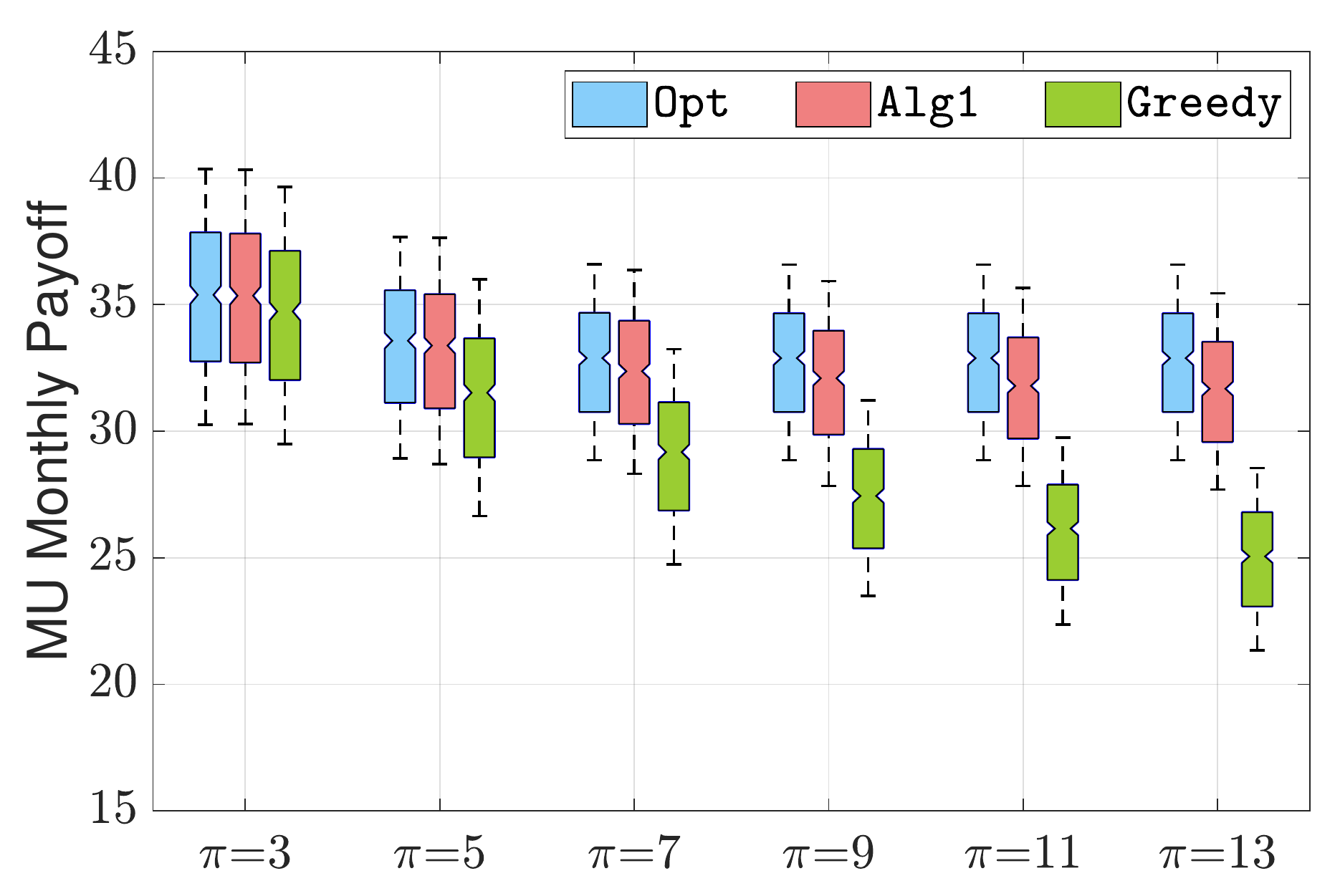}}
	\caption{Performance evaluation on the MU's  strategy $\mathcal{A}$.}
	\label{fig: Performance user}
\vspace{-5mm}
\end{figure*}

\section{Numerical Results}\label{Section: Numerical}
We carry out extensive evaluations on the mobile Internet ecosystem based on our previous analysis.
We will quantify the MU's online strategy $\mathcal{A}$ and ESP's pricing policy $\mathcal{P}$ in Section \ref{Subsection: Evaluation on MU's JOCP Strategy} and Section \ref{Subsection: Evaluation on ESP's Pricing Policy}, respectively.
We then evaluate the economic impact of the edge service on the entire ecosystem in Section \ref{Subsection: Economic Impact of Edge Service}.

\subsection{Evaluation on MU's JOCP Strategy}\label{Subsection: Evaluation on MU's JOCP Strategy}
We will evaluate the performance of the MU's strategy $\mathcal{A}$ in terms of the monthly payoff.
We first introduce the evaluation set-up, and then demonstrate the results.

\textit{Evaluation Set-Up:}
Recall that each MU is associated with the random content service model.
We specify the MU's content service model by randomly generating the three parameters $\data_t$, $\indata_{t}$, and $\cycle_t$ according to truncated normal distributions on their supports $[0,\bar{\data}]$, $[0,\bar{\cycle}]$, and $[0,\bar{\indata}]$, respectively.
In addition, the preference of MU is characterized based on the satisfaction and dissatisfaction.
We adopt the well-known alpha-fair utility to represent the MU's satisfaction, i.e., $u(x)=\frac{x^{1-a}}{1-a}$.
Moreover, we use a quadratic function to model MU's cost of local execution, i.e., $e(s)={s^2}/{2}$.
Furthermore, we randomly generate the MU's valuation parameter ${\val}_{t}$ and sensitivity parameter ${\cost}_{t}$ according to a truncated normal distribution with the range $[0,2]$ for normalization consideration, i.e., $\mathbb{E}[\val_{t}]=\mathbb{E}[\cost_{t}]=1$.

Under the above evaluation set-up, we will compare the MU's monthly payoff in the following three cases:
\begin{itemize}
	\item The case of $\mathtt{Opt}$ corresponds to the off-line optimal outcome discussed in Theorem \ref{Theorem: MU KKT}.
		
	\item The case of $\mathtt{Alg1}$ corresponds to the proposed online strategy $\mathcal{A}$ defined in Algorithm \ref{Algorithm: MU}.
	
	\item The case of $\mathtt{Greedy}$ corresponds to the greedy strategy that tends to maximize the daily payoff without taking into account the potential future over usage.\footnote{The greedy policy means that the MU is myopic. It is a reasonable benchmark for two reasons.
		First, it is easy to implement, since it does not requires that MUs should know the future information.
		Second, it also captures the bounded rationality behavior for MUs \cite{arthur1994inductive}.}
\end{itemize}

We visualize the MU's monthly payoff through box-plots in Fig. \ref{fig: Performance user}, where the three sub-figures investigate the impact of the monthly data cap, the edge service price, and the overage fee, respectively.

Fig. \ref{fig: User_p} shows how the edge service price $p$ affects the MU's monthly payoff in the aforementioned three cases.
For the visualization purpose, we set the same daily edge service price here.
Overall the MU's monthly payoff decreases in the edge service price, as the higher price means that it is less likely for MU to utilize the edge service to reduce the local execution cost.
The small gaps between the red boxes and the blue boxes indicate the good performance of the strategy $\mathcal{A}$, which achieves 95\% of the payoff in case $\mathtt{Opt}$.

Fig. \ref{fig: User_Q} shows that the MU's monthly payoff increases in his monthly data cap.
Comparing the blue boxes and the red boxes in Fig. \ref{fig: User_Q}, we find that the MU's payoff achieved by the strategy $\mathcal{A}$ is very close to the optimal payoff in hindsight.
Comparing the green boxes and the red boxes in Fig. \ref{fig: User_Q}, we find that the payoff gap between our strategy $\mathcal{A}$ and the greedy strategy differs in the monthly data cap.
The intuitions are as follows:
\begin{itemize}
	\item A small data cap (e.g., $\dcap=0.4$GB in Fig. \ref{fig: User_Q}) implies that the MU's wireless data demand will exceed the monthly data cap in most cases.
		Thus the shadow price updating in strategy $\mathcal{A}$ will quickly converge to $\hat{\lambda}_t=\adfee$, which is similar to the case of greedy strategy.
		Hence there is a little gap between $\mathtt{Alg1}$ and $\mathtt{Greedy}$ when the data cap is small.
	\item A large data cap (e.g., $\dcap=3.9$GB in Fig. \ref{fig: User_Q}) implies that the MU's wireless data demand is less than the monthly data cap in most cases.
		Thus the shadow price in strategy $\mathcal{A}$ will quickly converge to $\hat{\lambda}_t=0$, which is similar to the case of greedy strategy.
		Hence there is a little payoff gap between $\mathtt{Alg1}$ and $\mathtt{Greedy}$ when the data cap is small.
	\item A medium data cap (e.g., $\dcap=2.5$GB in Fig. \ref{fig: User_Q}) means that the MU's wireless data demand and the monthly data cap are comparable, thus there is a great uncertainty for the MU's cap-acquiring outcome.
		In this case, the greedy strategy fails in learning the optimal shadow price, thus the payoff gap between $\mathtt{Alg1}$ and $\mathtt{Greedy}$ is large.
		
\end{itemize}

Fig. \ref{fig: User_pi} plots how the overage fee $\adfee$ affects the MU's monthly payoff in the three cases.
Overall a larger overage fee corresponds to lower MU payoff.
However, comparing the green boxes and the red boxes in Fig. \ref{fig: User_pi}, we note that our strategy $\mathcal{A}$ is less sensitive to the overage fee than the greedy strategy.
This is because that our strategy $\mathcal{A}$ can iteratively learn the shadow price (based on MU's previous decisions), which helps the MU avoid great monetary cost for exceeding the monthly data cap.
While the greedy policy, aiming at the myopic benefit, cannot prevent it.
Moreover, the MU's monthly payoff under the strategy $\mathcal{A}$ is 93\% (on average) of the payoff in the case of $\mathtt{Opt}$.

\subsection{Evaluation on ESP's Pricing Policy}\label{Subsection: Evaluation on ESP's Pricing Policy}
We evaluate the performance of ESP's pricing policy $\mathcal{P}$ in terms of its total revenue from a group of MUs.
Specifically, we take into account a total of five hundred MUs with randomly generated parameters as in Section \ref{Subsection: Evaluation on MU's JOCP Strategy}.
Recall that the ESP's pricing policy $\mathcal{P}$ defined in Algorithm \ref{Algorithm: ESP} depends on three parameters (i.e., $\epsilon$, $\gamma$, and $\delta$), which will jointly affect the theoretic performance of the pricing policy $\mathcal{P}$.
In our evaluation, we set the three parameters according to Corollary \ref{Corollary: ESP alpha} with $\alpha=1$.
That is, the ESP's total revenue achieved by the pricing policy $\mathcal{P}$ is at least 50\% of the ex post optimal revenue if the condition in Corollary \ref{Corollary: ESP alpha} holds (which is true in our evaluation setup).
We run the evaluation for one hundred times and show the results in  Fig. \ref{fig: Performance ESP data} and Fig. \ref{fig: Performance ESP indata}.

\begin{figure*}
	\vspace{-3mm}
	\begin{minipage}{0.49\textwidth}
		\centering
		\setlength{\abovecaptionskip}{0pt}
		\setlength{\belowcaptionskip}{0pt}
		\subfigure[Average ESP revenue.]
		{\label{fig: ESPrevenues_data}\includegraphics[width=0.47\linewidth]{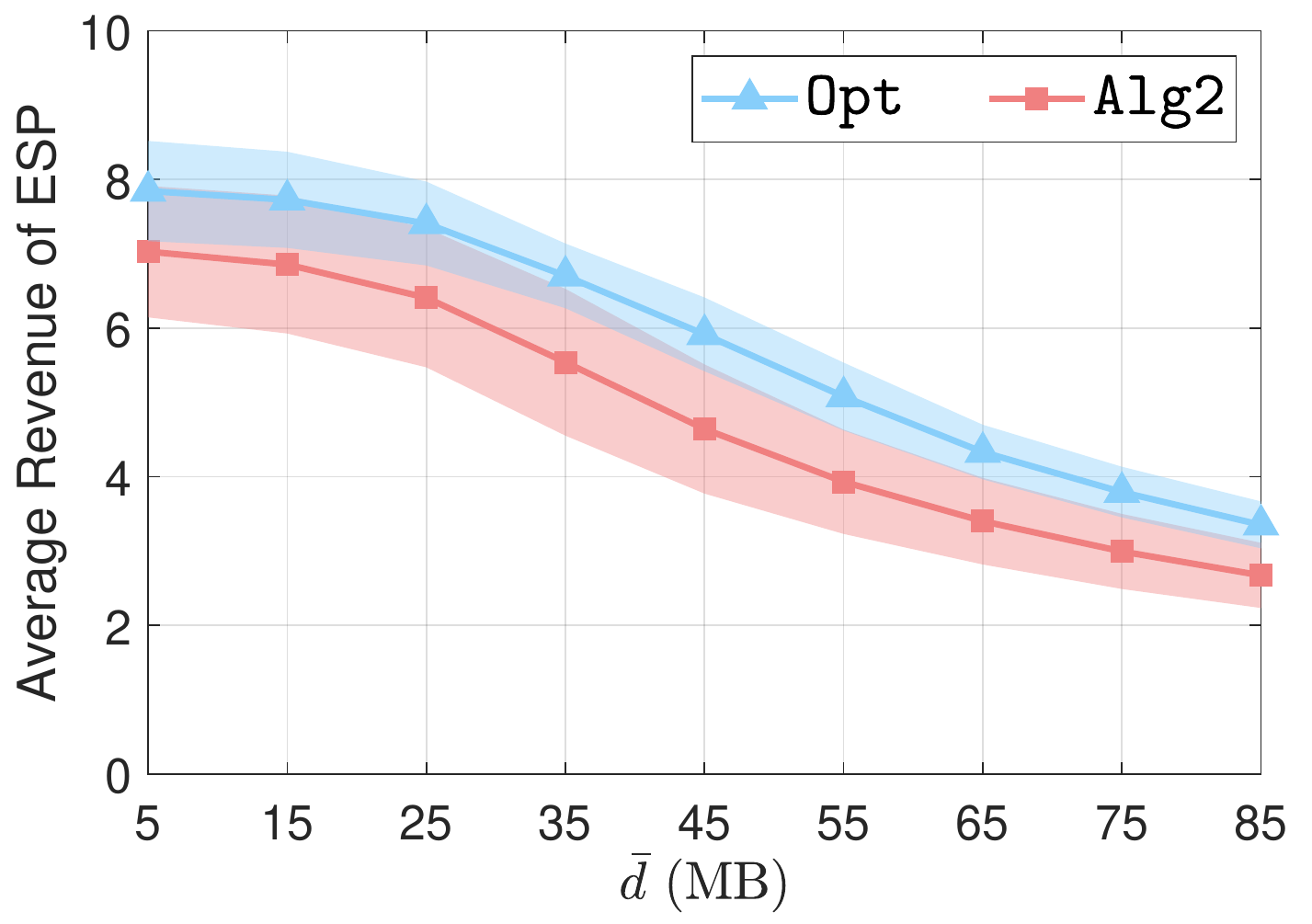}}\quad
		\subfigure[Ratio of the optimal revenue.]
		{\label{fig: ESPfraction_data}\includegraphics[width=0.47\linewidth]{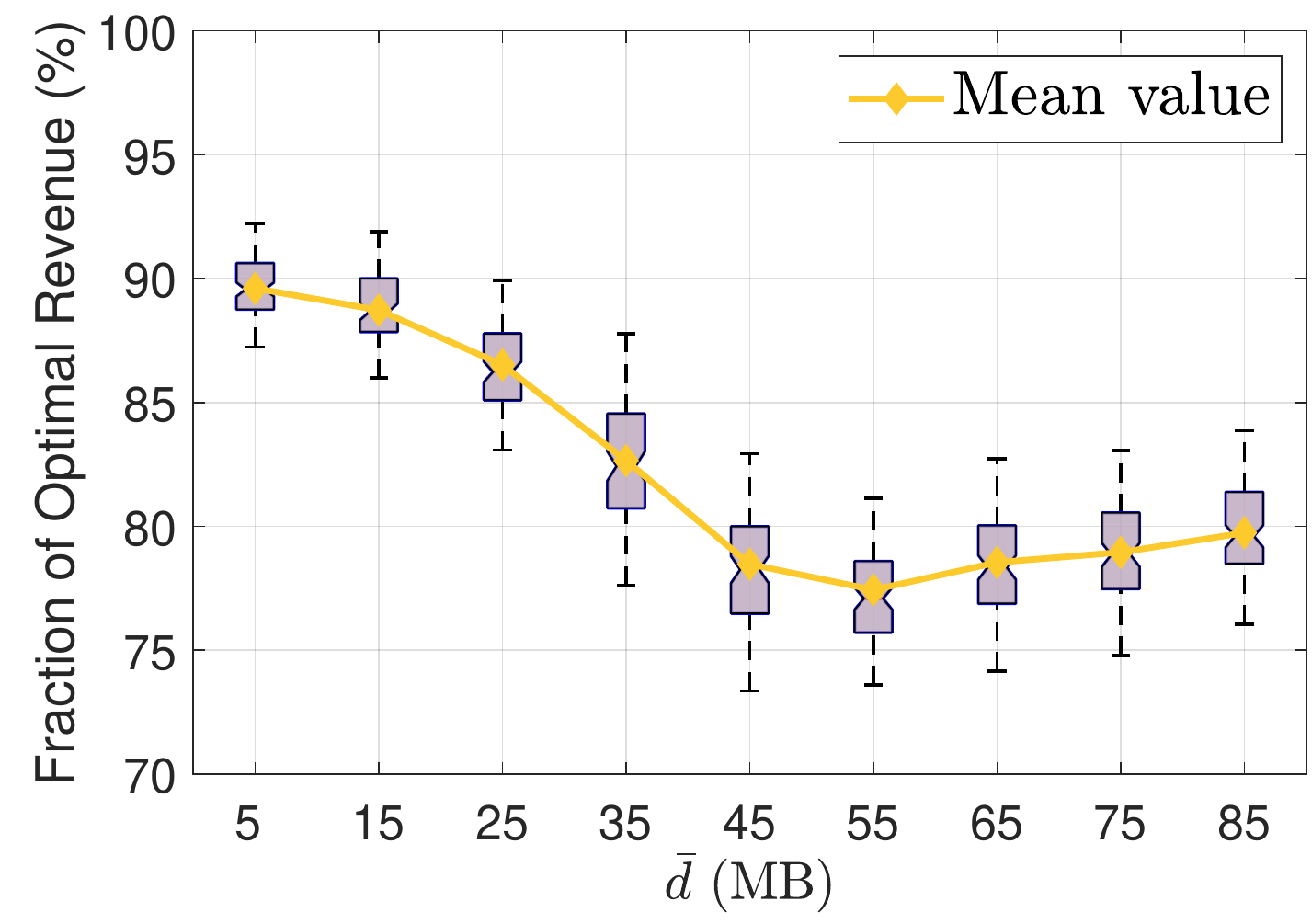}}
		\caption{ESP's revenue versus maximal usage level $\bar{\data}$.}
		\label{fig: Performance ESP data}
	\end{minipage}
	\begin{minipage}{0.49\textwidth}
		\setlength{\abovecaptionskip}{0pt}
		\setlength{\belowcaptionskip}{0pt}
		\centering
		\subfigure[Average ESP revenue.]
		{\label{fig: ESPrevenues_indata}\includegraphics[width=0.47\linewidth]{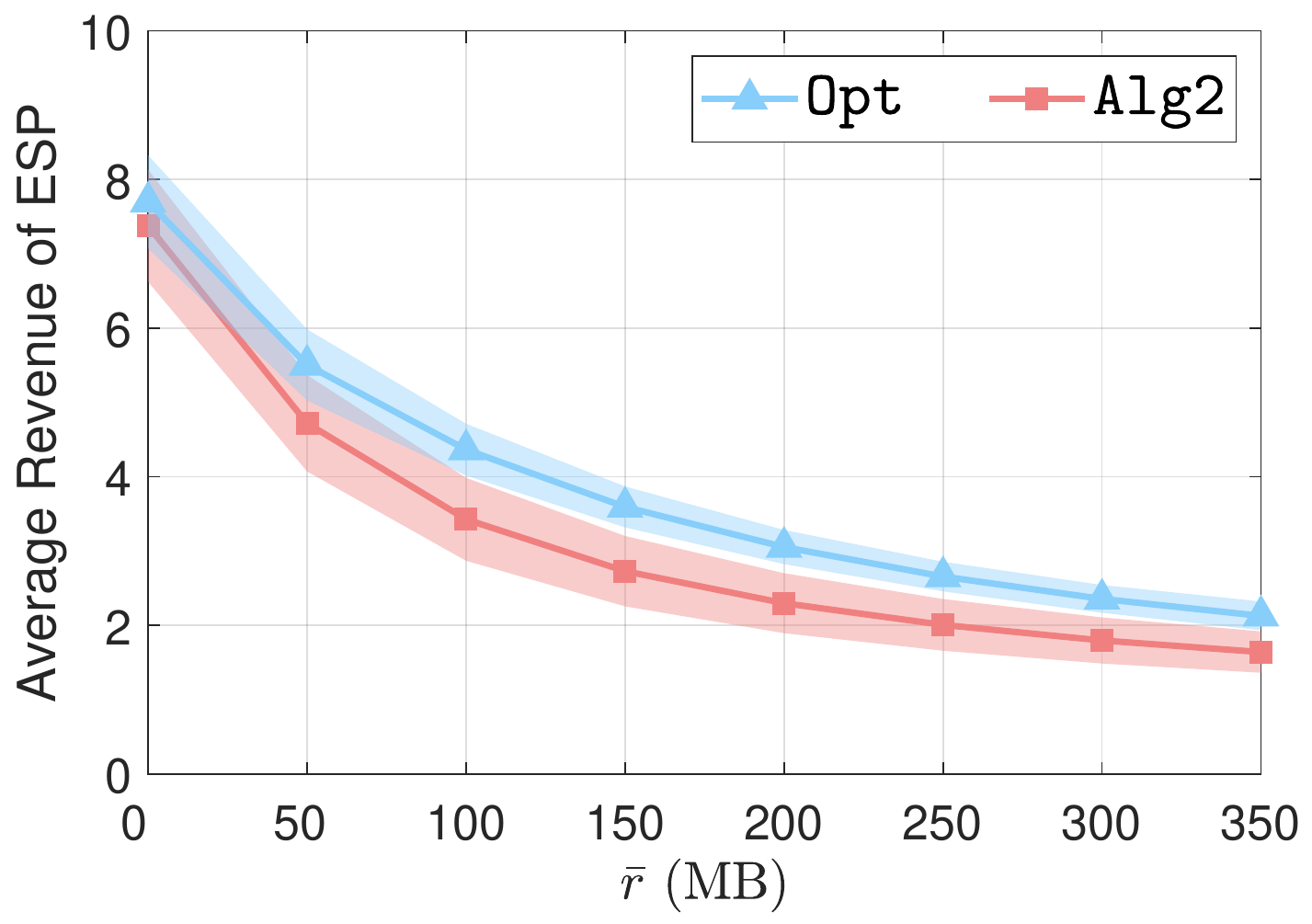}}\quad
		\subfigure[Ratio of the optimal revneue.]
		{\label{fig: ESPfraction_indata}\includegraphics[width=0.47\linewidth]{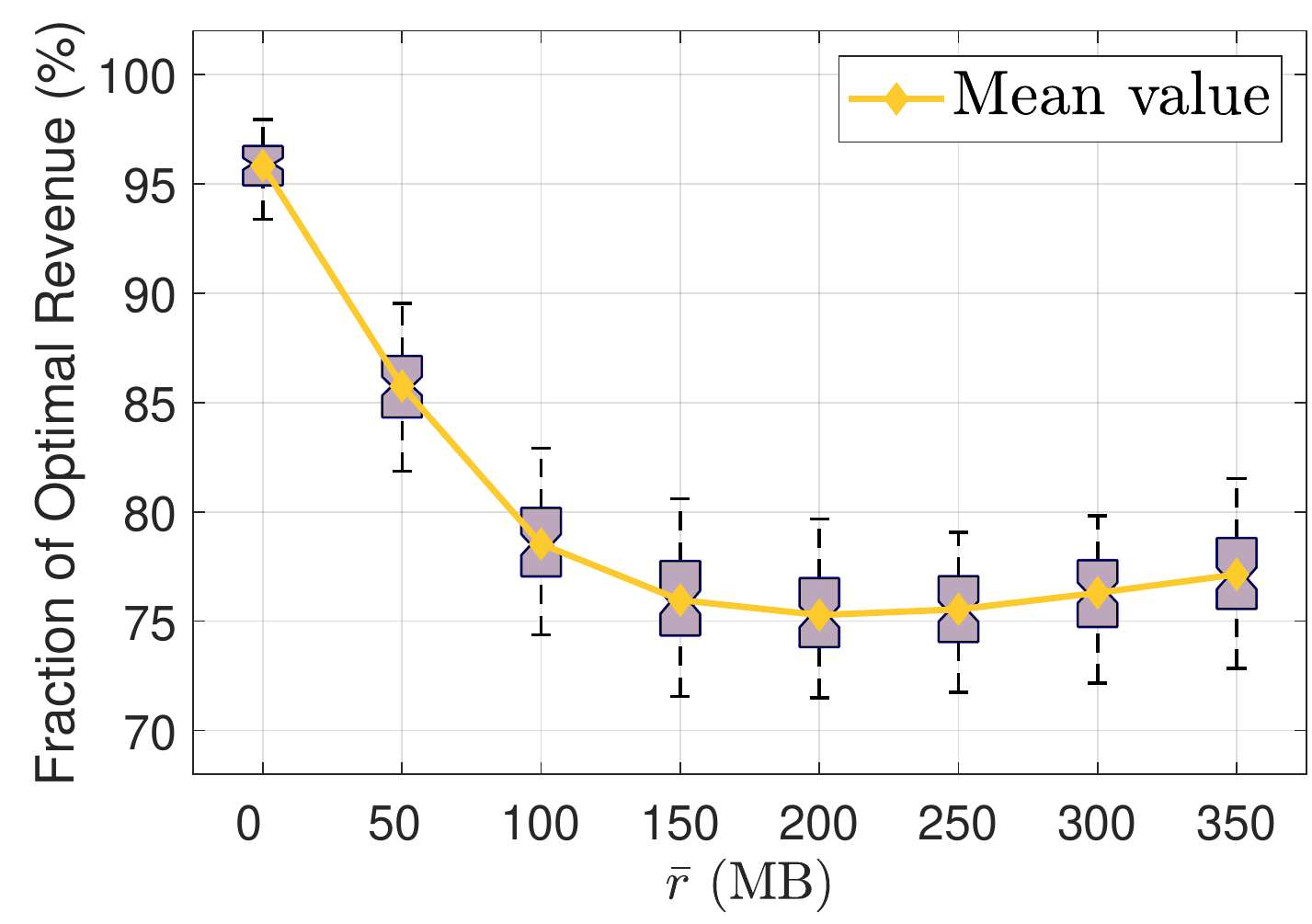}}
		\caption{ESP's average revenue versus maximal raw data $\bar{\indata}$.}
		\label{fig: Performance ESP indata}
	\end{minipage}
	\vspace{-3mm}
\end{figure*}

\begin{figure*}
	\centering
	\setlength{\abovecaptionskip}{0pt}
	\setlength{\belowcaptionskip}{0pt}
	\subfigure[MUs' payoffs.]
	{\label{fig: Impact_MU}\includegraphics[width=0.23\linewidth]{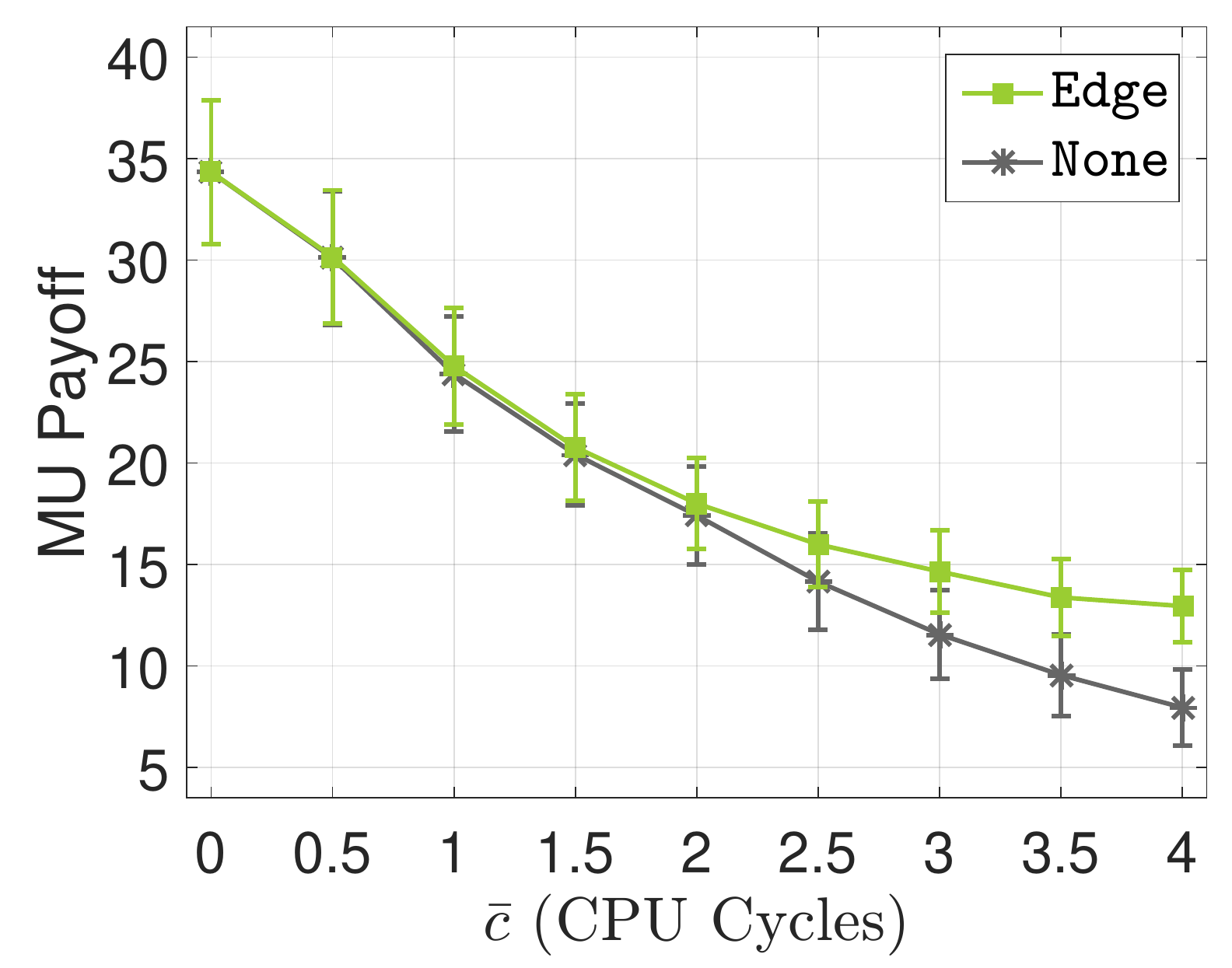}}\quad
	\subfigure[CP's revenue.]
	{\label{fig: Impact_CP}\includegraphics[width=0.23\linewidth]{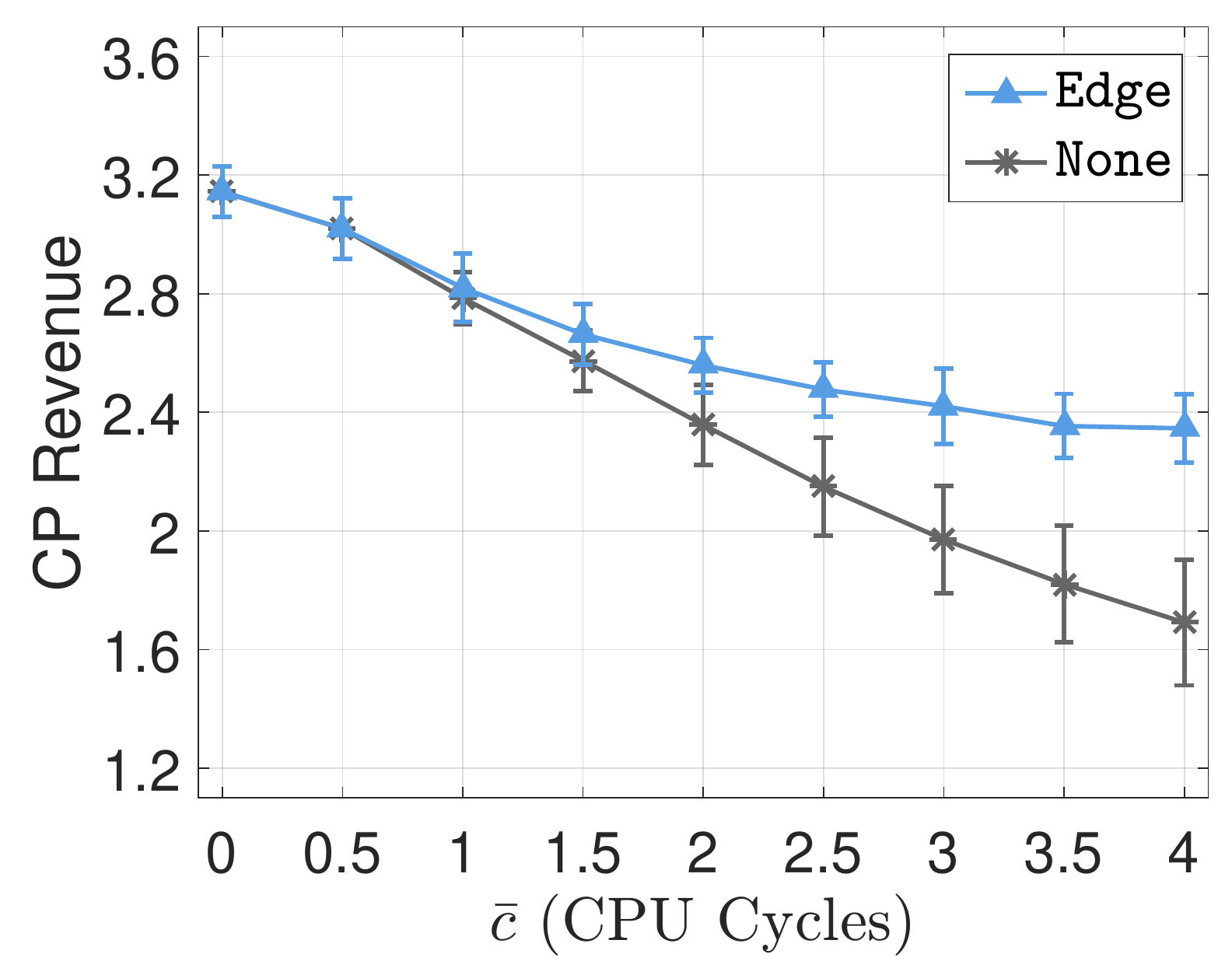}}\quad
	\subfigure[ISP's revenue.]
	{\label{fig: Impact_ISP}\includegraphics[width=0.23\linewidth]{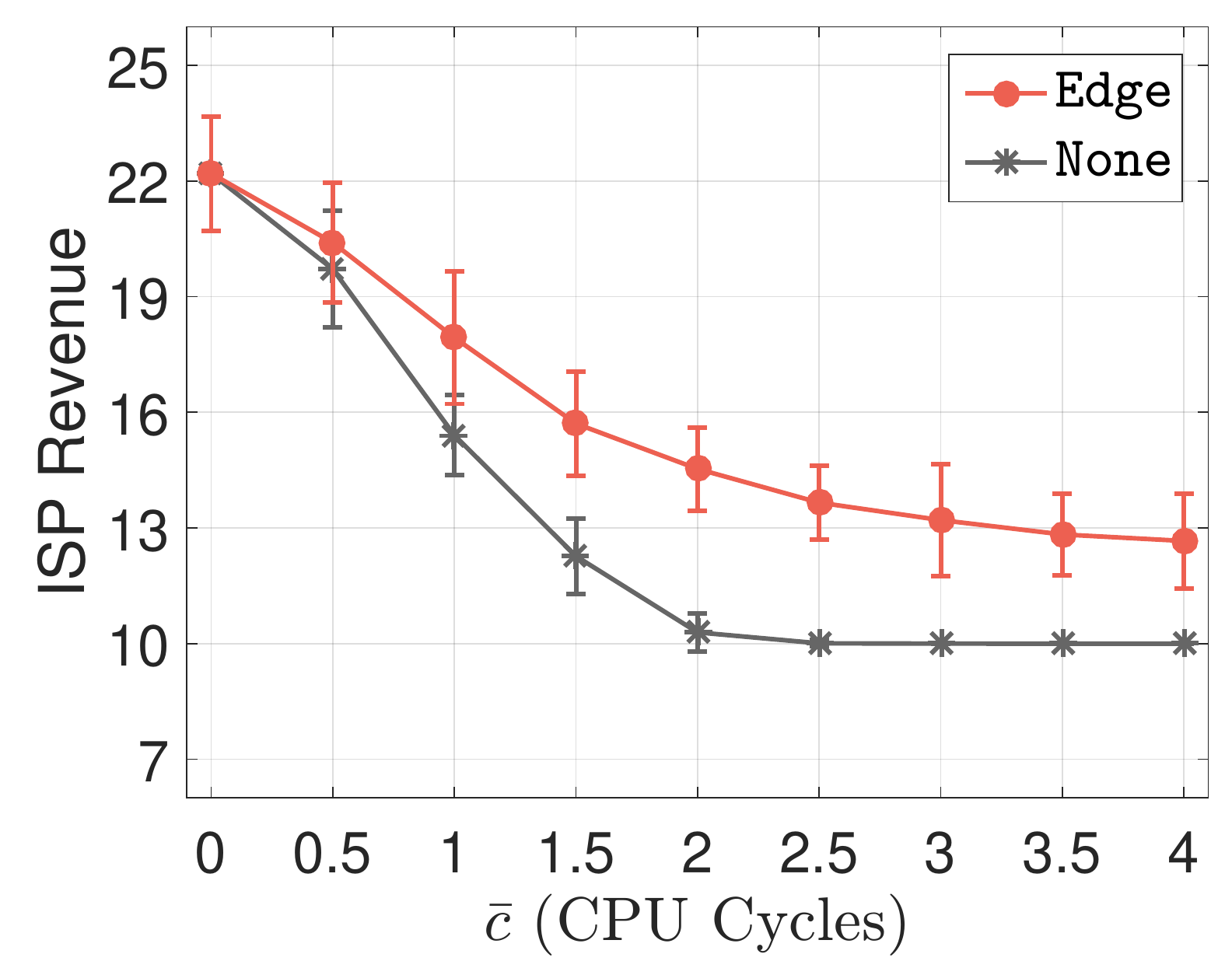}}\quad
	\subfigure[Improvement of edge service.]
	{\label{fig: Impact_Social}\includegraphics[width=0.23\linewidth]{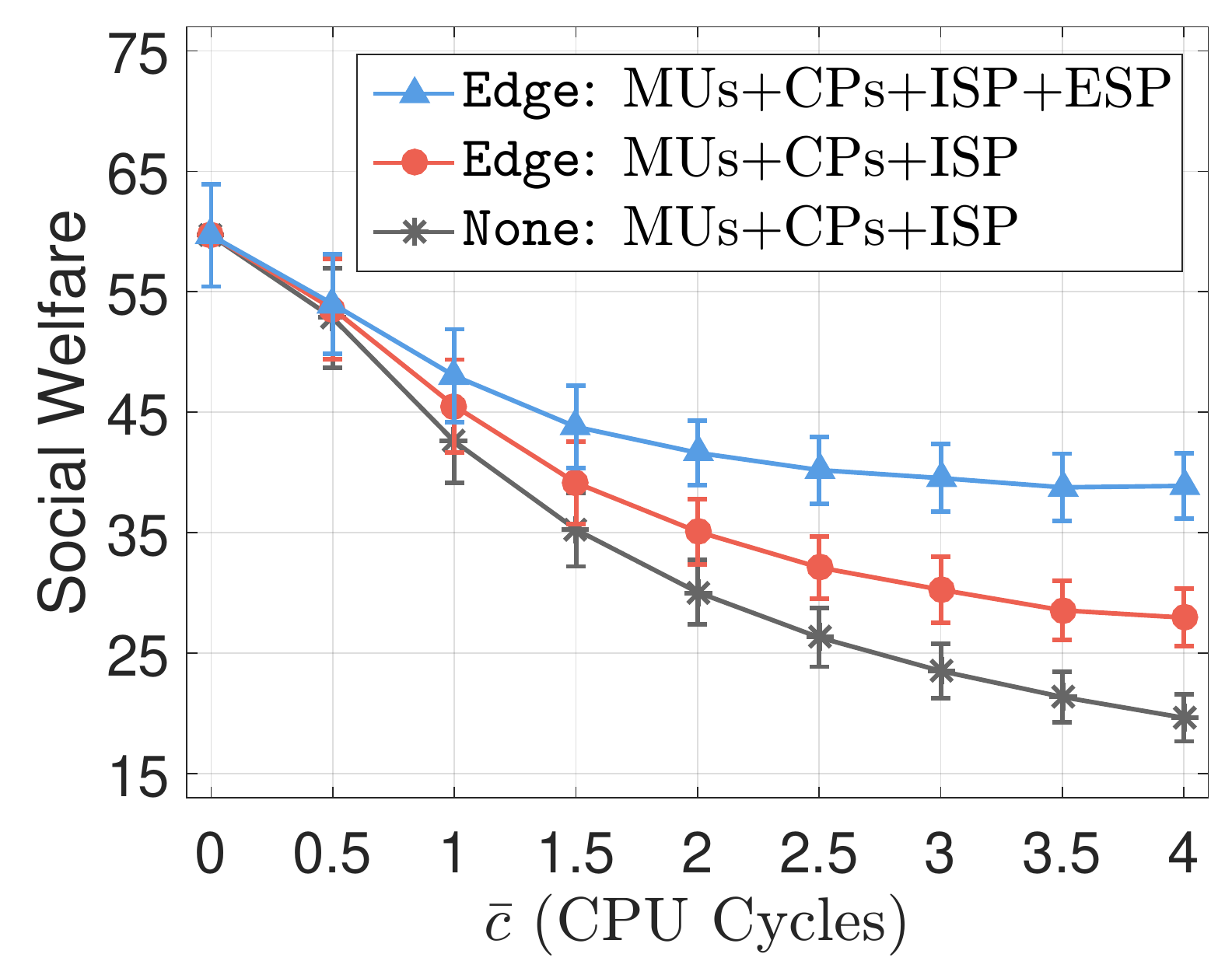}}
	\caption{Economic effect of edge service on the ecosystem.\vspace{-10pt}}
	\label{fig: Ecpnomic effect of edge service}
	\vspace{-3mm}
\end{figure*}

Fig. \ref{fig: Performance ESP data} plots the performance of the pricing policy $\mathcal{P}$ under different values of the maximal data-usage $\bar{\data}$.
We compare ESP's revenue under the pricing policy $\mathcal{P}$ (labeled by $\mathtt{Alg2}$) with the ex post optimal revenue (labeled by $\mathtt{Opt}$).
\begin{itemize}
	\item Fig. \ref{fig: ESPrevenues_data} plots the ESP's average revenue in the two cases.
		Specifically, the two curves with markers represent the average results (over multiple runs).
		The red and blue regions correspond to the three-sigma ranges.
		Overall ESP's revenue decreases in the maximal data usage level $\bar{\data}$.
		Intuitively, a larger $\bar{\data}$ value means that MUs have greater wireless data demand from the content delivery, thus less wireless data quota to offloading computation tasks.

	\item Fig. \ref{fig: ESPfraction_data} shows the ESP's revenue ratio of  case $\mathtt{Alg2}$ to case $\mathtt{Opt}$ through box-plot.
		The yellow curve represents the mean results over multiple runs.
		We note that the pricing policy $\mathcal{P}$ actually achieves (on average) a fraction 83\% of the ex post optimal revenue, which is much better than the theoretic lower bound (i.e., 50\% in our evaluation setup).
\end{itemize}

Fig. \ref{fig: Performance ESP indata} shows the impact of the maximal raw data amount $\bar{\indata}$ on the pricing policy $\mathcal{P}$.
Similarly, we compare the ESP's revenue under the pricing policy $\mathcal{P}$ (labeled by $\mathtt{Alg2}$) with the ex post optimal revenue (labeled by $\mathtt{Opt}$).
\begin{itemize}
	\item Fig. \ref{fig: ESPrevenues_indata} shows that the ESP's revenue decreases in the maximal raw data amount.
	This is because that a larger $\bar{\indata}$ value means a greater offloading cost for the MUs, leading to a lower demand on the edge service.
	\item Fig. \ref{fig: ESPfraction_indata} shows that the pricing policy $\mathcal{P}$ can achieve (on average) a fraction 79\% of the ex post optimal revenue, which is much better than the theoretic lower bound.
\end{itemize}

So far, we have illustrated the performance of the MU's strategy $\mathcal{A}$ and the ESP's pricing policy $\mathcal{P}$.
Next we further evaluate the economic effect of the edge service on the entire mobile Internet ecosystem.

\subsection{Economic Impact of Edge Service}\label{Subsection: Economic Impact of Edge Service}
We evaluate how the edge service affect the entire ecosystem, including MUs' payoffs, ISP's revenue, and CP's revenue.
Recall that the ISP's revenue, defined in (\ref{Equ: Revenue of ISP}), consists of the subscription fee and the overage fee.
In our evaluation, we consider the wireless data plan with data cap $\dcap=1$GB, subscription fee $\pcap=\$10$, and per-unit overage fee $\adfee=\$15$/GB.
Furthermore, the CPs' revenue, defined in (\ref{Equ: Revenue of CP}), mainly comes from displaying advertisements, thus is positively related to MUs' total content acquisitions.
In our evaluation, we follow \cite{joe2018sponsoring} and suppose that the advertising revenue takes the form $v(x)\triangleq x^{1-\tau}/(1-\tau)$.
Based on the above setup, we will compare the following two cases.
\begin{itemize}
	\item The case of $\mathtt{None}$ represents the classic mobile Internet ecosystem without edge service.
		MUs will execute the computation tasks locally.
	\item The case of $\mathtt{Edge}$ represents the ecosystem with ESP offering edge service based on the pricing policy $\mathcal{P}$.
\end{itemize}

Fig. \ref{fig: Ecpnomic effect of edge service} plots the economic impact of the edge service on the ecosystem.
The horizontal axis of the four sub-figures represents the maximal computation-intensity $\bar{\cycle}$.
Moreover, Fig. \ref{fig: Impact_MU}, Fig. \ref{fig: Impact_CP}, and Fig. \ref{fig: Impact_ISP} shows the average benefits (i.e., revenue or payoff) of ISP, CPs, and MUs, respectively.
Fig. \ref{fig: Impact_Social} shows the social welfare.

In Fig. \ref{fig: Impact_MU}, the black star curve and green square curve show the MUs' payoffs in case $\mathtt{None}$ and case $\mathtt{Edge}$, respectively.
Overall, the MUs' payoffs decrease in $\bar{\cycle}$.
This is because that the greater computation-intensity reduces MUs' content acquisitions due to the resource-limited devices and the costly edge service.
However, comparing the two curves in Fig. \ref{fig: Impact_MU}, we find that the edge service improves MUs' payoffs (up to 63\%), as it alleviates the local execution cost.

In Fig. \ref{fig: Impact_CP}, the black star curve and blue triangle curve show the revenue of CPs in case $\mathtt{None}$ and case $\mathtt{Edge}$, respectively.
As mentioned, the greater computation-intensity reduces MUs' content acquisitions, which eventually leads to the decreasing revenue for CPs.
Nevertheless, the edge service can stimulate the MUs' content acquisitions, which improves the revenue of CPs (up to 37\%).

In Fig. \ref{fig: Impact_ISP}, the black star curve and red circle curve show the revenue of ISP in case $\mathtt{None}$ and case $\mathtt{Edge}$, respectively.
ISP's revenue decreases in the computation-intensity $\bar{\cycle}$, as the greater computation-intensity reduces MUs' content acquisitions.
However, comparing the two curves in Fig. \ref{fig: Impact_ISP}, we find that the edge service increases the revenue of ISP (up to 40\%) as it stimulates MUs' content acquisitions.

Fig. \ref{fig: Impact_Social} compares the social welfares.
The black star curve corresponds to the case of $\mathtt{None}$.
The red circle curve represents the total welfare of MUs, CPs, and ISP under the case of $\mathtt{Edge}$.
The blue triangle curve plots the social welfare of MUs, CPs, ISP, and ESP under the case of $\mathtt{Edge}$.
Comparing the three curves, it is obvious that the edge service can significantly increase the social welfare of the mobile Internet ecosystem.

\section{Conclusion and Future Work}\label{Section: Conclusion}
This paper studies the economic interactions in the mobile Internet ecosystem.
Specifically, we investigate the MUs' Joint Content Acquisition and Task Offloading (J-CATO) problem and design an online strategy with provable performance based on the off-line insights.
Moreover, we propose for ESP an edge service pricing policy  that does not rely on any market information.
We find that the edge service helps the MUs reduce the local execution cost, leading to higher average MU payoff.
In addition, the edge service also stimulates MUs' content acquisitions, which increases the ISP's benefit from wireless data service and CPs' benefit from content service.
Therefore, the edge service leads to higher social benefit for the ecosystem.

As this is the first study on monetizing edge service, there are some open problems that deserve investigation.
\begin{itemize}
	\item It is necessary to study different business models.
		We view the edge service provider as a third-part.
		In practice, both the ISP and giant CP can deploy edge data center and monetize the edge computing service.
		This leads to different value chains together with different economic insights.
		
%
	\item It is also interesting to take into account the ESP's investment cost on the edge servers.
		In that case, the ESP faces a new issue, i.e., whether it can recoup the investment within the lifetime of the edge servers.
		
\end{itemize}

\section*{Acknowledgements}

This work is supported by the National Natural Science Foundation of
China (Grant No. 61972113 and 61801145),  Shenzhen Science and Technology
Program (Grant No. JCYJ20190806112215116,
JCYJ20180306171800589, and KQTD20190929172545139),  Guangdong Science and
Technology Planning Project under Grant
2018B030322004.
This work is also supported in part by the funding
from Shenzhen Institute of Artificial
Intelligence and Robotics for Society.

\bibliography{ref}
\bibliographystyle{IEEEtran}

\vspace{-30pt}
\begin{IEEEbiography}[{\includegraphics[width=1in,height=1.25in,clip,keepaspectratio]{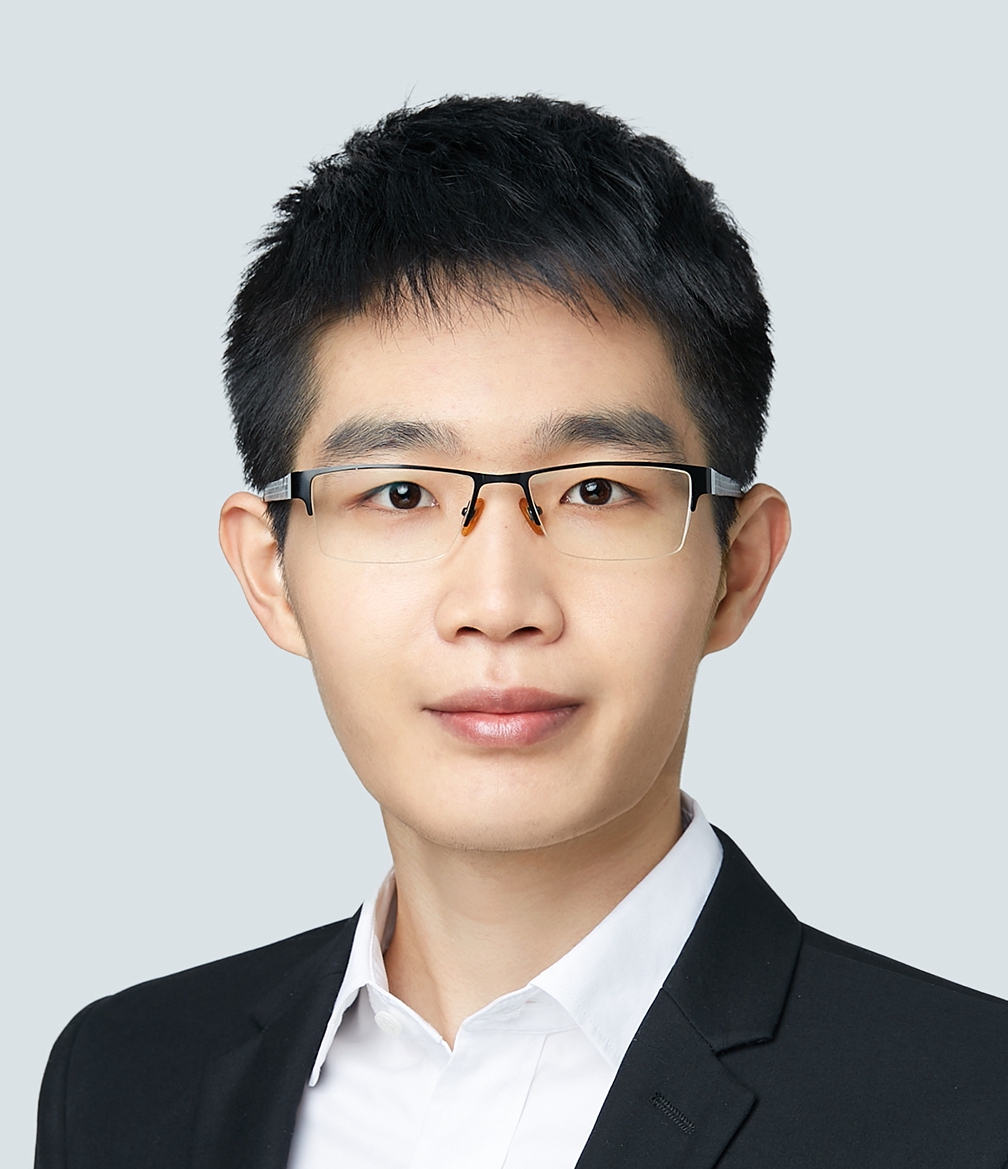}	}]{Zhiyuan Wang}
	is currently a Post-Doctoral Fellow in Department of Computer Science and Engineering, The Chinese University of Hong Kong.
	He received his Ph.D. degree in Information Engineering, from The Chinese University of Hong Kong, in 2019.
	He received the B.S. degree in Information Engineering, from Southeast University, Nanjing, in 2016.
	His research interests include the field of network science and game theory, with current emphasis on mean field analysis and edge computing.
\end{IEEEbiography}

\vspace{-30pt}
\begin{IEEEbiography}[{\includegraphics[width=1in,height=1.25in,clip,keepaspectratio]{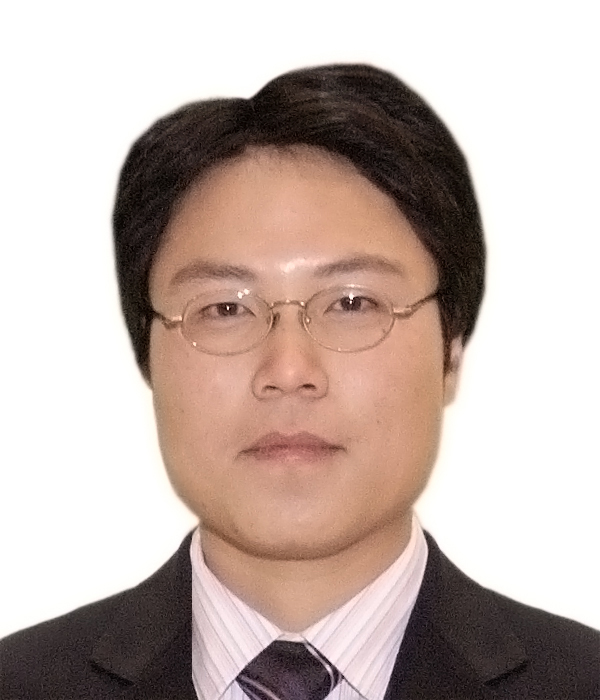}	}]	{Lin Gao}
	(S'08-M'10-SM'16) is an Associate Professor with the School of Electronic and Information Engineering, Harbin Institute of Technology, 	Shenzhen, China. He received the Ph.D. degree in Electronic Engineering from Shanghai Jiao Tong University in 2010. His main research
	interests are in the area of network economics 	and games, with applications in wireless communications and networking.
	He is the recipient of the 11th IEEE ComSoc AsiaPacific Outstanding Young Researcher Award in 2016.
\end{IEEEbiography}

\vspace{-30pt}
\begin{IEEEbiography}[{\includegraphics[width=1in,height=1.25in,clip,keepaspectratio]{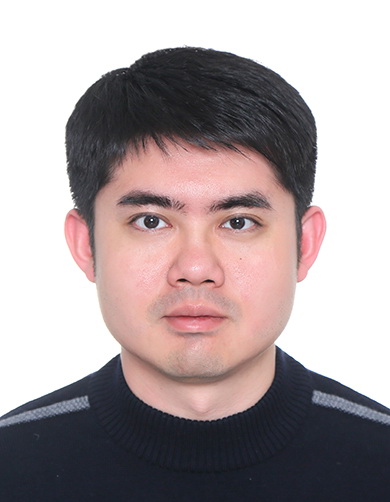}	}]	{Tong Wang}
(M'15) received the B.E. degree in electrical engineering and automation from
Beihang University, China, in 2006, and the M.S. degree in communications engineering and the Ph.D. degree in electronic engineering from The University of York, U.K., in 2008 and 2012, respectively.
From 2012 to 2015, he was a Research Associate with the Institute for Theoretical Information Technology, RWTH Aachen University, Aachen, Germany.
From 2014 to 2015, he was a Research Fellow of the Alexander von Humboldt Foundation.
He is an Assistant Professor in Department of Electronic and Information Engineering, Harbin Institute of Technology, Shenzhen, China. His research interests include sensor networks, cooperative communications, adaptive filtering, and resource optimization.
\end{IEEEbiography}

\vspace{-30pt}
\begin{IEEEbiography}[{\includegraphics[width=1in,height=1.25in,clip,keepaspectratio]{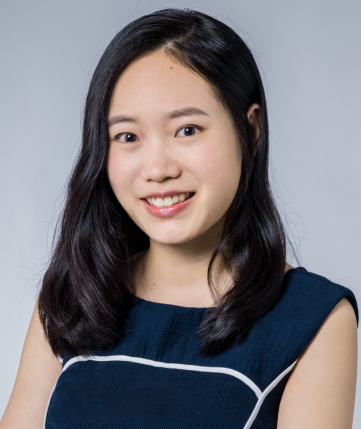}	}]	{Jingjing Luo}
 received the B.S. and Ph.D. degrees from the Department of Electronics and Information Engineering, Huazhong University of Science and Technology, Wuhan, China, in 2010 and 2015, respectively. Since 2016, she has been a Post-Doctoral Fellow with the Department of Information Engineering, Chinese University of
Hong Kong. Her current research interest includes scheduling and performance analysis in data center networks and content-centric networks.
\end{IEEEbiography}



\newpage
\appendices

\section{}\label{Appendix: MU offline}
\begin{proof}[\textbf{Proof of Theorem \ref{Theorem: MU KKT}}]
We prove this theorem based on the KKT conditions of Problem \ref{Problem: Reformulated JCOP}.
In this process, we actually prove Proposition \ref{Proposition: lambda} and Lemmas \ref{Lemma: primal-dual I}$\sim$\ref{Lemma: primal-dual IV}.
For notation clarity, we express Problem \ref{Problem: Reformulated JCOP} as follows:
\begin{subequations}
	\begin{align}
	\max &\quad \sum_{t=1}^{T}
				\Big[U_{t}(x_{t}) - \enCPT_{t}\big((x_{t}-z_{t})\cycle_{t}\big) - p_{t}z_{t} \cycle_{t}\Big] - \adfee s \\
	\textit{s.t.} 	
	&\quad 0\le z_{t} \le x_{t},\ \forall t\in\mathcal{T}, \label{Proof Equ: Reformulation Constraint r} \\
	&\quad s\ge \sum_{t=1}^{T}\left(\data_{t} x_{t} + \indata_{t} z_{t}\right) -\dcap. \label{Proof Equ: Reformulation Constraint s 2} \\
	&\quad s\ge 0, 														\label{Proof Equ: Reformulation Constraint s 1}	\\
	&\quad x_{t}\in[0,1],\ \forall t\in\mathcal{T},\\
	&\quad z_{t}\in[0,1],\ \forall t\in\mathcal{T},\label{Proof Equ: Reformulated JCOP integer}\\
	\textit{var.}	&\quad \bm{x},\bm{z},s.
	\end{align}
\end{subequations}
We let $\mu_t$ denote the Lagrangian multiplier associated with the constraint (\ref{Proof Equ: Reformulation Constraint r}).
We denote $\lambda$ and $\Lambda$ as the Lagrangian multipliers associated with the constraints (\ref{Proof Equ: Reformulation Constraint s 2}) and (\ref{Proof Equ: Reformulation Constraint s 1}), respectively.
Moreover, we let $\xi_{t}$ and $\psi_{t}$ denote the Lagrangian multipliers associated with the constraints $x_{t}\le1$ and $z_{t}\le1$ for any $t\in\mathcal{T}$, respectively.
Accordingly, the Lagrangian of Problem \ref{Problem: Reformulated JCOP} is given by
\begin{equation}
\begin{aligned}
& L(\bm{x},\bm{z},s,\lambda,\Lambda,\bm{\mu}) \\
=& \sum_{t=1}^{T}
\Big[U_{t}(x_{t}) - \enCPT_{t}\big((x_{t}-z_{t})\cycle_{t}\big) - p_{t}z_{t} \cycle_{t}\Big] - \adfee s \\
& +\lambda\left[ s+\dcap-\sum_{t=1}^{T}\left(\data_{t} x_{t}+ \indata_{t} z_{t}\right) \right] + \Lambda s ,\\
& + \sum_{t=1}^{T} \big[ \mu_t(x_t-z_t) + \xi_t(1-x_{t}) +  \psi_t(1-z_t) \big].
\end{aligned}
\end{equation}

We first express all the KKT conditions that are related to the variable $s$ as follows:
\begin{subequations}\label{Proof Equ: KKT s}
\begin{align}
	\frac{\partial L}{\partial s}=\lambda+\Lambda-\adfee&=0, \\
	\lambda\left[ s+\dcap-\sum_{t=1}^{T}\left(\data_{t} x_{t}+ \indata_{t} z_{t}\right) \right]&=0,\\
	\Lambda s&=0,\\
	\lambda,\Lambda&\ge0.
\end{align}
\end{subequations}

Note that equations (\ref{Proof Equ: KKT s}) imply $0\le\lambda\le\adfee$, which proves Proposition \ref{Proposition: lambda}.
Furthermore,  we express all the KKT conditions that are related to the variables $x_t$ and $z_t$ as follows:
\begin{subequations}\label{Proof Equ: KKT x z}
\begin{align}
\val_{t}u'_{t}(x_{t}) - \cost_{t}\cycle_{t}e'_{t}\big((x_t-z_t)\cycle_t\big) - \lambda\data_{t} + \mu_t -\xi_t &=0,\\
\cost_{t}\cycle_{t}e'_{t}\big((x_t-z_t)\cycle_t\big) - p_{t}\cycle_{t} - \lambda\indata_{t} - \mu_t -\psi_t&=0,\\
\mu_{t}(x_{t}-z_{t})&=0,\\
\xi_{t}(1-x_{t})&=0,\\
\psi_{t}(1-z_{t})&=0,\\
\mu_t,\xi_t,\psi_t&\ge0.
\end{align}
\end{subequations}
Next we derive the optimal primal-dual solutions based on (\ref{Proof Equ: KKT x z}) in the four cases discussed in Lemmas \ref{Lemma: primal-dual I}$\sim$\ref{Lemma: primal-dual IV}.

According to the definition of set $\Omega_{t}^{\text{I}}(\lambda)$ in (\ref{Equ: set I}), the case of $(\cost_{t},\val_{t})\in\Omega_{t}^{\text{I}}(\lambda^*)$ indicates the following inequalities
\begin{subequations}
	\begin{align}
	\val_{t}u'_{t}(1)-\cost_{t}\cycle_{t}e'_{t}(\cycle_{t})-\lambda^*\data_{t} &> 0,\\
	\cost_{t}\cycle_{t}e'_{t}(\cycle_{t}) - p_{t}\cycle_{t}-\lambda^*\indata_{t}&<0,
	\end{align}
\end{subequations}
which imply $x_{t}^*=1$, $z_{t}^*=0$, $\mu_t^*=0$, $\xi_t^*=0$, and $\psi_t^*=0$.
This completes the proof of Lemma \ref{Lemma: primal-dual I}.

According to the definition of set $\Omega_{t}^{\text{II}}(\lambda)$ in (\ref{Equ: set II}), the case of $(\cost_{t},\val_{t})\in\Omega_{t}^{\text{II}}(\lambda^*)$ indicates the following inequalities
\begin{subequations}
\begin{align}
\val_{t}u'_{t}(1)-\cost_{t}\cycle_{t}e'_{t}(\cycle_{t})-\lambda^*\data_{t} &> 0,\\
\cost_{t}\cycle_{t}e'_{t}(\cycle_{t}) - p_{t}\cycle_{t}-\lambda^*\indata_{t}&\ge0,
\end{align}
\end{subequations}
which imply $x_{t}^*=1$ and $z_{t}^*$ satisfies
\begin{equation}
\cost_{t}\cycle_{t}e'_{t}\big((1-z_t^*)\cycle_t\big) - p_{t}\cycle_{t} - \lambda^*\indata_{t} =0.
\end{equation}
Hence we have
\begin{equation}\label{Proof Equ: z II}
\textstyle
z_{t}^{\text{II}}(\lambda)
\triangleq
1 - {\encpt'^{-1}_{t}\left( \frac{p_{t}\cycle_{t}+\indata_{t}\lambda}{\cost_{t}\cycle_{t}} \right)}\Big/{\cycle_{t}},
\end{equation}
which completes the proof of Lemma \ref{Lemma: primal-dual II}

According to the definition of set $\Omega_{t}^{\text{III}}(\lambda)$ in (\ref{Equ: set III}), the case of $(\cost_{t},\val_{t})\in\Omega_{t}^{\text{III}}(\lambda^*)$ indicates the following inequalities
\begin{subequations}
	\begin{align}
	\val_{t}u'_{t}(1) - \cost_{t}\cycle_{t}e'_{t}(\cycle_{t})-\data_{t}\lambda^* & \le 0, \\
	\textstyle u'^{-1}_{t}\left( \frac{p_{t}\cycle_{t}+(\data_{t}+\indata_{t})\lambda}{\val_{t}} \right) &> \textstyle\encpt'^{-1}_{t}\left( \textstyle\frac{p_{t}\cycle_{t}+\indata_{t}\lambda}{\cost_{t}\cycle_{t}} \right)\Big/\cycle_{t},
	\end{align}
\end{subequations}
which imply that the variables $x^*_t$ and $z^*_t$ satisfy the following conditions:
\begin{subequations}
	\begin{align}
	\val_{t}u'_{t}(x^*_{t}) - \cost_{t}\cycle_{t}e'_{t}\big((x^*_t-z^*_t)\cycle_t\big) - \lambda^*\data_{t}  &=0,\\
	\cost_{t}\cycle_{t}e'_{t}\big((x^*_t-z^*_t)\cycle_t\big) - p_{t}\cycle_{t} - \lambda^*\indata_{t} &=0,
	\end{align}
\end{subequations}
and $\xi^*_{t}=\psi^*_{t}=\mu^*_{t}=0$.
Hence we have
\begin{subequations}
	\begin{align}
	x_{t}^{\text{III}}(\lambda)
	&\textstyle=  u'^{-1}_{t}\left( \frac{p_{t}\cycle_{t}+(\data_{t}+\indata_{t})\lambda}{\val_{t}} \right),\\
	z_{t}^{\text{III}}(\lambda)
	&\textstyle= x_{t}^{\text{III}}(\lambda)-\encpt'^{-1}_{t}\left( \frac{p_{t}\cycle_{t}+\indata_{t}\lambda}{\cost_{t}\cycle_{t}} \right)\Big/\cycle_{t},
	\end{align}	
\end{subequations}	
which completes the proof of Lemma \ref{Lemma: primal-dual III}.

According to the definition of set $\Omega_{t}^{\text{IV}}(\lambda)$ in (\ref{Equ: set IV}), the case of $(\cost_{t},\val_{t})\in\Omega_{t}^{\text{IV}}(\lambda^*)$ indicates the following inequalities
\begin{subequations}
	\begin{align}
	\val_{t}u'_{t}(1) - \cost_{t}\cycle_{t}e'_{t}(\cycle_{t})-\data_{t}\lambda^* & \le 0, \\
	\textstyle u'^{-1}_{t}\left( \frac{p_{t}\cycle_{t}+(\data_{t}+\indata_{t})\lambda}{\val_{t}} \right) &\le \textstyle\encpt'^{-1}_{t}\left( \textstyle\frac{p_{t}\cycle_{t}+\indata_{t}\lambda}{\cost_{t}\cycle_{t}} \right)\Big/\cycle_{t},
	\end{align}
\end{subequations}
which imply that $z_t^*=0$ and  $x^*_t$ satisfies
\begin{equation}
\textstyle
\val_{t}u'_{t}\left(x_{t}^*\right) -\cost\cycle_{t} \encpt'_{t}\left(x_{t}^{*}\cycle_{t}\right) =\data_{t}\lambda^* .
\end{equation}
This completes the proof of Lemma \ref{Lemma: primal-dual IV}.

The remaining part of proving this theorem is to derive the optimal Lagrangian multiplier $\lambda^*$.

Based on the above discussions and the definition (\ref{Equ: A_lambda}), we have
\begin{equation}
	\sum_{t=1}^{T} \data_{t} x_{t}^*(\lambda)+ \indata_{t} z_{t}^*(\lambda)= A(\lambda).
\end{equation}
Hence the optimal Lagrangian multiplier $\lambda^*$ has two possibilities:
\begin{itemize}
	\item Based on (\ref{Proof Equ: KKT s}), the case of $A(\pi)\ge\dcap$ indicates that $\lambda^*=\adfee$, $s^*=A(\pi)-\dcap\ge0$.
	\item Based on (\ref{Proof Equ: KKT s}), the case of $A(\pi)<\dcap$ indicates that $s^*=0$ and $\lambda^*=\lambda^{\dag}$.
\end{itemize}
This completes the proof of this theorem.
\end{proof}

\section{}\label{Appendix: MU online}
\begin{proof}[\textbf{Proof of Theorem \ref{Theorem: User Online}}]
	We prove this theorem by deriving a upper bound on the MU's payoff gap.
	For notation simplicity, we let $\lambda^{\star}$ denote the shadow price associate with the online decisions $(\bm{\hat{x}},\bm{\hat{z}})$.
	Then we have
	\begin{equation}\label{Proof Equ: regret}
		\begin{aligned}
			& S(\bm{x^*},\bm{z^*})-S(\bm{\hat{x}},\bm{\hat{z}}) \\
			=& \sum_{t=1}^{T}L_{t}(x^*_{t},z^*_{t},\lambda^*)  - \sum_{t=1}^{T}L_{t}(\hat{x}_{t},\hat{z}_{t},\hat{\lambda}^{\star}) \\
			=& \sum_{t=1}^{T}\left[L_{t}(x^*_{t},z^*_{t},\lambda^*)
			- L_{t}(x^*_{t},z^*_{t},\hat{\lambda}_{t}) \right] \\
			&\quad +\sum_{t=1}^{T}\left[L_{t}(x^*_{t},z^*_{t},\hat{\lambda}_{t})
			- L_{t}(\hat{x}_{t},\hat{z}_{t},\hat{\lambda}^{\star}) \right] ,\\
			\le& \sum_{t=1}^{T}\left[L_{t}(x^*_{t},z^*_{t},\lambda^*)
			- L_{t}(x^*_{t},z^*_{t},\hat{\lambda}_{t}) \right] \\
			&\quad +\sum_{t=1}^{T}\left[L_{t}(\hat{x}_{t},\hat{z}_{t},\hat{\lambda}_{t})
			- L_{t}(\hat{x}_{t},\hat{z}_{t},\hat{\lambda}^{\star}) \right].
		\end{aligned}
	\end{equation}
	Furthermore, for notation simplicity, we define $G(\bm{\hat{\lambda}})$ and $H(\bm{\hat{\lambda}})$ as follows:
	\begin{equation}
		G(\bm{\hat{\lambda}}) \triangleq
		\sum_{t=1}^{T}
		\left[L_{t}(x^*_{t},z^*_{t},\lambda^*) - L_{t}(x^*_{t},z^*_{t},\hat{\lambda}_{t}) \right],
	\end{equation}
	
	\begin{equation}
		H(\bm{\hat{\lambda}}) \triangleq
		\sum_{t=1}^{T}
		\left[L_{t}(\hat{x}_{t},\hat{z}_{t},\hat{\lambda}_{t}) - L_{t}(\hat{x}_{t},\hat{z}_{t},\hat{\lambda}^{\star}) \right],
	\end{equation}
	and we have $S(\bm{x^*},\bm{z^*})-S(\bm{\hat{x}},\bm{\hat{z}})= G(\bm{\hat{\lambda}})+H(\bm{\hat{\lambda}})$.
	Moreover, combining the following Lemma \ref{Lemma: H} and Lemma \ref{Lemma: G} proves this theorem.	
\end{proof}

\begin{lemma}\label{Lemma: H}
	The shadow price sequence $\bm{\hat{\lambda}}=\{\hat{\lambda}_{t},t\in\mathcal{T}\}$ generated by strategy $\mathcal{A}$ in Algorithm \ref{Algorithm: MU} satisfies
	\begin{equation}
	H(\bm{\hat{\lambda}})
	\le \frac{\adfee^2}{2} \frac{1}{\eta_{T}} +\Xi^2\sum_{t=1}^{T}\frac{\eta_{t}}{2}.
	\end{equation}
\end{lemma}
\begin{proof}[\textbf{Proof Lemma \ref{Lemma: H}}]
	We prove this lemma by showing an upper bound of the defined $H(\bm{\hat{\lambda}})$.
	According to the definition of $L(\cdot)$, we have
	\begin{equation}\label{Proof Equ: convexity}
		L_{t}(\hat{x}_{t},\hat{z}_{t},\hat{\lambda}_{t})-L_{t}(\hat{x}_{t},\hat{z}_{t},\hat{\lambda}^{\star})
		= \left(\textstyle \frac{\dcap}{T}-\tilde{h}_{t}(\hat{x}_{t},\hat{z}_{t}) \right)  \left( \hat{\lambda}_{t}-\hat{\lambda}^{\star} \right).
	\end{equation}
	
	Furthermore, the projection operation in Algorithm \ref{Algorithm: MU} indicates the following inequality
	\begin{equation}
		\begin{aligned}
			\left( \hat{\lambda}_{t+1}-\hat{\lambda}^{\star} \right)^2
			=&\left[\mathcal{P}_{[0,\adfee]}\left(\hat{\lambda}_{t} - \eta_{t}{\textstyle\left(\frac{\dcap}{T}-\tilde{h}_{t}(\hat{x}_{t},\hat{z}_{t})\right)}\right)-\hat{\lambda}^{\star}  \right]^2 \\
			\le& \left[ \hat{\lambda}_{t} - \eta_{t}{\textstyle\left(\frac{\dcap}{T}-\tilde{h}_{t}(\hat{x}_{t},\hat{z}_{t})\right)}-\hat{\lambda}^{\star} \right]^2 \\
			=&
			\left(\hat{\lambda}_{t}-\hat{\lambda}^{\star}\right)^2
			+ \eta_{t}^2 {\textstyle\left(\frac{\dcap}{T}-\tilde{h}_{t}(\hat{x}_{t},\hat{z}_{t})\right)}^2 \\
			&-2\eta_{t}\big( \hat{\lambda}_{t}-\hat{\lambda}^{\star} \big)\cdot{\textstyle\left(\frac{\dcap}{T}-\tilde{h}_{t}(\hat{x}_{t},\hat{z}_{t})\right)},
		\end{aligned}
	\end{equation}
	which implies that
	\begin{equation}\label{Proof Equ: projection}
		\begin{aligned}
			\big( \hat{\lambda}_{t}-\hat{\lambda}^{\star} \big)\cdot {\textstyle\left(\frac{\dcap}{T}-\tilde{h}_{t}(\hat{x}_{t},\hat{z}_{t})\right)}
			\le&\textstyle \frac{  (\hat{\lambda}_{t}-\hat{\lambda}^{\star}) ^2
				- (\hat{\lambda}_{t+1}-\hat{\lambda}^{\star}) ^2 }{2\eta_{t}}\\
			&\textstyle + \frac{\eta_{t}}{2} {\left(\frac{\dcap}{T}-\tilde{h}_{t}(\hat{x}_{t},\hat{z}_{t})\right)}^2 .
		\end{aligned}	
	\end{equation}
	Combining (\ref{Proof Equ: convexity}) and (\ref{Proof Equ: projection}), we obtain
	\begin{equation}\label{Proof Equ: combine}
		\begin{aligned}
			L_{t}(\hat{x}_{t},\hat{z}_{t},\hat{\lambda}_{t})-L_{t}(\hat{x}_{t},\hat{z}_{t},\hat{\lambda}^{\star})
			\le&\textstyle\frac{  \left(\hat{\lambda}_{t}-\hat{\lambda}^{\star}\right) ^2
				- \left(\hat{\lambda}_{t+1}-\hat{\lambda}^{\star}\right) ^2 }{2\eta_{t}}\\
			& + \textstyle\frac{1}{2}\eta_{t} {\textstyle\left(\frac{\dcap}{T}-\tilde{h}_{t}(\hat{x}_{t},\hat{z}_{t})\right)}^2 .
		\end{aligned}	
	\end{equation}
	Moreover, summing (\ref{Proof Equ: combine}) from $t = 1$ to $T$ leads to
	\begin{equation}
		\begin{aligned}
			&\sum_{t=1}^{T}L_{t}(\hat{x}_{t},\hat{z}_{t},\hat{\lambda}_{t})-L_{t}(\hat{x}_{t},\hat{z}_{t},\hat{\lambda}^{\star})\\
			\le&\sum_{t=1}^{T} \frac{  (\hat{\lambda}_{t}-\hat{\lambda}^{\star}) ^2
				- (\hat{\lambda}_{t+1}-\hat{\lambda}^{\star})^2 }{2\eta_{t}}
			+ \frac{\eta_{t}}{2} {\left(\frac{\dcap}{T}-\tilde{h}_{t}(\hat{x}_{t},\hat{z}_{t})\right)}^2\\
			\le& \sum_{t=1}^{T} \frac{ (\hat{\lambda}_{t}-\hat{\lambda}^{\star})^2}{2}\left(\frac{1}{\eta_{t}}-\frac{1}{\eta_{t-1}}\right)
			+\sum_{t=1}^{T} \frac{\eta_{t}}{2}{\left(\frac{\dcap}{T}-\tilde{h}_{t}(\hat{x}_{t},\hat{z}_{t})\right)}^2\\
			\le& \frac{\adfee^2}{2}\sum_{t=1}^{T}\left(\frac{1}{\eta_{t}}-\frac{1}{\eta_{t-1}}\right)  +\Xi^2\sum_{t=1}^{T}\frac{\eta_{t}}{2}\\
			\le& \frac{\adfee^2}{2} \frac{1}{\eta_{T}} +\Xi^2\sum_{t=1}^{T}\frac{\eta_{t}}{2}	 ,
		\end{aligned}
	\end{equation}
	where $\Xi$ is defined in Definition \ref{Definition: Xi}.
	This completes the proof of this lemma.
\end{proof}

\begin{lemma}\label{Lemma: G}
	The shadow price sequence $\bm{\tilde{\lambda}}=\{\tilde{\lambda}_{t},t\in\mathcal{T}\}$ generated by strategy $\mathcal{A}$ in Algorithm \ref{Algorithm: MU} satisfies
	\begin{equation}
	G(\bm{\tilde{\lambda}})
	\le
	\Xi\Psi\sum_{t=1}^{T}\eta_{t}.
	\end{equation}
\end{lemma}
\begin{proof}[\textbf{Proof of Lemma \ref{Lemma: G}}]
	We prove this lemma by showing an upper bound for $G(\bm{\hat{\lambda}})$.
	According to the definition of $L_{t}(\cdot)$, we have
	\begin{equation}
		\begin{aligned}
			G(\bm{\hat{\lambda}})
			&=  \sum_{t=1}^{T} L_{t}(x^*_{t},z^*_{t},\lambda^*) - L_{t}(x^*_{t},z^*_{t},\hat{\lambda}_{t})\\
			&=  \sum_{t=1}^{T} \left( \hat{\lambda}_{t} -\lambda^* \right)\left(  \tilde{h}_{t}(x^*_{t},z^*_{t}) - \frac{\dcap}{T} \right)\\
			&=  -\sum_{t=1}^{T} \hat{\lambda}_{t}  l_{t} - \lambda^{*}\sum_{t=1}^{T}-l_{t}\\
			&=  -\sum_{t=1}^{T} \hat{\lambda}_{t}  l_{t} - \adfee\left[ \sum_{t=1}^{T}-l_{t} \right]^+.
		\end{aligned}
	\end{equation}
	where $l_t$ is defined in Definition \ref{Definition: Phi}.
	
	Note that the shadow price sequence $\bm{\hat{\lambda}}$ generated by Algorithm \ref{Algorithm: MU} satisfies the following inequalities
	\begin{equation}\label{Proof Equ: lambda range}
		|\hat{\lambda}_{t}-\hat{\lambda}_{t+1}|\le \Xi\eta_{t},\quad\forall t\in\{1,2,...,T-1\}.
	\end{equation}

	Hence we can derive an upper bound for $G(\bm{\hat{\lambda}})$ by maximizing it over $\bm{\hat{\lambda}}=(\hat{\lambda}_{t},t\in\mathcal{T})$ under the inequality constraints (\ref{Proof Equ: lambda range}).
	That is, we have $G(\bm{\hat{\lambda}})\le\hat{G}$, where $\hat{G}$ is defined by
	\begin{equation}
		\begin{aligned}
			\hat{G}\triangleq \arg
			\max&\  -\sum_{t=1}^{T} {\lambda}_{t}  l_{t} - \adfee\left[ \sum_{t=1}^{T}-l_{t} \right]^+ \\
			\textit{s.t.}&\ |\lambda_{t}-\lambda_{t+1}|\le \Xi\eta_{t},\quad\forall t\in\{1,2,...,T-1\},\\
			\textit{var.}&\ \bm{\lambda}=\{\lambda_{t}, 1\le t\le T\}.
		\end{aligned}
	\end{equation}
	For notation simplicity, we define $\bm{\Delta}\in\mathbb{R}^{(T-1)\times T}$ as follows
	\begin{equation}
		\Delta_{i,j}=\left\{
		\begin{aligned}
			& 1,	&\textit{if } i=j,\\
			& -1,	&\textit{if } i+1=j,\\
			& 0,	&\textit{otherwise}.
		\end{aligned}
		\right.
	\end{equation}
	and introduce a set of auxiliary variables $\bm{s}\in\mathbb{R}^{T-1}$.	
	\begin{equation}
		\begin{aligned}
			\hat{G}\triangleq \arg
			\max&\  -\sum_{t=1}^{T} {\lambda}_{t}  l_{t} - \adfee\left[ \sum_{t=1}^{T}-l_{t} \right]^+ \\
			\textit{s.t.}&\ |s_{t}| \le \Xi\eta_{t},\quad\forall t\in\{1,2,...,T-1\},\\
			&\ \bm{s}=\bm{\Delta}\bm{\lambda},\\
			\textit{var.}&\ \bm{\lambda}=\{\lambda_{t},1\le t\le T\},\ \bm{s}=\{s_{t},1\le t\le T-1\}.
		\end{aligned}
	\end{equation}

	We express the Lagrangian as follows:
	\begin{equation}
		\begin{aligned}
			&\mathcal{L}(\bm{\lambda},\bm{s};\bm{\mu},\bm{\nu})\\
			=&-\sum_{t=1}^{T} {\lambda}_{t}l_{t} - \adfee\left[ \sum_{t=1}^{T}-l_{t} \right]^+
			+ \bm{\mu}^{\top}(\bm{s}-\bm{\Delta}\bm{\lambda}) \\
			&+ \sum_{t=1}^{T-1}\nu_{t}\left(\Xi\eta_{t}-|s_{t}|\right) \\
		\end{aligned}
	\end{equation}
	Moreover, the duality theory implies that
	\begin{equation}
		\hat{G} \le \min\limits_{\bm{\mu},\bm{\nu}}	\max\limits_{\bm{\lambda},\bm{s}} \mathcal{L}(\bm{\lambda},\bm{s};\bm{\mu},\bm{\nu}).
	\end{equation}
	We maximize $\mathcal{L}(\bm{\lambda},\bm{s};\bm{\mu},\bm{\nu})$ over $\bm{\lambda}$ and $\bm{s}$, and obtain
	\begin{equation}
		\begin{aligned}
			&\mathcal{L}(\bm{\lambda}^{*},\bm{s}^{*},\bm{\mu},\bm{\nu}) \\
			=& \left\{
			\begin{aligned}
				&\textstyle - \adfee\left[ \sum\limits_{t=1}^{T}-l_{t} \right]^+
				+ \adfee\sum\limits_{t=1}^{T}\left[-l_{t}-\left(\bm{\Delta}^{\top}\bm{\mu}\right)_{t}\right]^+\\
				& \qquad\qquad+ \sum\limits_{t=1}^{T-1}\nu_{t}\eta_{t}\Xi,		\qquad\qquad\textit{if } \nu_{t}\ge|\mu_{t}|,\ \forall t,\\
				& +\infty,		\qquad\qquad\qquad\qquad\qquad\qquad\textit{otherwise}.
			\end{aligned}
			\right.
		\end{aligned}	
	\end{equation}
	Then we minimize $\mathcal{L}(\bm{\lambda}^{*},\bm{s}^{*},\bm{\mu},\nu)$ over $\bm{\nu}$ and obtain
	\begin{equation}
		\begin{aligned}
			&\mathcal{L}(\bm{\lambda}^{*},\bm{s}^{*},\bm{\mu},\bm{\nu^{*}})\\
			=&\textstyle
			 - \adfee\left[ \sum\limits_{t=1}^{T}-l_{t} \right]^+
			+ \adfee\sum\limits_{t=1}^{T}\left[-l_{t}-\left(\bm{\Delta}^{\top}\bm{\mu}\right)_{t}\right]^+ \\
			&\textstyle + \sum\limits_{t=1}^{T-1}|\mu_{t}|\eta_{t}\Xi.
		\end{aligned}	
	\end{equation}
	
	Now we know that the following inequality holds
	\begin{equation}
		\hat{G}\le \mathcal{L}(\bm{\lambda}^{*},\bm{s}^{*},\bm{\mu},\nu^{*}),\quad\forall \bm{\mu}.
	\end{equation}
	Hence we define $\bm{\mu^{\star}}\in\mathbb{R}^{T-1}$ as follows:
	\begin{equation}
		\mu^{\star}_{t} = \bar{l}^{*}t - \sum_{k=1}^{t}l_{k} , \ \forall\ t\in\{1,2,...,T-1\},
	\end{equation}
	and obtain
	\begin{equation}
		\begin{aligned}
			\sum_{t=1}^{T}\left[-l_{t}-\left(\bm{\Delta}^{\top}\bm{\mu^{\star}}\right)_{t}\right]^+
			&=\sum_{t=1}^{T}\left[-l_{t}-\left( \bar{l}-r^{*}_{t} \right)\right]^+\\
			&=\sum_{t=1}^{T}\left[-\bar{l}\right]^+ \\
			&=  T\left[-\bar{l}\right]^+ = \left[\sum_{t=1}^{T}-l_{t}\right]^+.
		\end{aligned}
	\end{equation}
	Therefore, $\mathcal{L}(\bm{\lambda}^{*},\bm{s}^{*},\bm{\mu^{\star}},\nu^{*})$ is given by
	\begin{equation}
		\mathcal{L}(\bm{\lambda}^{*},\bm{s}^{*},\bm{\mu^{\star}},\nu^{*}) =
		\sum_{t=1}^{T-1}|\mu^{\star}_{t}|\eta_{t}\Xi \le \Xi\sum_{t=1}^{T}\psi_{t}\eta_{t} \le \Xi\Psi\sum_{t=1}^{T}\eta_{t}.
	\end{equation}
	
	This completes the proof of this lemma.	
\end{proof}

\section{}\label{Appendix: ESP}
\begin{proof}[\textbf{Proof of Lemma \ref{Lemma: discrete loss}}]
	Recall that $V_{\text{ESP}}^{\star}$ is the optimal revenue under fixed pricing in hindsight.
	We let $p_{\text{opt}}\in\mathbb{R}$ denote the corresponding optimal price.
	That is,
	\begin{equation}
		p_{\text{opt}}
		\triangleq
		\arg\max\limits_{p\ge p_{\textit{min}}}\sum_{t=1}^{T}\sum_{n=1}^{N}p\cdot\cycle_{n,t}\cdot x_{n,t}^{*}(p)\cdot y_{n,t}^{*}(p).
	\end{equation}
	According to the definition of the discrete price candidates, we suppose that $p_{\text{opt}}$ satisfies
	\begin{equation}
		p(\kappa) \le p_{\text{opt}} \le p(\kappa+1).
	\end{equation}
	Therefore, we have
	\begin{subequations}
		\begin{align}
			V_{\text{ESP}}(\kappa)
			&=p(\kappa)\cdot \sum_{n=1}^{N}\sum_{t=1}^{T}\cycle_{n,t}z_{n,t}^{*}\big( p(\kappa) \big)	\\
			&=\frac{p(\kappa+1)}{1+\epsilon}\cdot \sum_{n=1}^{N}\sum_{t=1}^{T}\cycle_{n,t}z_{n,t}^{*}\big( p(\kappa) \big)	\label{Proof Equ: discrete loss price denifition}\\
			&\ge \frac{p(\kappa+1)}{1+\epsilon}\cdot \sum_{n=1}^{N}\sum_{t=1}^{T}\cycle_{n,t}z_{n,t}^{*}\big(p_{\text{opt}}\big) \label{Proof Equ: discrete loss increasing price} \\
			&\ge \frac{p_{\text{opt}}}{1+\epsilon} \cdot \sum_{n=1}^{N}\sum_{t=1}^{T}\cycle_{n,t}z_{n,t}^{*}\big(p_{\text{opt}}\big)\\
			&=\frac{ V_{\text{ESP}}^{\star} }{1+\epsilon},
		\end{align}
	\end{subequations}
	where (\ref{Proof Equ: discrete loss increasing price}) is due to $p(\kappa)\le p_{\text{opt}}$.
	This completes the proof of this lemma.
\end{proof}


\begin{proof}[\textbf{Proof of Lemma \ref{Lemma: EXP3}}]

	We let $\Omega_{t}\triangleq\sum_{k=1}^{K}\omega_{t}(k)$ denote the total weight in slot $t$.
	
	The remaining proof consists of three parts.
	
	\textbf{Part I:} Derive for $\ln(\Omega_{T+1})-\ln(\Omega_{1})$ an upper bound that is related to $\sum_{k=1}^{K}V_{t}(k,\bm{\kappa}_{t})$.

	According to weight updating in (\ref{Equ: update weight}), we have
	\begin{equation}\label{Proof Equ: omega Omega}
		\begin{aligned}
			\frac{\Omega_{t+1}}{\Omega_{t}}
			=&\sum_{k=1}^{K}\frac{\omega_{t+1}(k)}{\Omega_{t}}
			=\sum_{k=1}^{K}\frac{\omega_{t}(k)}{\Omega_{t}}\cdot(1+\delta)^{ {\color{black}\hat{V}_{t}}(k,\bm{\kappa}_{t})}\\
			\le&\sum_{k=1}^{K}\frac{\omega_{t}(k)}{\Omega_{t}}\cdot \left[1+\delta\cdot{\color{black}\hat{V}_{t}}(k,\bm{\kappa}_{t})\right], \\
		\end{aligned}	
	\end{equation}
	where the last inequality follows $\hat{V}_{t}(k,\bm{\kappa}_{t})\le1$ (shown in Proposition \ref{Proposition: virtual candidate revenue}) together with the fact that $(1+\delta)^{x}\le 1+\delta x$ holds for any $x\in[0,1]$.
	
	Moreover, the definition of $h_{t}(k)$ implies that we can express $\frac{\omega_{t}(k)}{\Omega_{t}}$ as follows:
	\begin{equation}\label{Proof Equ: omega Omega ht}
		\frac{\omega_{t}(k)}{\Omega_{t}} = \frac{h_{t}(k)-{\gamma(1+\epsilon)^{k}}/{\mathcal{P}}}{1-\gamma},
	\end{equation}
	where $\mathcal{P}\triangleq\sum_{i=1}^{K}(1+\epsilon)^{i}$ is a constant.
	
	Combining (\ref{Proof Equ: omega Omega}) and (\ref{Proof Equ: omega Omega ht}), we obtain
	\begin{equation}\label{Proof Equ: omega Omega final}
		\begin{aligned}
			\frac{\Omega_{t+1}}{\Omega_{t}}
			&\le \sum_{k=1}^{K}\frac{h_{t}(k)-{\gamma(1+\epsilon)^{k}}/{\mathcal{P}}}{1-\gamma}\cdot \left[ 1 + \delta\cdot{\color{black}\hat{V}_{t}}(k,\bm{\kappa}_{t}) \right] \\
			&=1+ \sum_{k=1}^{K}\frac{h_{t}(k)-{\gamma(1+\epsilon)^{k}}/{\mathcal{P}}}{1-\gamma}\cdot\delta\cdot{\color{black}\hat{V}_{t}}(k,\bm{\kappa}_{t})  \\
			&\le 1
			+ \frac{\delta}{1-\gamma}\sum_{k=1}^{K} h_{t}(k)\cdot{\color{black}\hat{V}_{t}}(k,\bm{\kappa}_{t})\\
			&\le 1
			+ \frac{\gamma\delta}{(1-\gamma)N\bar{\cycle} p_{\textit{min}}\mathcal{P}}\sum_{k=1}^{K} V_{t}(k,\bm{\kappa}_{t}).
		\end{aligned}	
	\end{equation}
	We take the logarithmic operation of (\ref{Proof Equ: omega Omega final}) on both sides, and obtain
	\begin{equation}
		\begin{aligned}
			\ln(\Omega_{t+1})-\ln(\Omega_{t})
			&\le \ln\left( 1
			+ \frac{\gamma\delta\sum_{k=1}^{K} V_{t}(k,\bm{\kappa}_{t})}{(1-\gamma)N\bar{\cycle} p_{\textit{min}}\mathcal{P}}	\right)\\
			&\le \frac{\gamma\delta}{(1-\gamma)N\bar{\cycle} p_{\textit{min}}\mathcal{P}}\sum_{k=1}^{K}V_{t}(k,\bm{\kappa}_{t}),
		\end{aligned}
	\end{equation}
	where the second inequality follows that $\ln(1+x)\le x$ for any $x\ge0$.
	We then sum over $t=1$ to $t=T$  and obtain the desired upper bound in Part I as follows:
	\begin{equation}
		\begin{aligned}
			\ln(\Omega_{T+1})-\ln(K)
			&\le \frac{\gamma\delta}{(1-\gamma)N\bar{\cycle} p_{\textit{min}}\mathcal{P}} \sum_{t=1}^{T}\sum_{k=1}^{K} V_{t}(k,\bm{\kappa}_{t}).
		\end{aligned}
	\end{equation}
	
	\textbf{Part II:} Derive for $\ln(\Omega_{T+1})-\ln(\Omega_{1})$ a lower bound that is related to $\sum_{t=1}^{T}{\color{black}\hat{V}_{t}}(k,\bm{\kappa}_{t})$.
	
	Now we derive for $\ln(\Omega_{T+1})-\ln(K)$ a lower bound based on the definition of $\Omega_{t}$.
	Specifically, for any $k\in\mathcal{K}$ the following holds
	\begin{subequations}
		\begin{align}
			\ln(\Omega_{T+1})-\ln(K)
			&\ge \ln\left( \omega_{T+1}(k) \right) -\ln(K) ,\\
			&= \ln\left( (1+\delta)^{ \sum_{t=1}^{T}{\color{black}\hat{V}_{t}}(k,\bm{\kappa}_{t}) } \right) -\ln(K) ,\\
			&= \left( \sum_{t=1}^{T}{\color{black}\hat{V}_{t}}(k,\bm{\kappa}_{t}) \right) \ln(1+\delta)-\ln(K) ,\\
			&\ge \left( \sum_{t=1}^{T}{\color{black}\hat{V}_{t}}(k,\bm{\kappa}_{t}) \right) \left(\delta-\frac{\delta^2}{2}\right)-\ln(K) ,
		\end{align}	
	\end{subequations}
	where the last inequality follows that $x-\frac{x^2}{2}\le\ln(1+x)$ for any $x\ge0$.

	\textbf{Part III:}
	
	Combine the two inequities derived in Part I and Part III, we know that the following is true
	\begin{equation}\label{Proof Equ: Part III}
		\begin{aligned}
			&\sum_{t=1}^{T}{\color{black}\hat{V}_{t}}(k,\bm{\kappa}_{t})\cdot\left(\delta-\frac{\delta^2}{2}\right) -\ln(K) \le \\
			&\qquad \frac{\gamma\delta}{(1-\gamma)N\bar{\cycle} p_{\textit{min}}\mathcal{P}} \sum_{t=1}^{T}\sum_{k=1}^{K} V_{t}(k,\bm{\kappa}_{t}), \quad\forall\kappa\in\mathcal{K}.
		\end{aligned}	
	\end{equation}
	
	We take the expectation of (\ref{Proof Equ: Part III}) over all the random variables $[\bm{\kappa}_{1},\bm{\kappa}_{2},...,\bm{\kappa}_{T}]$, and obtain
	\begin{equation}\label{Proof Equ: Part III expected}
		\begin{aligned}
			&\mathbb{E}\left[\sum_{t=1}^{T}{\color{black}\hat{V}_{t}}(k,\bm{\kappa}_{t})\right]\cdot\left(\delta-\frac{\delta^2}{2}\right) - \ln(K) \le\\
			& \qquad
			\frac{\gamma\delta}{(1-\gamma)N\bar{\cycle} p_{\textit{min}}\mathcal{P}}\mathbb{E}\left[\sum_{t=1}^{T}\sum_{k=1}^{K}V_{t}(k,\bm{\kappa}_{t})\right],\quad\forall\kappa\in\mathcal{K}.
		\end{aligned}	
	\end{equation}
	Mathematically, we note that
	\begin{equation}\textstyle
		\mathbb{E}\left[\sum_{t=1}^{T}\sum_{k=1}^{K}V_{t}(k,\bm{\kappa}_{t})\right] = \mathbb{E}\left[ V_{\text{ESP}}(\mathcal{P}) \right].
	\end{equation}
	Moreover, we have
	\begin{equation}
		\begin{aligned}
			&\mathbb{E}\left[\sum_{t=1}^{T}{\color{black}\hat{V}_{t}}(k,\bm{\kappa}_{t})\right]\\
			=&\frac{\gamma}{N\bar{\cycle} p_{\textit{min}}\mathcal{P}} \sum_{t=1}^{T}\mathbb{E}\left[\frac{\sum_{n\in\mathcal{N}} V_{t,n}(\kappa_{t,n})\mathbb{I}(\kappa_{t,n}=k)}{h_{t}(k)} \right] \\
			=&\frac{\gamma}{N\bar{\cycle} p_{\textit{min}}\mathcal{P}} \sum_{t=1}^{T}\left[\frac{\sum_{n\in\mathcal{N}} V_{t,n}(k)}{h_{t}(k)}\cdot h_{t}(k) \right]\\
			=&\frac{\gamma}{N\bar{\cycle} p_{\textit{min}}\mathcal{P}} \sum_{t=1}^{T}\sum_{n=1}^{N}  V_{t,n}(k) \\
			=& \frac{\gamma }{N\bar{\cycle} p_{\textit{min}}\mathcal{P}} V_{\text{ESP}}(k) .
		\end{aligned}	
	\end{equation}	
	Substituting the above two expectation results into (\ref{Proof Equ: Part III expected}), we know that the following inequality holds
	\begin{equation}
		\begin{aligned}
			&\mathbb{E}\left[ V_{\text{ESP}}(\mathcal{P}) \right]\ge \\
			& (1-\gamma)\left(1-\frac{\delta}{2}\right)V_{\text{ESP}}(k) - N\bar{\cycle} p_{\textit{min}}\mathcal{P}\frac{1-\gamma}{\gamma\delta}\ln(K)	,	\quad\forall k\in\mathcal{K} .
		\end{aligned}	
	\end{equation}
	Accordingly, the remaining proof is to show that $\Phi(\epsilon,\delta,\gamma)$ is the upper bound of the loss term, as following
	\begin{equation}
		\begin{aligned}
			& N\bar{\cycle} p_{\textit{min}}\mathcal{P}\frac{1-\gamma}{\gamma\delta}\ln(K) \\
			=& \frac{1-\gamma}{\gamma\delta}\cdot N\bar{\cycle} p_{\textit{min}}\cdot\frac{[(1+\epsilon)^{K}-1](1+\epsilon)}{\epsilon} \cdot \ln(K) \\
			<& \frac{1-\gamma}{\gamma\delta}\cdot N\bar{\cycle} \cdot\frac{p_{\textit{min}}(1+\epsilon)^{K}(1+\epsilon)}{\epsilon} \cdot \ln(K) \\
			=& \frac{1-\gamma}{\gamma}\cdot \frac{1+\epsilon}{\epsilon} \cdot \frac{N\bar{E}\bar{\cycle}}{\delta} \cdot \ln\left( K\right) \\
			\le& \frac{1-\gamma}{\gamma}\cdot \frac{1+\epsilon}{\epsilon} \cdot \frac{N\bar{E}\bar{\cycle}}{\delta} \cdot \ln\left( \frac{ \ln \left({\bar{E}}/{p_{\textit{min}}}\right) }{ \ln(1+\epsilon) } \right) = \Phi(\epsilon,\delta,\gamma),\\
		\end{aligned}	
	\end{equation}
	which completes the proof.
\end{proof}

\begin{proposition}\label{Proposition: virtual candidate revenue}
	Given the sequence $\bm{\kappa}_{t}$ generated in Algorithm \ref{Algorithm: ESP}, we have
	\begin{equation}
		\hat{V}_{t}(k,\bm{\kappa}_{t})\le1,\quad\forall k\in\mathcal{K},\ t\in\mathcal{T}.
	\end{equation}
\end{proposition}
\begin{proof}[\textbf{Proof of Proposition \ref{Proposition: virtual candidate revenue}}]
	We prove this lemma by showing the upper bound of $\hat{V}_{t}(k,\bm{\kappa}_{t})$.
	According to the definition of $\hat{V}_{t}(k,\bm{\kappa}_{t})$ in (\ref{Equ: virtual candidate revenue}), we have
	\begin{subequations}
		\begin{align}
			\hat{V}_{t}(k,\bm{\kappa}_{t})
			&=\frac{V_{t}(k,\bm{\kappa}_{t})}{ N\bar{\cycle} p_{\textit{min}} } \cdot \frac{\gamma}{h_{t}(k)\sum_{i=1}^{K}(1+\epsilon)^{i}}	\\
			&\le\frac{V_{t}(k,\bm{\kappa}_{t})}{ N\bar{\cycle} p_{\textit{min}} } \cdot \frac{\gamma}{\frac{\gamma\cdot(1+\epsilon)^{k}}{\sum_{i=1}^{K}(1+\epsilon)^{i}}\cdot\sum_{i=1}^{K}(1+\epsilon)^{i}} \label{Proof Equ: virtual candidate revenue - ht} \\
			&=\frac{V_{t}(k,\bm{\kappa}_{t})}{ N\bar{\cycle} p_{\textit{min}} } \cdot \frac{1}{(1+\epsilon)^{k}} \\
			&\le\frac{p_{\textit{min}}(1+\epsilon)^{k}N\bar{\cycle}}{ N\bar{\cycle} p_{\textit{min}} } \cdot \frac{1}{(1+\epsilon)^{k}}=1, \label{Proof Equ: virtual candidate revenue - Vt}
		\end{align}
	\end{subequations}
	where (\ref{Proof Equ: virtual candidate revenue - ht}) is due to the fact $h_{t}(k)\ge \frac{\gamma\cdot(1+\epsilon)^{k}}{\sum_{i=1}^{K}(1+\epsilon)^{i}}$, and (\ref{Proof Equ: virtual candidate revenue - Vt}) follows $V_{t}(k,\bm{\kappa}_{t}) \le p_{\textit{min}}(1+\epsilon)^{k}N\bar{\cycle}$.
	
	This completes the proof of this proposition.
\end{proof}

\begin{proof}[\textbf{Proof of Corollary \ref{Corollary: ESP alpha}}]
	We prove this corollary based on Theorem \ref{Theorem: ESP} under the inequality condition $V^{\star}_{\text{ESP}}\ge \frac{8}{\alpha}\cdot\Phi\left(\frac{\alpha}{3},\frac{\alpha}{6},\frac{\alpha}{12}\right)$ for some $\alpha\in(0,1]$.
	
	Theorem \ref{Theorem: ESP} under the parameters $(\epsilon,\delta,\gamma)=\left(\frac{\alpha}{3},\frac{\alpha}{6},\frac{\alpha}{12}\right)$ indicates
	\begin{equation}
		\begin{aligned}
			\mathbb{E}\left[ V_{\text{ESP}}^{T}(\mathcal{P}) \right]
			&\ge\frac{(1-\frac{\alpha}{12})\left(1-\frac{\alpha}{12}\right)}{1+\frac{\alpha}{3}} V^{\star}_{\text{ESP}}
			-\Phi\left(\frac{\alpha}{3},\frac{\alpha}{6},\frac{\alpha}{12}\right) \\
			&\ge\frac{(1-\frac{\alpha}{12})^2}{1+\frac{\alpha}{3}} V^{\star}_{\text{ESP}} - \frac{\alpha}{8}V^{\star}_{\text{ESP}} \\
			&=\frac{(1-\frac{\alpha}{12})^2 - \left(1+\frac{\alpha}{3}\right)\frac{\alpha}{8}}{1+\frac{\alpha}{3}} V^{\star}_{\text{ESP}}.
		\end{aligned}	
	\end{equation}
	Note that we have
	\begin{subequations}
		\begin{align}
			\frac{(1-\frac{\alpha}{12})^2 - \left(1+\frac{\alpha}{3}\right)\frac{\alpha}{8}}{1+\frac{\alpha}{3}}
			&>\frac{\left(1-\frac{\alpha}{6}\right) - \left(1+\frac{\alpha}{3}\right)\frac{\alpha}{8} }{1+\frac{\alpha}{3}}\\
			&=\frac{1 - \frac{7\alpha+\alpha^2}{24} }{1+\frac{\alpha}{3}}\\
			&\ge\frac{1 - \frac{\alpha}{3} }{1+\frac{\alpha}{3}} \label{Proof Equ: corollary 1} \\
			&>\frac{1 }{1+\alpha}, \label{Proof Equ: corollary 2}
		\end{align}
	\end{subequations}
	where (\ref{Proof Equ: corollary 1}) follows that $\frac{7\alpha+\alpha^2}{24}\le\frac{\alpha}{3}$ for any $\alpha\in(0,1]$.
	And (\ref{Proof Equ: corollary 2}) follows that $\frac{1+\alpha/3}{1-\alpha/3}\le1+\alpha$ holds for any $\alpha\in(0,1]$.
	
	This completes the proof of this corollary.
\end{proof}

\section{Different Utility and Cost Functions}
This section extend the numerical results in Section \ref{Section: Numerical} by taking into account different utility and cost functions.
Specifically, Section \ref{Section: Numerical} focuses on $\alpha=0.5$ and $\beta=1$ given the series of utility $u(x)=\frac{x^{1-\alpha}}{1-\alpha}$ and the cost $e(s)=\frac{s^{1+\beta}}{1+\beta}$, respectively.
This section further considers the case of $\alpha=0.6$ and $\beta=2$.
Fig. \ref{fig: Different utility and cost} plots the utility and cost functions.
It is obvious that the concaveness or convexness is different.

	\begin{figure}[t]
		\setlength{\abovecaptionskip}{0pt}
		\setlength{\belowcaptionskip}{0pt}
		\centering	
		\subfigure[Utility with $\alpha\in\{0.5,0.6\}$.]
		{\label{fig: Appendix_utility}\includegraphics[height=0.35 \linewidth]{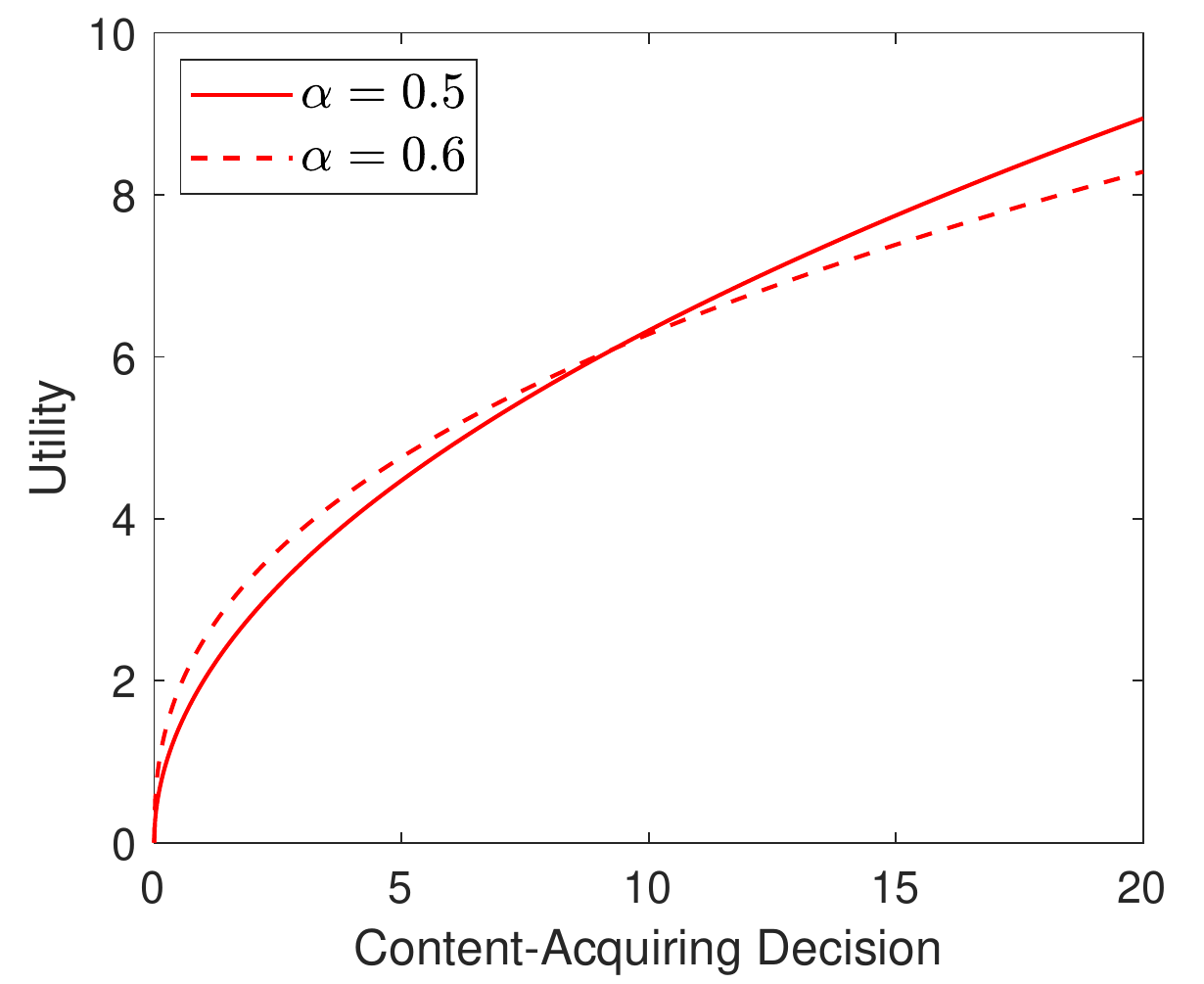}}
		\subfigure[Cost with $\beta\in\{1,2\}$.]
		{\label{fig: Appendix_cost}\includegraphics[height=0.35 \linewidth]{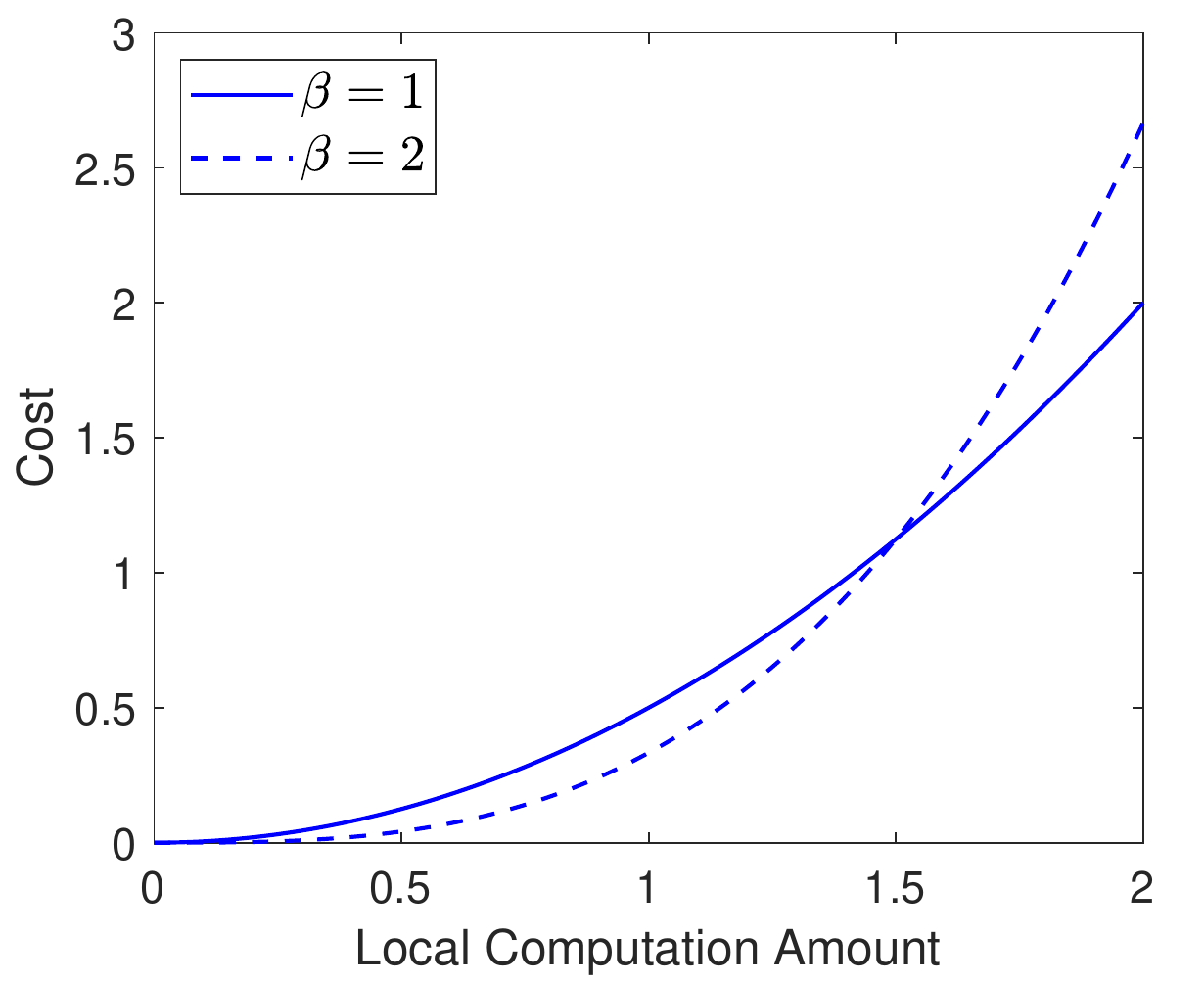}}
		\caption{Different utility and cost functions}
		\label{fig: Different utility and cost}
	\end{figure}

\begin{figure}
	\centering
	\includegraphics[height=0.4\linewidth]{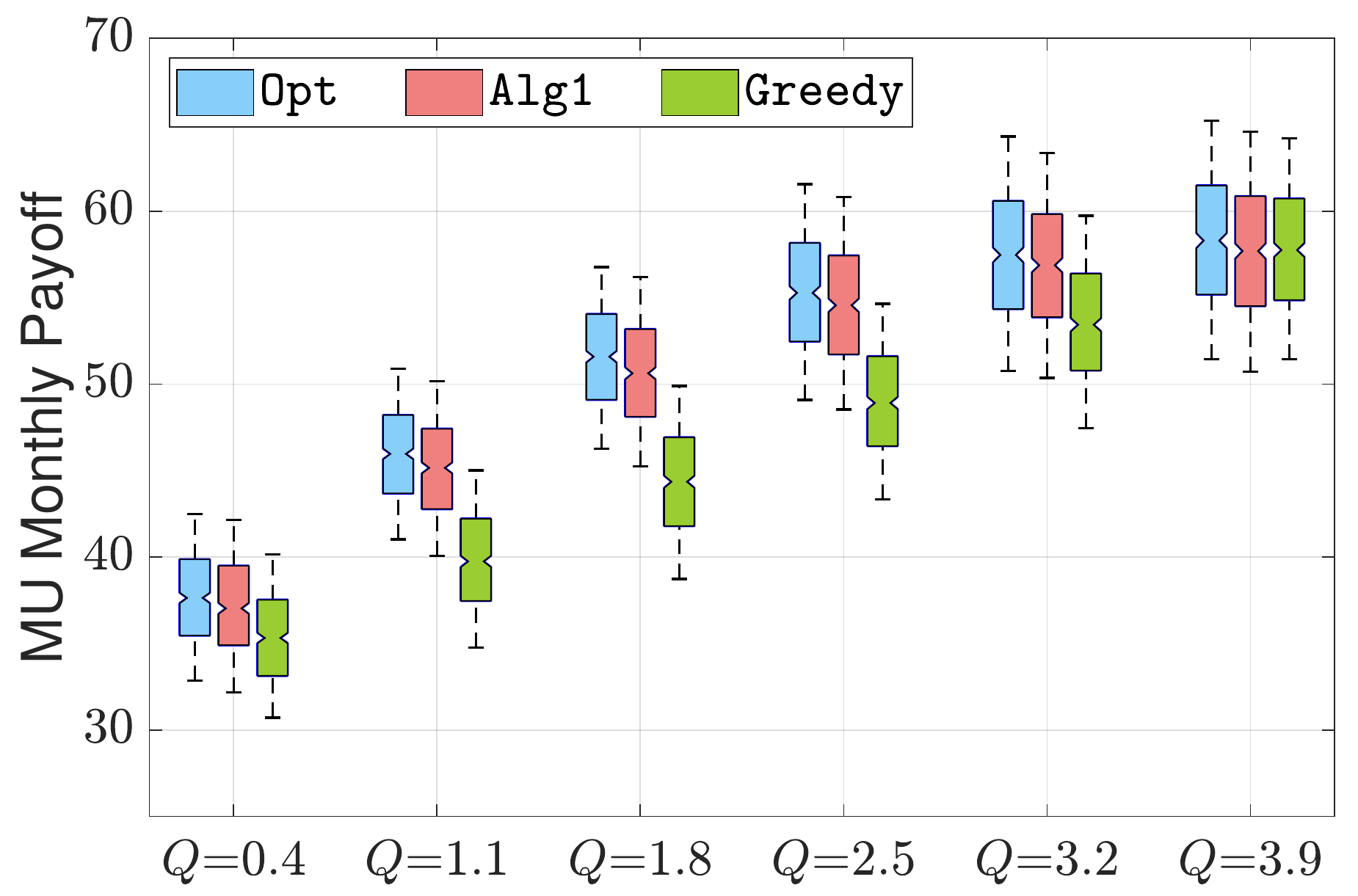}
	\caption{MU monthly payoff}
	\label{fig: Appendix_MU_Payoff_Q}
\end{figure}

Given the parameters $\alpha=0.6$ and $\beta=2$, we will compare the MU's monthly payoff in the following three cases:
\begin{itemize}
	\item The case of $\mathtt{Opt}$ corresponds to the off-line optimal outcome discussed in Theorem \ref{Theorem: MU KKT}.
	
	\item The case of $\mathtt{Alg1}$ corresponds to the proposed online strategy $\mathcal{A}$ defined in Algorithm \ref{Algorithm: MU}.
	
	\item The case of $\mathtt{Greedy}$ corresponds to the greedy strategy that tends to maximize the daily payoff without taking into account the potential future over usage.
\end{itemize}

Fig. \ref{fig: Appendix_MU_Payoff_Q} plots the MU's monthly payoff under different monthly data caps.
Comparing Fig. \ref{fig: Appendix_MU_Payoff_Q} with Fig. \ref{fig: User_Q}, we note that the MU's monthly payoff increases, but the relative performance among the three cases is similar.
This also verifies the theoretical result for the MU's online strategy.

\begin{figure}[h]
	\setlength{\abovecaptionskip}{0pt}
	\setlength{\belowcaptionskip}{0pt}
	\centering
	\subfigure[Average ESP revenue.]
	{\label{fig: Appendix_ESPrevenues_indata}\includegraphics[width=0.48\linewidth]{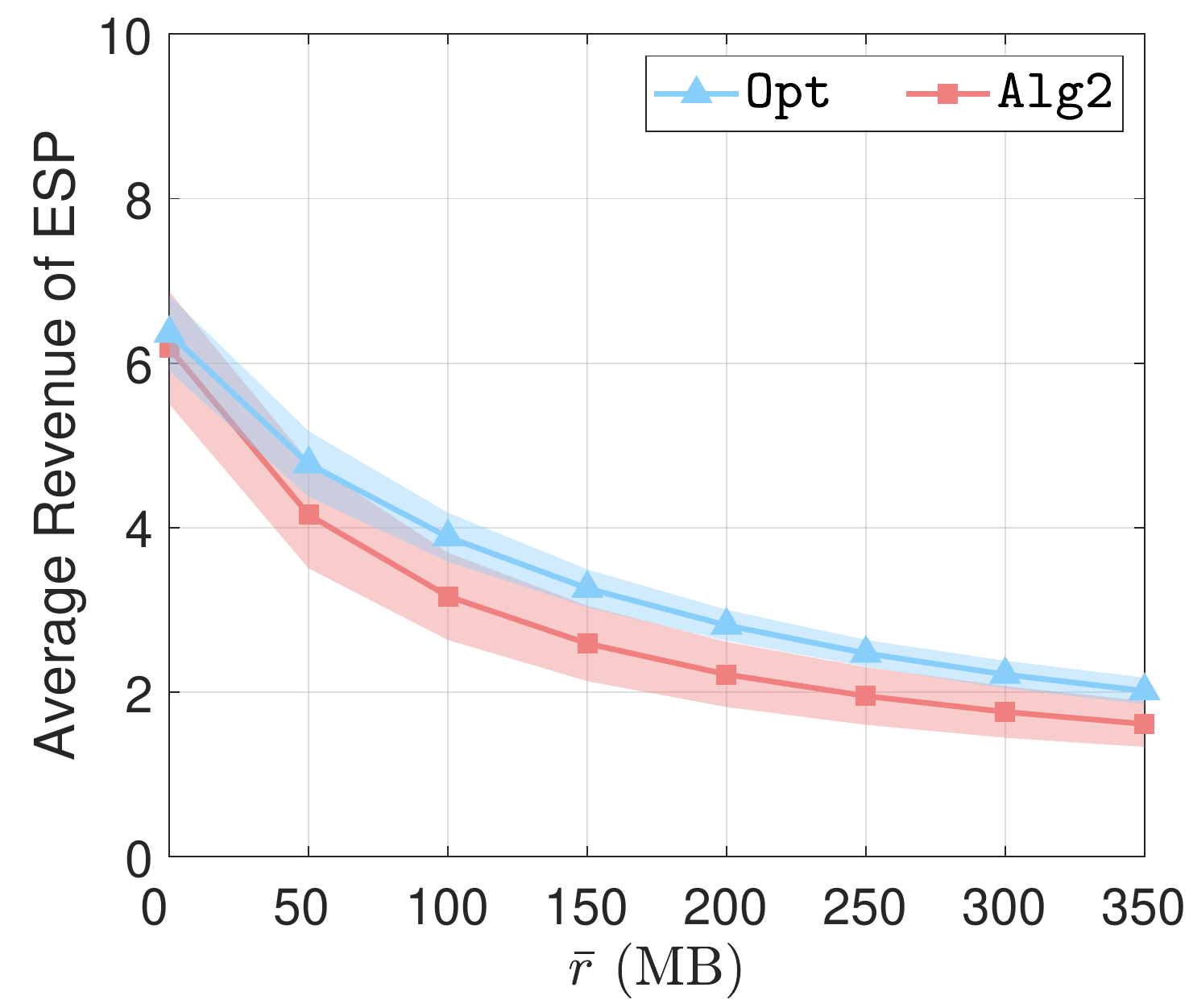}}\
	\subfigure[Ratio of the optimal revneue.]
	{\label{fig: Appendix_ESPfraction_indata}\includegraphics[width=0.48\linewidth]{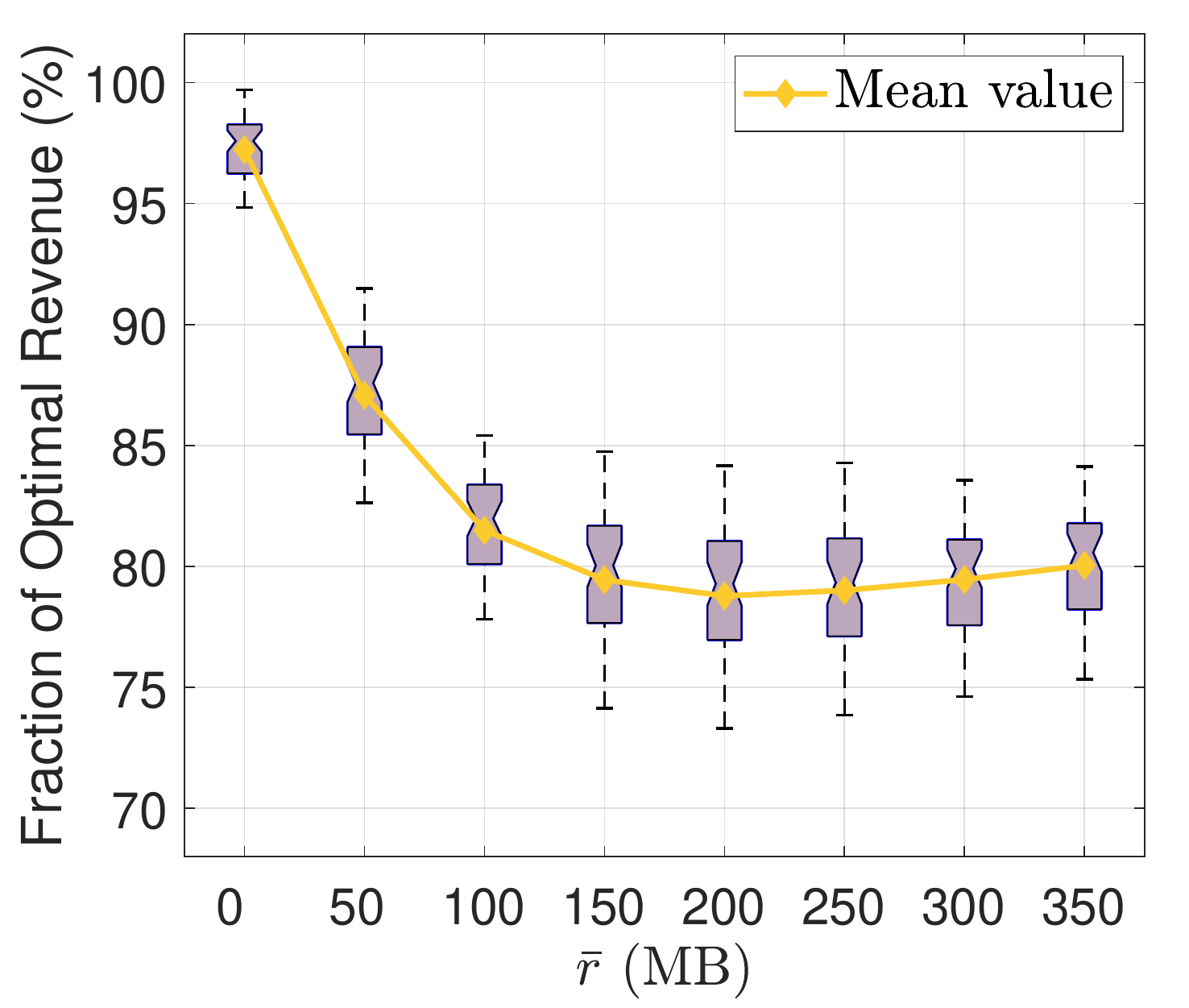}}
	\caption{ESP's average revenue versus maximal raw data $\bar{\indata}$.}
	\label{fig: Appendix_indata}
\end{figure}

Fig. \ref{fig: Appendix_indata} plots the performance of the pricing policy $\mathcal{P}$ under different values of the maximal raw data amount $\bar{\indata}$.
We compare ESP's revenue under the pricing policy $\mathcal{P}$ (labeled by $\mathtt{Alg2}$) with the offline optimal revenue (labeled by $\mathtt{Opt}$).
Comparing Fig. \ref{fig: Appendix_ESPrevenues_indata} with Fig \ref{fig: ESPrevenues_indata}, we note that the ESP's revenues of the two cases decrease.
However, comparing Fig. \ref{fig: Appendix_ESPfraction_indata} to Fig. \ref{fig: ESPfraction_indata}, we note that the relative performance of the pricing policy is robust.

\end{document}